\newif\ifarxiv
\arxivtrue
\newif\ifblind
\newif\ifspiderweb
\spiderwebtrue

\documentclass[a4paper,UKenglish,cleveref, autoref,pagebackref, thm-restate,authorcolumns,numberwithinsect]{lipics-v2019}

\ifarxiv{}
  \nolinenumbers
  \hideLIPIcs
\fi{}
\usepackage[T1]{fontenc}

\usepackage{xcolor}
\usepackage{amssymb,amsmath,amsthm}

\usepackage{graphicx}
\usepackage{booktabs}
\usepackage{tabularx}
\usepackage{multicol}
\usepackage[numbers,sort&compress]{natbib}
\usepackage{csquotes}

\usepackage{tikz}
\usetikzlibrary{decorations.pathreplacing}

\usepackage{paralist}
\usepackage{mathtools}
\usepackage[text size=footnotesize,disable]{todonotes} %
\usepackage{xargs}
\newcommandx{\set}[2][1=1]{\ensuremath{\{#1,\ldots,#2\}}}
\newcommandx{\tlog}[3][1=,3=]{\log_{#1}^{#3}(#2)}
\newcommandx{\dist}[2][1=]{\ensuremath{\operatorname{dist}_{#1}(#2)}}
\newcommandx{\ith}[2][1=th]{#2\nobreakdash-#1}

\newtheorem{observation}{Observation}

\theoremstyle{definition}

\newtheorem{rrule}{Reduction Rule}
\newtheorem{construction}{Construction}
\newtheorem{algorithm}{Algorithm}
\newcommand{\cqed}{\hfill$\diamond$}
\newenvironment{algorithm*}
{\pushQED{\cqed{}}\algorithm}
 {\popQED\endalgorithm}
\newenvironment{construction*}
{\pushQED{\cqed{}}\construction}
 {\popQED\endconstruction}

\crefname{observation}{Observation}{Observations}
\crefname{rrule}{Reduction Rule}{Reduction Rules}
\crefname{construction}{Construction}{Constructions}
\crefname{theorem}{Theorem}{Theorems}
\Crefname{theorem}{Thm.}{Thms.}
\crefname{corollary}{Corollary}{Corollaries}
\crefname{lemma}{Lemma}{Lemmata}
\Crefname{corollary}{Cor.}{Cors.}
\crefname{proposition}{Proposition}{Propositions}
\Crefname{proposition}{Prop.}{Props.}

\newcommand{\yes}{\texttt{yes}}
\newcommand{\no}{\texttt{no}}
\newcommand{\RD}{$(\Rightarrow)\quad$}
\newcommand{\LD}{$(\Leftarrow)\quad$}

\newcommand{\problemdef}[3]{
  \begin{center}
  \begin{minipage}{0.99\textwidth}
    \noindent
    \textsc{#1}
    \begin{compactdesc}
    \item[Input:] #2
    \item[Question:] #3
    \end{compactdesc}
  \end{minipage}
  \end{center}
}

\newcommand{\N}{\mathbb{N}}

\renewcommand{\O}{\mathcal{O}}

\newcommand{\prob}[1]{\textsc{#1}}
\newcommand{\msp}{\prob{MstP}}
\newcommand{\vmsp}{\prob{V$\triangle$V-MstP}}
\newcommand{\emsp}{\prob{E$\triangle$E-MstP}}
\newcommand{\vimsp}{\prob{V$\cap$V-MstP}}
\newcommand{\eimsp}{\prob{E$\cap$E-MstP}}

\newcommand{\MSP}{\prob{Multistage $s$-$t$~Path}}

\newcommand{\TG}{\mathcal{G}}

\newcommand{\ltime}{\tau}
\newcommand{\TGfull}{\TG=(V, E_1, E_2, \ldots, E_\ltime)}
\newcommand{\TGcompact}{\TG=(V, (E_i)_{i=1}^\ltime)}
\newcommand{\ug}[1]{\ensuremath{#1_{\downarrow}}}

\newcommand{\cocl}[1]{\ensuremath{\operatorname{#1}}}
\newcommand{\W}[1]{\cocl{W[#1]}}

\newcommand{\NP}{\cocl{NP}}
\newcommand{\FPT}{\cocl{FPT}}
\newcommand{\fpt}{fixed-parameter tractable}
\newcommand{\coNP}{\cocl{coNP}}
\newcommand{\XP}{\cocl{XP}}
\newcommand{\poly}{\cocl{poly}}

\newcommand{\pNPh}{p-\NP-h.}
\newcommand{\NPincoNPslashpoly}{\ensuremath{\NP\subseteq \coNP/\poly}}
\newcommand{\NPnotincoNPslashpoly}{\ensuremath{\NP\not\subseteq \coNP/\poly}}

\newcommand{\per}{polynomial equivalence relation}

\newcommand{\calC}{\mathcal{C}}
\newcommand{\calP}{\mathcal{P}}
\newcommand{\calR}{\mathcal{R}}
\newcommand{\symdif}[2]{\ensuremath{#1\triangle#2}}

\newcommand{\alert}[1]{\emph{#1}}%
\newcommand{\tref}[1]{{\scriptsize(\Cref{#1})}}
\newcommand{\vmod}[1]{\ensuremath{\widehat{#1}}}
\newcommand{\etal}{et~al.}
\newcommand{\ceq}{\ensuremath{\coloneqq}}

\newcommand{\appsymb}{$\bigstar$}
\newcommand{\appref}[1]{\hyperref[proof:#1]{\appsymb}}

\newcommand{\appendixsection}[1]{%
  \ifarxiv{}\else{}
	\gappto{\appendixProofText}{\section{Additional Material for Section~\ref{#1}}\label{app:#1}} %
	\fi{}
}

\newcommand{\toappendix}[1]{%
  \ifarxiv{}#1\else{}
  \gappto{\appendixProofText}
  {{
    #1
  }}
  \fi{}
}

\newcommand{\appendixproof}[2]{%
  \ifarxiv{}#2\else{}\gappto{\appendixProofText}
  {
    \subsection{Proof of \autoref{#1}}\label{proof:#1}
    #2
  }
  \fi{}
}

\ifarxiv{}
\newcommand{\thetitle}{Multistage $s$-$t$ Path: Confronting Similarity with Dissimilarity}
\else{}
\newcommand{\thetitle}{Multistage $s$-$t$ Path: Confronting Similarity with Dissimilarity in Temporal Graphs}
\fi{}
\title{\thetitle} %

\titlerunning{Multistage $s$-$t$ Path} %

\ifblind{}
 \author{$ $}{$ $}{}{}{}
 \authorrunning{Multistage $s$-$t$ Path}
 \Copyright{\thetitle}
\else{}
  \author{Till Fluschnik}
  {Technische Universit\"at Berlin, Algorithmics and Computational Complexity, Berlin, Germany}
  {till.fluschnik@tu-berlin.de}
  {https://orcid.org/0000-0003-2203-4386}
  {Supported by DFG, project TORE, NI~369/18.}%
  \author{Rolf Niedermeier}
  {Technische Universit\"at Berlin, Algorithmics and Computational Complexity, Berlin, Germany}
  {rolf.niedermeier@tu-berlin.de}
  {https://orcid.org/0000-0003-1703-1236}
  {}%
  \author{Carsten Schubert}
  {Technische Universit\"at Berlin, Algorithmics and Computational Complexity, Berlin, Germany}
  {carsten.gm.schubert@campus.tu-berlin.de}
  {}
  {}%
  \author{Philipp Zschoche}
  {Technische Universit\"at Berlin, Algorithmics and Computational Complexity, Berlin, Germany}
  {zschoche@tu-berlin.de}
  {https://orcid.org/0000-0001-9846-0600}
  {}

  \authorrunning{T.~Fluschnik, R.~Niedermeier, C.~Schubert, P.~Zschoche} %

  \Copyright{T.~Fluschnik, R.~Niedermeier, C.~Schubert, P.~Zschoche} %
\fi{}



\category{} %

\relatedversion{} %

\supplement{}%

\EventEditors{John Q. Open and Joan R. Access}
\EventNoEds{2}
\EventLongTitle{42nd Conference on Very Important Topics (CVIT 2016)}
\EventShortTitle{CVIT 2016}
\EventAcronym{CVIT}
\EventYear{2016}
\EventDate{December 24--27, 2016}
\EventLocation{Little Whinging, United Kingdom}
\EventLogo{}
\SeriesVolume{42}
\ArticleNo{23}

\definecolor{lilla}{HTML}{750787}

\begin{document}
\maketitle

\begin{abstract}
Addressing a quest by~Gupta~\etal~[ICALP'14],
we provide a first, comprehensive study of 
finding a short $s$-$t$ path 
in the multistage graph model,
referred to as the \textsc{Multistage $s$-$t$ Path} problem.
Herein, 
given a sequence of graphs over the same vertex set but changing edge sets,
the task is to find short~$s$-$t$ paths in each graph (``snapshot'') such 
that in the found path sequence the
consecutive $s$-$t$ paths are ``similar''.
We measure similarity by the size of the symmetric difference of either the vertex set (vertex-similarity) or the edge set (edge-similarity) of any two consecutive paths.
We prove that these two variants of \textsc{Multistage $s$-$t$ Path} are already \NP-hard for an input sequence of only two graphs and maximum vertex 
degree four.
Motivated by this fact and natural applications of this scenario 
e.g.\ in traffic route planning,
we perform a parameterized complexity analysis. 
Among other results, for both variants,
vertex- and edge-similarity, 
we prove parameterized hardness (\W{1}-hardness) regarding the parameter path length (solution size) for both variants,
vertex- and edge-similarity.
As a further conceptual study,
we then modify the multistage model by asking for \emph{dissimilar} 
consecutive paths.
	As one of the main technical results (employing so-called representative sets known from non-temporal settings), we prove that dissimilarity allows for fixed-parameter tractability for the parameter solution size, contrasting our W[1]-hardness proof of the corresponding
similarity case. We also provide partially positive results concerning efficient
	and effective data reduction (kernelization).
\end{abstract}
\section{Introduction}
Finding short paths is perhaps the most fundamental task in algorithmic graph
theory and network analysis. There are numerous applications, including 
operations research, robotics, social network analysis, traffic and transportation, and VLSI design.
More specifically, we are concerned with
finding a short path connecting two designated vertices~$s$ and~$t$.
It is fair to say that for \emph{static} graphs the algorithmics 
(also from a practical side) of finding short(est) 
paths is very well understood.
This is much less so when considering path finding in \emph{temporal}
graphs, that is, 
graphs whose edge sets change over
time\footnote{Holme and Saram\"{a}ki~\cite{HS13,HS19} 
and~Michail~\cite{Michail16} survey algorithmic aspects 
of temporal graphs.}, a framework that in recent years received more 
and more attention in the field of network
science. For instance, models concerned with disease spreading or traffic 
routing typically are more realistic when taking into account that links 
between network nodes change over time.
In this work, we study path finding in temporal graphs with the 
additional (``multistage'') assumption that $s$-$t$-paths for consecutive
snapshots of the temporal graph shall be sufficiently ``similar''. 
We confront this 
with the opposite view that $s$-$t$-paths for consecutive
snapshots of the temporal graph shall be significantly ``dissimilar''. 
Herein, similarity can naturally be measured both by comparing the edge sets 
of the $s$-$t$~paths or the vertex sets of the $s$-$t$~paths.
Altogether, we end up with four natural problem variants.

A few words on motivation.
Both scenarios address different aspects of robustness in an
environment changing over time.
Let us first look at the dissimilarity scenario.
Here one may think of a situation where because of necessary recovery
or cleansing costs (in pandemic times one may think
of disinfection measures) one wants to avoid that subsequent
``agents'' on the way from start to goal share too many parts of their
routing paths. Moreover, one may also think of applications in the
context of so-called VIP~routing, which address 
security aspects~\cite{FKNS19,FMS19}.
As to the similarity scenario, one may think of robustness in the
sense of ``path maintenance'': every deviation from the path used
before causes additional costs (set up, preparation, checking) and
thus shall be kept at a minimum. This can be interpreted in the
spirit of incremental changes (evolutionary rather than radical
changes)~\cite{CCFM04,HN13}.

Formally,
a temporal graph~$\TGfull$ consists of a set~$V$ of 
vertices and lifetime~$\tau$ many edge sets~$E_1,E_2,\dots,E_\tau$ over~$V$.
Finding an~$s$-$t$ path over time,
also known as temporal~$s$-$t$ path,
has already been studied~\cite{WuCKHHW16,HimmelBNN19}. There, however, 
a path may use edges from $\bigcup_{i=1}^{\tau} E_i$, while in our setting 
we search for path sequences consisting of $\tau$~paths, one for each~$E_i$.
With focusing on similar \emph{and} dissimilar paths here, however, 
we introduce a new view on finding paths in temporal graphs.
More specifically, addressing a quest of Gupta~\etal~\cite{GuptaTW14}, 
one of the first studies on multistage 
problems,
this paper initiates a study of finding short~$s$-$t$ paths 
in the~\emph{multistage} model,
that is, 
finding a short~$s$-$t$ path in each \emph{snapshot}~$(V,E_i)$ of the temporal graph~$\TG$ such that consecutive $s$-$t$ paths do not differ too much; 
formally, we have the following (where $\Pi$ refers to a requested property of a solution path):

\problemdef{$\Pi$ \MSP{} ($\Pi$-\msp)}
{A temporal graph~$\TGfull$, two distinct vertices~$s,t\in V$, and two integers~$k,\ell\in \N_0$.
}
{Is there a sequence~$(P_1, P_2, \ldots,P_\ltime)$ such that~$P_i$ is an~$s$-$t$~path in~$(V,E_i)$ 
	with~$|V(P_i)|\leq k$ for all~$i\in\set{\ltime}$, 
	and~$\dist[\Pi]{P_i,P_{i+1}}\leq \ell$ for all~$i \in \set{\ltime-1}$?
}

The multistage model requests snapshot solutions such that 
(in time) consecutive ones are similar to each other.
Herein,
similarity is measured by the symmetric difference of the sets 
describing the consecutive snapshot solutions.
For paths,
there are two natural choices for comparing: the sets of vertices and 
the sets of edges. 
Thus, we obtain two distance measures defined as follows.
\begin{align*}
 \dist[\text{V$\triangle$V}]{P_i,P_{i+1}}&\ceq |\symdif{V(P_i)}{V(P_{i+1})}| && \text{(\vmsp{}),} \\
 \dist[\text{E$\triangle$E}]{P_i,P_{i+1}}&\ceq |\symdif{E(P_i)}{E(P_{i+1})}| && \text{(\emsp{}).} 
\end{align*}
Confronting the similarity request of the multistage framework 
with a dissimilarity request instead leads to the following.
\begin{align*}
 \dist[\text{V$\cap$V}]{P_i,P_{i+1}}&\ceq |(V(P_i)\cap V(P_{i+1}))\setminus\{s,t\}| && \text{(\vimsp{}),} \\
 \dist[\text{E$\cap$E}]{P_i,P_{i+1}}&\ceq |E(P_i)\cap E(P_{i+1})| && \text{(\eimsp{}).}
\end{align*}
Note that we can easily compute each of the four distance measures in linear time. 

In the following,
we study the classical and parameterized complexity of 
all four variants~\emsp{}, \vmsp{}, \vimsp{}, and \eimsp{}.
With performing a parameterized complexity analysis,
we do not only aim for a better understanding of the influence of 
several natural problem parameters like path length~$k-1$ 
or the upper bound~$\ell$ on the distance values between consecutive snapshots,
but we also want to find out where (and why) the problem variants 
are potentially different from each other; in particular, this means 
confronting the similarity (a.k.a.\ as classical multistage) view with 
the dissimilarity view.

\subparagraph*{Our contributions.}
We introduce four natural variants of the \textsc{Multistage $s$-$t$ Path} 
problem by employing four different ways to measure the distance between 
consecutive solutions.
Doing so, seemingly for the first time for multistage models in general,
we provide a seemingly first systematic study on the impact on the algorithmic complexity when 
switching between edge and vertex distances on the one hand,
and similarity versus dissimilarity distance measurements on the other hand. 

We prove all four problems to be~\NP-complete, 
even in the restricted case of only two snapshots, 
each snapshot being series-parallel and the underlying graph being 
of maximum degree four.
We provide an extensive study on the parameterized complexity landscape 
of the problems
regarding the parameters $k$~(path length), $\ell$~(maximum path distance between consecutive snapshots), $\tau$~(lifetime), $n$~(number of graph vertices),
$\ug{\nu}$~(vertex cover number of the ``underlying graph''), 
and~$\ug{\Delta}$~(maximum vertex degree in the underlying graph);
see~\cref{tab:results} for an overview.
\ifspiderweb{}%
  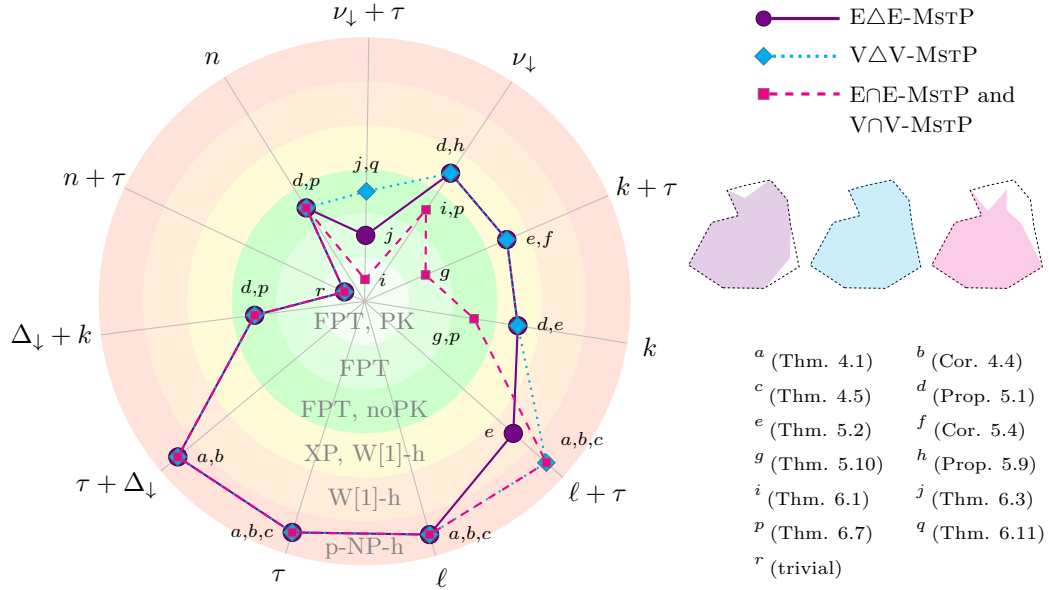
\begin{figure}
    \centering
    \begin{tikzpicture}
    \usetikzlibrary{shapes}
    \def\xr{1}
    \def\yr{1}

    \def\swRLA{\ensuremath{a}};
    \def\swRLB{\ensuremath{b}};
    \def\swRLC{\ensuremath{c}};
    \def\swRLD{\ensuremath{d}};
    \def\swRLE{\ensuremath{e}};
    \def\swRLF{\ensuremath{f}};
    \def\swRLG{\ensuremath{g}};
    \def\swRLH{\ensuremath{h}};
    \def\swRLI{\ensuremath{i}};
    \def\swRLJ{\ensuremath{j}};
    \def\swRLK{\ensuremath{p}};
    \def\swRLL{\ensuremath{q}};
    \def\swRLM{\ensuremath{r}};
    
    \def\swColA{lilla}
    \def\swColB{cyan}
    \def\swColC{magenta} 
    
    \begin{scope}[xshift=\xr*0cm,yshift=\yr*0cm]
          \newcommand{\swD}{11} %
          \newcommand{\swC}{6} %
          \newcommand{\swU}{6} %
          \newcommand{\swFs}{\small\color{darkgray!70!white}} %
          \newcommand{\swFsL}{\normalsize} %

          \newdimen\swR %
          \swR=3.5cm 
          \newdimen\swL %
          \swL=3.8cm

          \newcommand{\swA}{360/\swD} %

		  \begin{scope}[rotate=220]
          \draw[fill=orange!50!red!15!white,draw=none] circle [radius=\swR*6/\swC];
          \draw[fill=orange!15!white,draw=none] circle [radius=\swR*5/\swC];
          \draw[fill=yellow!20!white,draw=none] circle [radius=\swR*4/\swC];
          \draw[fill=green!20!white,draw=none] circle [radius=\swR*3/\swC];
          \draw[fill=green!12!white,draw=none] circle [radius=\swR*2/\swC];
          \draw[fill=green!5!white,draw=none] circle [radius=\swR*1/\swC];
        
          \path (0:0cm) coordinate (O); %
            \foreach \swX in {1,...,\swD}{
              \draw[thin,color=lightgray] (\swX*\swA:0) -- (\swX*\swA:\swR);
            }

            \foreach \swY in {1,...,\swU}{
              \foreach \swX in {1,...,\swD}{
                \path (\swX*\swA:\swY*\swR/\swU-0.5*\swR/\swU) coordinate (D\swX-\swY);
              }
            }
            \foreach \x/\y in {
            1/$\tau$,
            2/$\ell$,
            3/~~~~$\ell+\tau$,
            4/$k$,
            5/~~~~$k+\tau$,
            6/$\nu_{\downarrow}$,
            7/$\nu_{\downarrow}+\tau$,
            8/$n$,
            9/$n+\tau$~~~,
            10/$\Delta_{\downarrow}+k$~~~~~~~,
            11/$\tau+\Delta_{\downarrow}$~~~~~~
            }{\path (\x*\swA:\swL) node (L\x) {\swFsL \y};}
      \end{scope}

          \path (270:6*\swR/\swU-0.4*\swR/\swU) node[] {\swFs p-NP-h};
          \path (270:5*\swR/\swU-0.5*\swR/\swU) node {\swFs W[1]-h};
          \path (270:4*\swR/\swU-0.5*\swR/\swU) node {\swFs XP, W[1]-h};
          \path (270:3*\swR/\swU-0.5*\swR/\swU) node {\swFs FPT, noPK};
          \path (270:2*\swR/\swU-0.5*\swR/\swU) node {\swFs FPT};
          \path (270:1*\swR/\swU-0.5*\swR/\swU) node {\swFs FPT, PK};
          
          \newcommand{\swPolygonL}[3]{
            \foreach \x/\y/\pos/\z in {#1}{\node (t\x) at (D\x-\y)[scale=0.5,draw,#2,label={\pos:\scriptsize\z}]{};}
            \foreach[evaluate=\x as \y using int(\x+1)] \x in {1,2,...,10}{\draw[#3] (t\x) to (t\y);}\draw[#3] (t\swD) to (t1);
          }
          \newcommand{\swPolygon}[3]{
              \foreach \x/\y in {#1}{\node (t\x) at (D\x-\y)[scale=0.5,draw,#2]{};}
                \foreach[evaluate=\x as \y using int(\x+1)] \x in {1,2,...,10}{\draw[#3] (t\x) to (t\y);}\draw[#3] (t\swD) to (t1);
          }

          \swPolygonL{
            1/6/left/{$\swRLA$,$\swRLB$,$\swRLC$},
            2/6/right/{$\swRLA$,$\swRLB$,$\swRLC$},
            3/5/left/{$\swRLE$}, %
            4/4/right/{$\swRLD$,$\swRLE$},
            5/4/right/{$\swRLE$,$\swRLF$},
            6/4/above/{$\swRLD$,$\swRLH$},
            7/2/right/{$\swRLJ$},
            8/3/above/{$\swRLD$,$\swRLK$},
            9/1/left/{$\swRLM$},
            10/3/above/{$\swRLD$,$\swRLK$},
            11/6/right/{$\swRLA$,$\swRLB$}
            }{fill=\swColA,draw=\swColA!50!black,circle,scale=1.45}{thick,\swColA}
          \swPolygonL{
            1/6/above/{},
            2/6/above/{},
            3/6/above right/{$\swRLA$,$\swRLB$,$\swRLC$},
            4/4/above/{},
            5/4/above/{},
            6/4/above/{},
            7/3/above/{$\swRLJ$,$\swRLL$},
            8/3/above/{},
            9/1/above/{},
            10/3/above/{},
            11/6/above/{}
            }{fill=\swColB,draw=\swColB!50!black,diamond,scale=1}{thick,dotted,\swColB}
          \swPolygonL{
            1/6/above/{},
            2/6/above/{},
            3/6/above/{},
            4/3/below left/{$\swRLG$,$\swRLK$},
            5/2/right/{$\swRLG$},
            6/3/right/{$\swRLI$,$\swRLK$},
            7/1/right/{$\swRLI$},
            8/3/above/{},
            9/1/above/{},
            10/3/above/{},
            11/6/above/{}
            }{fill=\swColC,draw=\swColC!50!black,scale=0.95}{thick,dashed,\swColC}
      \end{scope}

      \begin{scope}[transform canvas={xshift = \xr*5cm,yshift=\yr*1.1cm,scale=0.3}]
        \path[dashed,thick,draw=black] (D1-6) to (D2-6) to (D3-6) to (D4-4) to (D5-4) to (D6-4) to (D7-3) to (D8-3) to (D9-1) to (D10-3) to (D11-6) to cycle;
      \end{scope}
      
      \begin{scope}[transform canvas={xshift = \xr*6.6cm,yshift=\yr*1.1cm,scale=0.3}]
        \path[dashed,thick,draw=black] (D1-6) to (D2-6) to (D3-6) to (D4-4) to (D5-4) to (D6-4) to (D7-3) to (D8-3) to (D9-1) to (D10-3) to (D11-6) to cycle;
      \end{scope}
      
      \begin{scope}[transform canvas={xshift = \xr*8.2cm,yshift=\yr*1.1cm,scale=0.3}]
        \path[dashed,thick,draw=black] (D1-6) to (D2-6) to (D3-6) to (D4-4) to (D5-4) to (D6-4) to (D7-3) to (D8-3) to (D9-1) to (D10-3) to (D11-6) to cycle;
      \end{scope}

      \begin{scope}[transform canvas={xshift = \xr*5cm,yshift=\yr*1.1cm,scale=0.3}]
        \path[fill=\swColA,opacity=0.2,draw=\swColA] (D1-6) to (D2-6) to (D3-5) to (D4-4) to (D5-4) to (D6-4) to (D7-2) to (D8-3) to (D9-1) to (D10-3) to (D11-6) to cycle; 
      \end{scope}
      
      \begin{scope}[transform canvas={xshift = \xr*6.6cm,yshift=\yr*1.1cm,scale=0.3}]
        \path[fill=\swColB,opacity=0.2,draw=\swColB] (D1-6) to (D2-6) to (D3-6) to (D4-4) to (D5-4) to (D6-4) to (D7-3) to (D8-3) to (D9-1) to (D10-3) to (D11-6) to cycle;
      \end{scope}
      
      \begin{scope}[transform canvas={xshift = \xr*8.2cm,yshift=\yr*1.1cm,scale=0.3}]
        \path[fill=\swColC,opacity=0.2,draw=\swColC] (D1-6) to (D2-6) to (D3-6) to (D4-3) to (D5-2) to (D6-3) to (D7-1) to (D8-3) to (D9-1) to (D10-3) to (D11-6) to cycle;
      \end{scope}

      \begin{scope}[xshift=\xr*5.25cm,yshift=\yr*3.75cm,font=\footnotesize]
        \node (t1) at (0,0)[fill=\swColA, draw=\swColA!50!black,circle,scale=0.6]{};
        \node (t2) at (\xr*1,0)[anchor=west]{\emsp{}};
        \draw[very thick,\swColA] (t1) to (t2);
        \node (t1) at (0,-0.5*\yr)[fill=\swColB,draw=\swColB!50!black,diamond,scale=0.5]{};
        \node (t2) at (\xr*1,-0.5*\yr)[anchor=west]{\vmsp{}};
        \draw[very thick,dotted,\swColB] (t1) to (t2);
        \node (t1) at (0,-1*\yr)[fill=\swColC,draw=\swColC!50!black,scale=0.6]{};
        \node (t2) at (\xr*1,-1*\yr)[anchor=west,align=left]{\eimsp{} and};
        \node (t2x) at (\xr*1,-1.4*\yr)[anchor=west,align=left]{\vimsp{}};
        \draw[very thick,dashed,\swColC] (t1) to (t2);
      \end{scope}
      \begin{scope}[xshift=\xr*7.0cm,yshift=\yr*-2.15cm]
        \node[text width=3.75cm] at (0,0) 
        {
        \begin{tabular}{@{}ll@{}}
            $^\swRLA$\,\tref{thm:emspnphard} 	& $^\swRLB$\,\tref{cor:vertexmspnphardtau} \\
            $^\swRLC$\,\tref{thm:vmsppNPhelltau} & $^\swRLD$\,\tref{thm:allxp}  \\
            $^\swRLE$\,\tref{thm:w1hard}			& $^\swRLF$\,\tref{cor:vmspw1hardktau} \\
            $^\swRLG$\,\tref{thm:eimspfptk} & $^\swRLH$\,\tref{prop:vmspvcwhard} \\
            $^\swRLI$\,\tref{thm:interseckernel-vc} & $^\swRLJ$\,\tref{thm:symdifexpkernelnutau} \\
            $^\swRLK$\,\tref{thm:empvincompr} & $^\swRLL$\,\tref{thm:vmspnopkvctau} \\
            $^\swRLM$\,{\scriptsize(trivial)}
          \end{tabular}
        };
      \end{scope}
    \end{tikzpicture}
    \caption{Overview of our results. ``p-\NP-h.'', ``\W{1}-h.'', ``FPT'',
      ``PK'', and ``noPK'' respectively abbreviate para-\NP-hard,
      \W{1}-hard,
      \fpt{},
      polynomial kernel,
      and ``no polynomial kernel unless~\NPincoNPslashpoly''.
      Note that~$\ell\leq 2k$ and~$k\leq 2\ug{\nu}+1$.\\
	  }
    \label{tab:results}
    \end{figure}
\fi{}%
The results of our parameterized complexity analysis reveal
a clear distinction between similarity and dissimilarity.
When parameterized by the maximum number~$k$ of vertices in each $s$-$t$~path,
while~\emsp{} and~\vmsp{} are \W{1}-hard,  %
\eimsp{} and~\vimsp{} are \fpt{}. To this end, we develop one of the first
uses of the technique of representative sets~\cite{monien1985find,FLPS16} in the context of 
temporal graphs.
In addition, we show that, under standard complexity-theoretic 
assumptions, the similarity problem \vmsp{} parameterized 
by the number of vertices 
has no polynomial kernel, while the dissimilarity problem \vimsp{}
has one.

\subparagraph*{Related work.}
Our studies are within algorithmic temporal graph theory and, 
more specifically, contribute and extend a series of studies on the 
multistage model. Notably, all previous studies (on various basic computational
problems) within the multistage framework adhere to the ``similarity view'';
we extend this by introducing also a ``dissimilarity view''.

To the best of our knowledge, the \emph{multistage} model (which is a 
temporal model 
not necessarily only applying to graph problems) 
first appeared in~2014 in works of Eisenstat~\etal~\cite{EMS14} and Gupta~\etal~\cite{GuptaTW14}.
In a nutshell, the model
considers a sequence~$(I_1,\ldots,I_\tau)$ of instances 
of some problem~$P$ as input,
and it asks for a ``robust'' sequence of solutions to the instances in the sense that any two consecutive solutions are similar.
Several classical problems have been studied in the multistage model, 
both from an approximate~\cite{BEK19,BampisET19,BampisELP18,BampisEST19} 
and from a 
parameterized~\cite{FluschnikNRZ19,HeegerHKNRS19} algorithmics point of view.
While~\emsp{} and~\vmsp{} adhere to the original multistage model,
our two problems \eimsp{} and \vimsp{} can be seen as a novel and natural variation of the multistage model by replacing the goal of consecutive similarity with consecutive dissimilarity.

Several basic temporal graph problems are closely 
related to the task of finding a (short) temporal~$s$-$t$ path
(finding an~$s$-$t$ path over time,
that is,
an $s$-$t$~path where the edges along the path have non-decreasing 
time stamps)~\cite{HimmelBNN19,CasteigtsHMZ19,ErlebachS18,ErlebachKLSS19,EnrightMMZ19,EnrightM18,KempeKK02,WuCKHHW16,ZschocheFMN20}.
While these problems typically are concerned  with temporal $s$-$t$~paths that may span over several snapshots of the temporal graph,
in our multistage-inspired framework we aim for finding 
an~$s$-$t$ path in \emph{each} snapshot.

We mention in passing that there is also somewhat related work on 
short paths in multiplex networks (also known as multilayer or 
multimodal networks)~\cite{GSMJ17}. The main difference to 
our scenario is that the temporal aspect imposes an ordering 
of the layers whereas the multiplex view does not; in addition, 
Ghariblou et al.~\cite{GSMJ17} perform a multiobjective optimization,
being particularly interested in Pareto efficiency.

\ifarxiv{}\else{}
  Due to the lack of space, many details had to be moved to an appendix (marked by~\appsymb).
\fi{}

\section{Preliminaries}
\label{sec:prelim}
\appendixsection{sec:prelim}

We denote by~$\N$ and~$\N_0$ the natural numbers excluding and including~$0$, respectively.
By~$\log(\cdot)$ we denote the logarithm to base two.
We use basic notation from graph theory %
and parameterized algorithmics%
\ifarxiv{}.\else{},
see \cref{app:sec:prelim} for details.\fi{}

\smallskip\noindent
\textbf{Graph theory.}
An undirected graph~$G=(V,E)$ is a tuple consisting of a set~$V$ of vertices and a set~$E\subseteq \{\{v,w\}\mid v,w\in V,v\neq w\}$ of edges.
For a graph~$G$, we also denote by~$V(G)$ and~$E(G)$ the vertex and edge set of~$G$,
respectively.
For a vertex set~$W\subseteq V$,
the induced  subgraph~$G[W]$ is defined as the graph~$(W,\{\{v,w\}\in E\mid v,w\in W\})$.
A path~$P=(V,E)$ is a graph with a set~$V=\{v_1,\dots,v_k\}$ of distinct vertices and edge set~$E=\{\{v_i,v_{i+1}\}\mid 1\leq i<k\}$
(we often represent path~$P$ by the tuple~$(v_1,v_2,\dots,v_k)$);
we say that~$P$ is a~$v_1$-$v_k$ path.
The length of a path is its number of edges.
For two vertices~$s,t\in V(G)$,
an~$s$-$t$ separator~$S\subseteq V(G)\setminus\{s,t\}$ is a set of vertices such that there is no~$s$-$t$ path in~$G-S$,
where~$G-S=G[V\setminus S]$.
We denote by~$N_G(v)=\{w\in V\mid \{w,v\}\in E\}$ the neighborhood of a vertex~$v$ in~$G$,
and by~$\deg(v)=|N_G(v)|$ the degree of~$v$ in~$G$.
Moreover,
we denote by~$\Delta$ (or~$\Delta(G)$) the maximum vertex-degree of~$G$,
that is,
$\Delta(G)=\max_{v\in V}\deg(v)$.
A \emph{vertex cover} of~$G$ is a set~$W$ of vertices such that~$G-W$ contains no edge;
we denote by~$\nu$ (or~$\nu(G)$) the smallest size of a vertex cover in~$G$.
A graph with distinct terminal vertices~$s,t$ is series-parallel if it can be turned into a single edge by a sequence of contractions of degree-two vertices except~$s$ and~$t$ while removing any parallel edge that appears~\cite{Duffin65}.

\smallskip\noindent
\textbf{Temporal graph theory.}
A temporal graph~$\TGfull$ consists of a set~$V$ of vertices and lifetime~$\tau$ many edge sets~$E_1,E_2,\dots,E_\tau$ over~$V$.
We also denote by~$\tau(\TG)$ the lifetime of~$\TG$.
The size of $\TG$ is $|\TG| \ceq  |V| + \sum_{i=1}^\tau |E_i|$.
The static graph~$(V,E_i)$ is called the~\ith{$i$} snapshot.
The \emph{underlying graph}~$\ug{\TG}$ of~$\TG$ is the static graph~$(V,E_1\cup\dots\cup E_\tau)$.
The underlying vertex cover number~$\ug{\nu}$ is~$\nu(\ug{\TG})$.
The underlying maximum degree~$\ug{\Delta}$ is~$\Delta(\ug{\TG})$.

\toappendix{
\subparagraph*{Parameterized complexity.}
Let~$\Sigma$ denote a finite alphabet.
A parameterized problem~$L\subseteq \{(x,k)\in \Sigma^*\times \N_0\}$ is a subset of all instances~$(x,k)$ from~$\Sigma^*\times \N_0$,
where~$k$ denotes the parameter.
A parameterized problem~$L$ is 
\begin{inparaenum}[(i)]
 \item \fpt{} if there is an algorithm that decides every instance~$(x,k)$ for~$L$ in~$f(k)\cdot |x|^{O(1)}$ time,
 \item contained in the class~\XP{} if there is an algorithm that decides every instance~$(x,k)$ for~$L$ in~$|x|^{f(k)}$ time, and
 \item para-\NP-hard if the problem for some constant value of the parameter is~\NP-hard,
\end{inparaenum}
where~$f$ is some computable function only depending on the parameter.
For two parameterized problems~$L,L'$,
an instance~$(x,k)\in \Sigma^*\times\N_0$ of~$L$ is equivalent to an instance~$(x',k')\in \Sigma^*\times\N_0$ for~$L'$ if~$(x,k)\in L\iff (x',k')\in L'$.
A problem~$L$ is hard for the class~\W{1} (\W{1}-hard) if for every problem~$L'\in \W{1}$ there is an algorithm that maps any instance~$(x,k)$
in~$f(k)\cdot |x|^{O(1)}$ time to an equivalent instance~$(x',k')$ with~$k'=g(k)$ for some computable functions~$f,g$.
It holds true that~$\FPT\subseteq \W{1}\subseteq \XP$,
where~\FPT{} denotes the class of all \fpt{} parameterized problems.
It is believed that~$\FPT\neq \W{1}$,
and that hence no~\W{1}-hard problem is~\fpt{}.
A problem kernelization for a parameterized problem~$L$ is a polynomial-time algorithm that maps any instance~$(x,k)$ of~$L$ to an equivalent instance~$(x',k')$ of~$L$ (the kernel)
such that~$|x'|+k\leq f(k)$ for some computable function~$f$;
If~$f$ is a polynomial, we say that the problem kernelization (and kernel) is polynomial.
It is well-known that a decidable parameterized problem is \fpt{} if and only if it admits a problem kernelization. %
}

\section{Relation between distance measures: from edges to vertices}
\label{sec:relation}
\appendixsection{sec:relation}

We show that there are polynomial-time algorithms that,
given an instance of~\emsp{} or of~\eimsp{},
construct an equivalent instance of the respective vertex-counterpart.%

\ifarxiv{}
\begin{proposition}
\else{}
\begin{proposition}[\appref{prop:edgetovertex}]
\fi{}
 \label{prop:edgetovertex}
 There is an algorithm that,
 on every input~$(\TG,s,t,k,\ell)$ to~\emsp,
 computes in~$\O(|\TG|\cdot\ell)$ time an equivalent instance~$(\TG',s,t,k',\ell')$ of~\vmsp{}
 such that~$k'\in O(k\cdot \ell)$, 
 $\ell'\in O(\ell^2)$,
 $\Delta(\ug{\TG})=\Delta(\ug{\TG'})$,
 and~$\tau(\TG)=\tau(\TG')$.
\end{proposition}

\appendixproof{prop:edgetovertex}
{
    \begin{proof}
    Let~$I=(\TG=(V,E_1,\ldots,E_\ltime),s,t,k,\ell)$ be an instance of~\emsp.
    Let initially~$V'=V$.
    For each edge~$e\in E\ceq E_1\cup\dots\cup E_\ltime$, 
    add the set~$V_e=\{v_e^1,\dots,v_e^{\ell+1}\}$ of~$\ell+1$ vertices to~$V'$.
    For each~$i\in\set{\ltime}$,
    set~$E_i'$ to~$\bigcup_{e\in E_i} P_e$, 
    where~$P_e=\{\{a,v_e^1\},\{v_e^{\ell+1},b\}\}\cup \bigcup_{1\leq j\leq \ell} \{\{v_e^j,v_e^{j+1}\}\}$.
    This finishes the construction of~$\TG'=(V',E_1',\dots,E_\ltime')$.
    Finally,
    set~$k'=k+(k-1)(\ell+1)$ and~$\ell'=(\ell+1)^2-1$.
    We claim that~$I$ is a \yes-instance if and only if~$I'\ceq (\TG',s,t,k',\ell')$ is a \yes-instance.
    
    \RD{}
    Let~$\calP=(P_1,\ldots,P_\ltime)$ be a solution to~$I$.
    For each~$i\in\set{\ltime}$,
    construct~$P_i'$ with~$V(P_i')=V(P_i)\cup \{V_e\mid e\in E(P_i)\}$ and~$E(P_i')=\{P_e\mid e\in E(P_i)\}$.
    Clearly~$P_i'$ is an~$s$-$t$ path in~$(V',E_i')$.
    Moreover,~$|V(P_i')|=|V(P_i)|+|\{V_e\mid e\in E(P_i)\}| \leq k+(k-1)\cdot (\ell+1)=k'$
    and~$|\symdif{V(P_i')}{V(P_{i+1}')}|\leq \ell+(\ell+1)\cdot|\symdif{E(P_i)}{E(P_{i+1})}|\leq \ell+(\ell+1)\ell = \ell'$.
    
    \LD{}
    Let~$\calP'=(P_1',\ldots,P_\ltime')$ be a solution to~$I'$.
    For each~$i\in\set{\ltime}$,
    construct~$P_i$ with~$V(P_i)=V(P_i')\setminus \{V_e\mid P_e\subseteq E(P_i')\}$ and~$E(P_i)=\{e\mid P_e\subseteq E(P_i')\}$.
    Clearly~$P_i$ is an~$s$-$t$ path in~$(V,E_i)$.
    Moreover,
    note that~$k^*=|V(P_i')\cap V|\leq k$,
    since otherwise we have more than~$(k+1)+k^*(\ell+1)$ vertices in~$P_i'$,
    contradicting~$\calP'$ to be a solution.
    Hence,
    we have that~$|V(P_i)|=|V(P_i')\cap V|\leq k$.
    Further
    note that~$|\{e\in E\mid V_e\subseteq \symdif{V(P_i')}{V(P_{i+1}')}\}|\leq \ell$,
    since otherwise~$|\symdif{V(P_i')}{V(P_{i+1}')}|\geq (\ell+1)\cdot(\ell+1)>\ell'$.
    Hence,
    $|\symdif{E(P_i)}{E(P_{i+1})}|=|\{e\in E\mid V_e\subseteq \symdif{V(P_i')}{V(P_{i+1}')}\}|\leq \ell$.
    \end{proof}
}

\noindent

\ifarxiv{}
\begin{proposition}
		\else{}
\begin{proposition}[\appref{prop:edgetovertexintersect}]
\fi{}
 \label{prop:edgetovertexintersect}
 There is an algorithm that,
 on every input~$(\TG,s,t,k,\ell)$ to~\eimsp,
 computes in $\O(|\TG|)$ time an equivalent instance~$(\TG',s,t,k',\ell')$ of~\vimsp{}
 such that~$k'=2k-1$, 
 $\ell'=\ell$,
 $\Delta(\ug{\TG})=\max\{\Delta(\ug{\TG'}),4\}$, %
 and~$\tau(\TG)=\tau(\TG')$.
\end{proposition}

\appendixproof{prop:edgetovertexintersect}
{
    \begin{proof}
    Let $I=(\TG=(V,E_1,\ldots,E_\tau),s,t,k,\ell)$ be an instance of~\eimsp{},
    and denote by~$E=E_1\cup\dots\cup E_\tau$.
    Define~for each~$v\in V\setminus\{s,t\}$ the set~$V_v=V_v^0\cup V_v^1$, 
    where~$V_v^i=\{v^i\}$ for each~$i\in\{0,1\}$, 
    and define~$V_s=\{s\}$ and~$V_t=\{t\}$.
    Set~$V^*=\bigcup_{v\in V} V_v$.
    We set~$V'=V^* \cup \{x_e\mid e\in E\}$.
    Next,
    for each edge~$e=\{v,w\}\in E$ with~$v,w\not\in\{s,t\}$, 
    let~$E_e^0=\{\{v^0,x_e\},\{w^0,x_e\}\}$ and~$E_e^1=\{\{v^1,x_e\},\{w^1,x_e\}\}$,
    and for each edge~$e=\{v,w\}\in E$ with~$v\in\{s,t\}$ and~$w\not\in\{s,t\}$, 
    let~$E_e^0=\{\{v^0,x_e\},\{s,x_e\}\}$ and~$E_e^1=\{\{v^1,x_e\},\{s,x_e\}\}$.
    If~$e=\{s,t\}\in E$, then set $E_e^0=E_e^1=\{\{\{s,x_e\},\{x_e,t\}\}$.
    Finally,
    let~$E_e=E_e^0\cup E_e^1$ and
    $E_i'=\bigcup_{e\in E_i} E_e$.
    Set $k'=2k-1$ and $\ell'=\ell$.
    This finishes the construction of instance~$I'\ceq (\TG'=(V',E_1',\ldots,E_\tau'),s,t,k',\ell')$ of~\vimsp{}.
    Note that~$I'$ can be constructed in~$\O(|\TG|)$ time.
    We claim that~$I$ is a \yes-instance if and only if~$I'$ is a \yes-instance.
    
    \RD{}
    Let~$(P_1,\dots,P_\tau)$ be a solution to~$I$.
    We claim that~$(P_1',\ldots,P_\tau')$ with~$V(P_i')=\bigcup_{v\in V(P_i)} V_v^{i\bmod 2} \cup \{x_e\mid e\in E(P_i)\}$ and~$E(P_i')=\bigcup_{e\in E(P_i)} E_e^{i\bmod 2}$
    is a solution to~$I'$.
    First, observe that each~$P_i'$ is an~$s$-$t$ path,
    and~$|V(P_i')|=|V(P_i)|+|E(P_i)|=2k-1$.
    Moreover, 
    $|(V(P_i')\cap V(P_{i+1}'))\setminus\{s,t\}|=|\{x_e\mid e\in E(P_i)\cap E(P_{i+1})\}|\leq \ell=\ell'$.
    
    \LD{}
    Let~$(P_1',\ldots,P_\tau')$ be a solution to~$I'$ such that for each~$P_i'$ it holds true that~$|V_v\cap V(P_i')|\leq 1$.
    Note that~$V(P_i')=\{s,t\}\uplus W_i\uplus X_i$ with~$W_i\subseteq V^*$ and~$X_i\subseteq \{x_e\mid e\in E\}$.
    We claim that~$(P_1,\dots,P_\tau)$ with~$V(P_i)=\{v\mid v^i\in W_i\}\cup\{s,t\}$ and~$E(P_i)=\{e\mid x_e\in X_i\}$ is a solution to~$I$.
    First, observe that each~$P_i$ is an~$s$-$t$ path,
    and~$|V(P_i)|\leq k$.
    Moreover,
    $|E(P_i)\cap E(P_{i+1})|\leq |X_i\cap X_{i+1}|\leq \ell'=\ell$.
\end{proof}
}

\noindent
Due to~\cref{prop:edgetovertex,prop:edgetovertexintersect},
often we just may prove lower bounds for \emsp{} and~\eimsp{}, 
and upper bounds for~\vmsp{} and~\vimsp{},
and transfer the results to their respective counterparts.

\section{NP-hardness even for two snapshots of maximum degree four}
\label{sec:nphardness}
\appendixsection{sec:nphardness}

In this section,
we prove that all four problems are \NP-hard even for only two snapshots and the maximum underlying vertex-degree being four.

\ifarxiv{}
\begin{theorem}
		\else{}
\begin{theorem}[\appref{thm:emspnphard}]
\fi{}
 \label{thm:emspnphard}
	\emsp{} and~\eimsp{},
	the latter with~$\ell=0$,
	are \NP-hard even if~$\TG$ consists of two snapshots both being series-parallel graphs and~$\Delta(\ug{\TG})=4$. 
\end{theorem}
\appendixproof{thm:emspnphard}
{
\begin{proof}
		The theorem follows directly from \cref{prop:emspnphardtau,prop:eimspnphardtau}.
\end{proof}
}
\noindent
We give two polynomial-time many-one reductions from the \NP-complete \prob{3-SAT}, 
each employing the following.

\begin{construction*}
 \label{constr:emspnphard}
 Let~$(X=\{x_1,\ldots,x_n\},\calC=(C_1,\ldots,C_n))$ be an instance of~\prob{3-SAT} where w.l.o.g.\ the number~$n$ of variables equals the number of clauses, 
 and let~$d\geq 2$ denote the most frequent appearance (along the clause sequence) of any literal of some variable in~$X$.
 We construct a temporal graph~$\TG=(V,E_1,E_2)$ as follows (see~\cref{fig:emspnphard} for an illustration).
 \begin{figure}[t]
  \begin{tikzpicture}
    \tikzstyle{xnode}=[circle, draw,scale=2/3]
    \tikzstyle{xedge}=[thick]
    \def\xr{0.98}
    \def\yr{0.75}
    \def\ys{0.5}
	\begin{scope}
    \node (s) at (-0.5*\xr,1.5*\yr)[]{(a)};
    \node (s) at (0,0)[xnode,label=-90:{$s$}]{};
    \node (c11) at (1*\xr,0)[xnode,label=-90:{$c_1^1$}]{};
    \node (a11) at (2*\xr, \ys*\yr)[xnode,label=90:{$a_1^1$}]{};
    \node (b11) at (2*\xr, -\ys*\yr)[xnode,label=-90:{$b_1^1$}]{};
    \node (a12) at (3*\xr, \ys*\yr)[xnode,label=90:{$a_2^1$}]{};
    \node (b12) at (3*\xr, -\ys*\yr)[xnode,label=-90:{$b_2^1$}]{};
    \node (a1x) at (4*\xr, \ys*\yr)[scale=0.6]{$\cdots$};
    \node (b1x) at (4*\xr, -\ys*\yr)[scale=0.6]{$\cdots$};
    \node (a12d) at (5*\xr, \ys*\yr)[xnode,label=90:{$a_{2d}^1$}]{};
    \node (b12d) at (5*\xr, -\ys*\yr)[xnode,label=-90:{$b_{2d}^1$}]{};
    \node (c12) at (6*\xr,0)[xnode,label=-90:{$c_2^1$}]{};
    \node (c13) at (7*\xr,0)[xnode,label=-90:{$c_3^1$}]{};
    \node (a21) at (8*\xr, \ys*\yr)[xnode,label=90:{$a_1^2$}]{};
    \node (b21) at (8*\xr, -\ys*\yr)[xnode,label=-90:{$b_1^2$}]{};
    \node at (8.75*\xr, \ys*\yr)[]{$\cdots$};
    \node at (8.75*\xr, -\ys*\yr)[]{$\cdots$};
    \node at (9.5*\xr, 0*\yr)[]{$\cdots$};
    \node at (10.25*\xr, \ys*\yr)[]{$\cdots$};
    \node at (10.25*\xr, -\ys*\yr)[]{$\cdots$};
    \node (an2d) at (11*\xr, \ys*\yr)[xnode,label=90:{$a_{2d}^n$}]{};
    \node (bn2d) at (11*\xr, -\ys*\yr)[xnode,label=-90:{$b_{2d}^n$}]{};
    \node (c12n) at (12*\xr,0)[xnode,label=-90:{$c_{2n}^1$}]{};
    \node (t) at (13*\xr,0)[xnode,label=-90:{$t$}]{};
    \draw[xedge] (s) to (c11) to (a11) to (a12) to (a1x) to (a12d) to (c12) to (c13) to (a21);
    \draw[xedge] (c11) to (b11) to (b12) to (b1x) to (b12d) to (c12) to (c13) to (b21);
    \draw[xedge] (an2d) to (c12n) to (t);
    \draw[xedge] (bn2d) to (c12n) to (t);
    \draw[ultra thick,blue] (a11) to (a12);
    \end{scope}

    \begin{scope}[yshift=\yr*-4.0cm]
    \node (s) at (-0.5*\xr,2*\yr)[]{(b)};
    \node (s) at (0,0)[xnode,label=-90:{$s$}]{};
    \node (c11) at (1*\xr,0)[xnode,label=-90:{$c_1^2$}]{};
    \node (a11) at (3*\xr, 1*\yr)[xnode,label={[yshift=\yr*-0.1cm]90:{$a_1^1$}}]{};
    \node (bj1) at (3*\xr, 0*\yr)[xnode,label={[yshift=\yr*-0.1cm]90:{$b_1^j$}}]{};
    \node (ai1) at (3*\xr, -1*\yr)[xnode,label={[yshift=\yr*-0.1cm]90:{$a_1^i$}}]{};
    \node (a12) at (4*\xr, 1*\yr)[xnode,label={[yshift=\yr*-0.1cm]90:{$a_2^1$}}]{};
    \node (bj2) at (4*\xr, 0*\yr)[xnode,label={[yshift=\yr*-0.1cm]90:{$b_2^j$}}]{};
    \node (ai2) at (4*\xr, -1*\yr)[xnode,label={[yshift=\yr*-0.1cm]90:{$a_2^i$}}]{};
    \node (c12) at (6*\xr,0)[xnode,label=-90:{$c_2^2$}]{};
    \node (c13) at (7*\xr,0)[xnode,label=-90:{$c_3^2$}]{};
    \node (ax) at (8*\xr, 1*\yr)[]{};
    \node (bxx) at (8*\xr, 0*\yr)[]{};
    \node (bx) at (8*\xr, -1*\yr)[]{};
    \node at (8.5*\xr, 1*\yr)[]{$\cdots$};
    \node at (8.5*\xr, 0*\yr)[]{$\cdots$};
    \node at (8.5*\xr, -1*\yr)[]{$\cdots$};
    \node (temp) at (9.5*\xr, 0*\yr)[]{$\cdots$};
    \node at (10.5*\xr, 1*\yr)[]{$\cdots$};
    \node at (10.5*\xr, 0*\yr)[]{$\cdots$};
    \node at (10.5*\xr, -1*\yr)[]{$\cdots$};
    \node (anx) at (11*\xr, 1*\yr)[]{};
    \node (bnxx) at (11*\xr,0*\yr)[]{};
    \node (bnx) at (11*\xr, -1*\yr)[]{};
    \node (c12n) at (12*\xr,0)[xnode,label=-90:{$c_{2n}^2$}]{};
    \node (t) at (13*\xr,0)[xnode,label=-90:{$t$}]{};
    \draw[xedge] (s) to (c11) to (a11) to (a12) to (c12) to (c13) to (ax);
    \draw[xedge] (c11) to (bj1) to (bj2) to (c12) to (ai2) to (ai1) to (c11);
    \draw[xedge] (ax) to (c13);
    \draw[xedge] (bx) to (c13);
    \draw[xedge] (bxx) to (c13);
    \draw[xedge] (anx) to (c12n) to (t);
    \draw[xedge] (bnx) to (c12n);
    \draw[xedge] (bnxx) to (c12n);
    \draw[ultra thick,blue] (a11) to (a12);
    \end{scope}

    \end{tikzpicture}
  \caption{Illustration to~\cref{constr:emspnphard} with~(a) illustrating the first snapshot and (b) illustrating the second snapshot, exemplified for clause~$C_1=(x_1 \lor \overline{x_j} \lor x_i)$.
  The edge~$\{a_1^1,a_2^1\}$ is highlighted in both~(a) and~(b).}
  \label{fig:emspnphard}
 \end{figure}
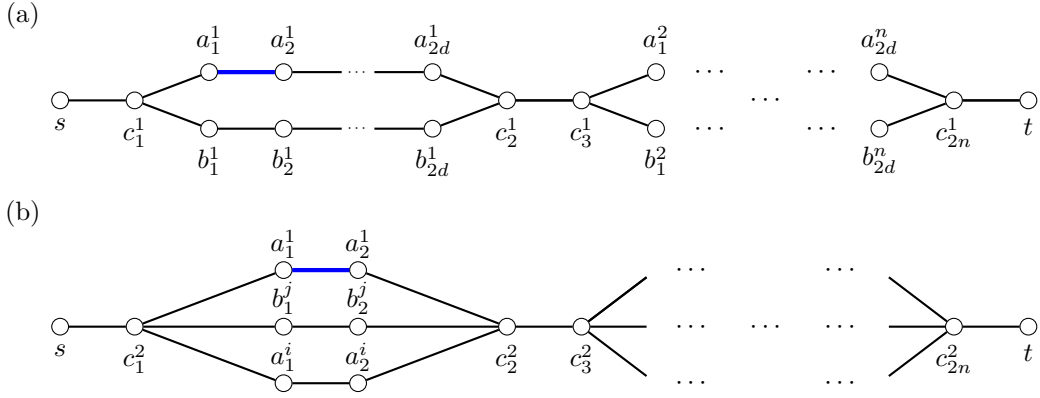

 Let 
 $V:=\{s,t\}\cup \{c_1^i,\ldots,c_{2n}^i\mid i\in\{1,2\}\}\cup\{a_1^i,\ldots,a_{2d}^i\mid x_i\in X\}\cup\{b_1^i,\ldots,b_{2d}^i\mid x_i\in X\}$.
 Let~$E_{i,a}:=\bigcup_{1\leq j< 2d}\{\{a_{j}^i,a_{j+1}^i\}\}$ and~$E_{i,b}:=\bigcup_{1\leq j< 2d}\{\{b_{j}^i,b_{j+1}^i\}\}$.
 Then~$E_1$ contains the edge~$\{s,c_1^1\}$, the edge set~$\bigcup_{1\leq i \leq n} \{\{c_{2i-1}^1,a_1^i\},\{c_{2i-1}^1,b_1^i\}\}$,
 the edge set~$\bigcup_{1\leq i \leq n} \{\{c_{2i}^1,a_{2d}^i\},\allowbreak\{c_{2i}^1,b_{2d}^i\}\}$,
 the edge~$\{t,c_{2n}^1\}$,
 the edge set~$\bigcup_{1\leq i < n} \{\{c_{2i}^1,c_{2i+1}^1\}\}$,
 and the edge sets~$\bigcup_{1\leq i\leq n} E_{i,a}$ and~$\bigcup_{1\leq i\leq n} E_{i,b}$.
 For~$E_2$,
 for each clause~$C_q\in \calC$ we define the vertex set~$V_{C_q}$ and edge set~$E_{C_q}$ as follows.
 If~${C_q}$ contains the \ith{$j$} appearance of the positive literal~$x_i$, 
 then add~$a_{2j-1}^i,a_{2j}^i$ to~$V_{C_q}$ 
 and the edges~$\{a_{2j-1}^i,a_{2j}^i\},\{c_{2q-1}^2,a_{2j-1}^i\},\{c_{2q}^2,a_{2j}^i\}$ to~$E_{C_q}$.
 If~${C_q}$ contains the \ith{$j$} appearance of the negative literal~$\overline{x_i}$, 
 then add~$b_{2j-1}^i,b_{2j}^i$ to~$V_{C_q}$ 
 and the edges~$\{b_{2j-1}^i,b_{2j}^i\},\{c_{2q-1}^2,b_{2j-1}^i\},\{c_{2q}^2,b_{2j}^i\}$ to~$E_{C_q}$.
 Then,
 $E_2$ contains the edges~$\{s,c_1^2\}$, 
 $\{t,c_{2n}^2\}$, 
 the edge set~$\bigcup_{1\leq i < n} \{\{c_{2i}^2,c_{2i+1}^2\}\}$,
 and
 $E_{C_q}$ for each~$q\in\set{n}$.
 This finishes the construction of~$\TG$.
 It is not difficult to see that~$(V,E_1)$ and~$(V,E_2)$ are series-parallel graphs.
 Moreover,
 $\Delta(\ug{\TG})=4$.
 Set~$k=2+2n+2d\cdot n$.
\end{construction*}

\noindent
Intuitively,
if an instance constructed using~\cref{constr:emspnphard} is a \yes-instance for \emsp{},
then the $s$-$t$~path in the first snapshot selects setting variables to true or false such that
the $s$-$t$~path in the second snapshot can pass a literal for each clause.
It follows that~\cref{constr:emspnphard} is a polynomial-time many-one reduction.

The next two propositions, 
\cref{prop:emspnphardtau,{prop:eimspnphardtau}},
together prove~\cref{thm:emspnphard}.

\ifarxiv{}
\begin{proposition}
		\else{}
\begin{proposition}[\appref{prop:emspnphardtau}]
\fi{}
 \label{prop:emspnphardtau}
 \emsp{} is \NP-hard even if~$\TG$ consists of two snapshots both being series-parallel graphs and~$\Delta(\ug{\TG})=4$. 
\end{proposition}

\appendixproof{prop:emspnphardtau}
{
    \begin{proof}
    Let~$I=(X=\{x_1,\ldots,x_n\},\calC=(C_1,\ldots,C_n))$ be an instance of~\prob{3-SAT} such that the number~$n$ of variables equals the number of clauses, 
    and let~$d$ denote the largest appearance of any literal of some variable in~$X$.
    Let~$I'=(\TG=(V,E),s,t,k,\ell)$ with~$\ell=5n+2dn+4$ and $k=2+2n+2d\cdot n$ be the instance of~\emsp{} obtained from~$I$ using \cref{constr:emspnphard}.
    We claim that~$I$ is a \yes-instance if and only if~$I'$ is a \yes-instance.

    \RD{}
    Let~$X'\subseteq X$ be a solution.
    We construct the paths~$(P_1,P_2)$ as follows.
    Vertex set~$V(P_1)$ contains~$\{s,t\}\cup \{c_1^1,\ldots,c_{2n}^1\}$
    and~$V(P_2)$ contains~$\{s,t\}\cup \{c_1^2,\ldots,c_{2n}^2\}$.
    For each~$i\in\set{n}$,
    if~$x_i\in X'$, then~$V(P_1)$ contains~$\{a_1^i,\ldots,a_{2d}^i\}$, and
    if~$x_i\not\in X'$, then~$V(P_1)$ contains~$\{b_1^i,\ldots,b_{2d}^i\}$.
    Set~$E(P_1)=E(G[V(P_1)])$.
    Note that
    $P_1$ is an~$s$-$t$ path and~$|V(P_1)|=2+2n+2d\cdot n=k$.
    Observe that $V_{C_q}\cap V(P_1)\neq \emptyset$,
    since~$X'$ is a solution.
    For~$E(P_2)$,
    for each~$q\in\set{n}$,
    if~$a_{2j-1}^i,a_{2j}^i \in V_{C_q}\cap V(P_1)$,
    then~$E(P_2)$ contains the edges~$\{a_{2j-1}^i,a_{2j}^i\},\{c_{2q-1}^2,a_{2j-1}^i\},\{c_{2q}^2,a_{2j}^i\}$ (analogously for~$b$).
    Note that
    $P_2$ is an~$s$-$t$ path in~$(V,E_2)$ with~$|V(P_2)|=2+2n+2n<k$.
    It remains to consider~$\symdif{E(P_1)}{E(P_2)}$.
    Let~$B=\{\{v,w\}\mid v,w\in V_{C_q}\cap V(P_2)\}$.
    Observe that~$E(P_1)\cap E(P_2)=B$,
    since for all other edges in~$E(P_2)\setminus B$ we have that at least one endpoint is in~$\{c_1^2,\dots,c_{2d}^2\}$,
    which is disjoint from~$V(P_1)$.
    Hence~$|\symdif{E(P_1)}{E(P_1)}|\leq |E(P_1)\cup E(P_2)|-|E(P_1)\cap E(P_2)|=(2+2n+2dn+2+2n+2n-2)-n=5n+dn+2=\ell$.

    \LD{}
    Let~$(P_1,P_2)$ be a solution to~$I'$.
    Observe that for all~$i\in\set{n}$,
    $V(P_1)$ contains as a subset either the set~$\{a_1^i,\dots,a_{2d}^i\}$ or the set~$\{b_1^i,\dots,b_{2d}^i\}$.
    Let~$X'=\{x_i\in X\mid a_1^i,\dots,a_{2d}^i\in V(P_1) \}$.
    We claim that~$X'$ is a solution to~$I$.
    Let~$C_q$ be an arbitrary clause from~$\calC$.
    Let~$\{c_{2q-1}^2,v,w,c_{2q}^2\}$ be the vertices on the subpath from~$P_2$ connecting~$c_{2q-1}^2$ with~$c_{2q}^2$,
    where~$v,w\in V_{C_q}$.
    Note that~$\{v,w\}\in E(P_1)$, 
    since otherwise~$|E(P_1)\cup E(P_2)|-|E(P_1)\cap E(P_2)|> (2+2n+2dn+2+2n+2n-2)-n=\ell$.
    Hence,
    if~$\{v,w\}=\{a_{2j-1}^i,a_{2j}^i\}$ for some~$i\in\set{n}$ and~$j\in\set{2d-1}$, then
    $x_i\in X'$, setting~$C_q$ to true.
    Otherwise,
    if~$\{v,w\}=\{b_{2j-1}^i,b_{2j}^i\}$ for some~$i\in\set{n}$ and~$j\in\set{2d-1}$, then
    $x_i\not\in X'$, setting~$C_q$ to true ($x_i$ is negated in~$C_q$).
    Since~$C_q$ was chosen arbitrarily,
    it follows that~$X'$ is a solution to~$I$.
    \end{proof}
}

\noindent
Interestingly,
\cref{constr:emspnphard} also gives a polynomial-time many-one reduction for \eimsp{}.
Here the intuition is opposite:
the first snapshot path selects setting the variables to the \emph{complement} of a satisfying assignment such that the second snapshot path can pass the ``clause gadgets'' without passing any edge contained in the first snapshot path.

\ifarxiv{}
\begin{proposition}
		\else{}
\begin{proposition}[\appref{prop:eimspnphardtau}]
\fi{}
 \label{prop:eimspnphardtau}
 \eimsp{} is \NP-hard even if~$\TG$ consists of two snapshots both being series-parallel graphs, $\Delta(\ug{\TG})=4$, and~$\ell=0$. 
\end{proposition}

\appendixproof{prop:eimspnphardtau}
{
    \begin{proof}
    Let~$I=(X=\{x_1,\ldots,x_n\},\calC=(C_1,\ldots,C_n))$ be an instance of~\prob{3-SAT} such that the number~$n$ of variables equals the number of clauses, 
    and let~$d$ denote the largest appearance of any literal of some variable in~$X$.
    Let~$I'=(\TG=(V,E),s,t,k,\ell)$ with~$\ell=0$ and $k=2+2n+2d\cdot n$ be the instance of~\eimsp{} obtained from~$I$ using \cref{constr:emspnphard}.
    We claim that~$I$ is a \yes-instance if and only if~$I'$ is a \yes-instance.
    The proof works analogously to the proof of \cref{prop:edgetovertex},
    except for the fact that~$P_1$ selects the complement of a satisfying assignment.
    
    \RD{}
    Let~$X'\subseteq X$ be a solution.
    We construct the paths~$(P_1,P_2)$ as follows.
    Vertex set~$V(P_1)$ contains~$\{s,t\}\cup \{c_1^1,\ldots,c_{2n}^1\}$
    and~$V(P_2)$ contains~$\{s,t\}\cup \{c_1^2,\ldots,c_{2n}^2\}$.
    Let~$H$ be an auxiliary, initially empty vertex set.
    For each~$i\in\set{n}$,
    if~$x_i\in X'$, then~$V(P_1)$ contains~$\{b_1^i,\ldots,b_{2d}^i\}$ and~$H$ contains~$\{a_1^i,\ldots,a_{2d}^i\}$, and
    if~$x_i\not\in X'$, then~$V(P_1)$ contains~$\{a_1^i,\ldots,a_{2d}^i\}$ and~$H$ contains~$\{b_1^i,\ldots,b_{2d}^i\}$.
    Note that~$H\cap V(P_1)=\emptyset$.
    Set~$E(P_1)=E(G[V(P_1)])$.
    Note that
    $P_1$ is an~$s$-$t$ path and~$|V(P_1)|=2+2n+2d\cdot n=k$.
    Observe that $V_{C_q}\cap H\neq \emptyset$,
    since~$X'$ is a solution.
    For~$P_2$,
    for each~$q\in\set{n}$,
    if~$a_{2j-1}^i,a_{2}^i \in V_{C_q}\cap H$ with~$i$ smallest,
    then~$V(P_2)$ contains~$a_{2j-1}^i,a_{2}^i$ and~$E(P_2)$ contains the edges~$\{a_{2j-1}^i,a_{2j}^i\}$, $\{c_{2q-1}^2,a_{2j-1}^i\}$, and~$\{c_{2q}^2,a_{2j}^i\}$ (analogously for~$b$).
    Note that
    $P_2$ is an~$s$-$t$ path in~$(V,E_2)$ with~$|V(P_2)|=2+2n+2n<k$.
    It remains to consider~${E(P_1)}\cap {E(P_2)}$.
    Note that~$E(P_1)\cap E(P_2)=\emptyset$,
    since~$V(P_1)\cap H=\emptyset$, and~$V(P_2)\cap V_{C_q}\subseteq H$ for all~$q\in\set{n}$.

    \LD{}
    Let~$(P_1,P_2)$ be a solution to~$I'$.
    Observe that for all~$i\in\set{n}$,
    $P_1$~contains as a subset either the set~$\{a_1^i,\dots,a_{2d}^i\}$ or the set~$\{b_1^i,\dots,b_{2d}^i\}$.
    Let~$X'=\{x_i\in X\mid b_1^i,\dots,b_{2d}^i\in V(P_1) \}$.
    We claim that~$X'$ is a solution to~$I$.
    Let~$C_q$ be an arbitrary clause from~$\calC$.
    Let~$\{c_{2q-1}^2,v,w,c_{2q}^2\}$ be the vertices on the subpath from~$P_2$ connecting~$c_{2q-1}^2$ with~$c_{2q}^2$,
    where~$v,w\in V_{C_q}$.
    Note that~$\{v,w\}\not\in E(P_1)$, 
    since otherwise~$|E(P_1)\cap E(P_2)|> 0$.
    Hence,
    if~$\{v,w\}=\{a_{2j-1}^i,a_{2j}^i\}$ for some~$i\in\set{n}$ and~$j\in\set{2d-1}$,
    then~$\{b_{2j-1}^i,b_{2j}^i\}\subseteq V(P_1)$ and hence~$x_i\in X'$, setting~$C_q$ to true.
    Otherwise,
    if~$\{v,w\}=\{b_{2j-1}^i,b_{2j}^i\}$ for some~$i\in\set{n}$ and~$j\in\set{2d-1}$,
    then~$\{a_{2j-1}^i,a_{2j}^i\}\subseteq V(P_1)$ and hence~$x_i\not\in X'$ setting~$C_q$ to true ($x_i$ is negated in~$C_q$).
    Since~$C_q$ was chosen arbitrarily,
    it follows that~$X'$ is a solution to~$I$.
    \end{proof}
}

\noindent
Due to~\cref{prop:edgetovertex,prop:edgetovertexintersect},
we get the following from~\cref{thm:emspnphard}.

\begin{corollary}
 \label{cor:vertexmspnphardtau}
 \vmsp{} and~\vimsp{} with~$\ell=0$ are \NP-hard even if~$\tau=2$ and~$\Delta(\ug{\TG})=4$.
\end{corollary}

We proved
\eimsp{} and \vimsp{} to remain~\NP-hard even if~$\ell=0$ and~$\tau=2$.
This leads us to ask whether for a constant value of~$\ell+\tau$,
\emsp{} or \vmsp{} remain \NP-hard.
In fact,
we prove this to be true for the vertex-variant.

\ifarxiv{}
\begin{theorem}
		\else{}
\begin{theorem}[\appref{thm:vmsppNPhelltau}]
\fi{}
 \label{thm:vmsppNPhelltau}
 \vmsp{} is~\NP-hard and admits no~$2^{o(k)}\cdot (|\TG|)^{O(1)}$-time algorithm unless the Exponential Time Hypothesis fails, 
 even if~$\ell=0$ and~$\tau=2$.
\end{theorem}

\noindent
It remains open whether~\emsp{} is contained in~\XP{} regarding~$\ell+\tau$. 

\appendixproof{thm:vmsppNPhelltau}
{
\noindent
We give a polynomial-time many-one reduction from the following \NP-complete~\cite{GareyJ79} problem.
\problemdef{Hamiltonian Path}
{An undirected graph~$G=(V,E)$.}
{Is there a Hamiltonian path in~$G$, i.e., a path in~$G$ that contains every vertex of~$G$?}

\begin{construction*}
 \label{constr:vmsppNPhelltau}
 Let~$(G=(V,E))$ be an instance of~\prob{Hamiltonian Path} and let~$V=\{v_1,\dots,v_n\}$ be enumerated. 
 We construct the temporal graph~$\TG=(V',E_1,E_2)$ with~$V'\ceq V\cup\{s,t\}$ as follows.
 Set~$E_1\ceq \{\{s,v_1\}\}\cup \{\{v_n,t\}\}\cup \bigcup_{i=1}^{n-1} \{\{v_i,v_{i+1}\}\}$.
 Set~$E_2\ceq E\cup\bigcup_{i=1}^n\{\{s,v_i\},\{t,v_i\}\}$.
 Finally, 
 set~$k=n+2$ and~$\ell=0$.
\end{construction*}

    \begin{proof}[Proof of~\cref{thm:vmsppNPhelltau}]
    Let~$I=(G=(V,E))$ be an instance of~\prob{Hamiltonian Path} and let~$V=\{v_1,\dots,v_n\}$ be enumerated.
    Let~$I'=(\TG=(V',E_1,E_2),s,t,k,\ell)$ be the instance obtained from~$I$ using~\cref{constr:vmsppNPhelltau}.
    We claim that~$I$ is a \yes-instance if and only if~$I'$ is a \yes-instance.
    
    \RD{}
    Let~$P$ be a Hamiltonian path in~$G$ with endpoints~$v_i$ and~$v_j$.
    Construct~$(P_1,P_2)$ as follows.
    Let~$V(P_1)=V'$ and~$E(P_1)=E_1$.
    Let~$V(P_2)=V'$ and~$E(P_2)=E(P)\cup\{\{s,v_i\},\{t,v_j\}\}$.
    Since~$V(P_1)=V(P_2)=V'$,
    we have that~$|V(P_1)|=n+2=k$ and~$\symdif{V(P_1)}{V(P_2)}=\emptyset$.
    Hence,
    $(P_1,P_2)$ is a solution to~$I'$.

    \LD{}
    Let~$I'$ be a \yes-instance of~\vmsp{} and let~$(P_1,P_2)$ be a solution.
    By the construction of~$(V,E_1)$ and the fact that~$(P_1,P_2)$ is a solution to~$I'$,
    we know that~$V(P_1)=V(P_2)=V'$ .
    We construct a Hamiltonian path~$P=(V_P,E_P)$ from~$P_2$ as follows.
    Let~$V_P=V(P_2)\setminus \{s,t\}$,
    and let~$E_P=\{e\in E(P_2)\mid e\cap\{s,t\}=\emptyset\}$.
    That is,
    $P$ is the subpath of~$P_2$ where the neighbors of~$s$ and~$t$ on~$P_2$ form the endpoints.
    It follows that~$P$ is a path in~$G$ contain all vertices in~$V$,
    and hence, 
    $I$ is a \yes-instance.
    
    Finally, 
    note that since~$k=n+2$, 
    and by the fact that~\prob{Hamiltonian Path} admits no~$2^{o(n)}\cdot (n+m)^{O(1)}$-time algorithm unless the Exponential Time Hypothesis fails,
    the second part of the theorem follows.
    \end{proof}

}

\section{The role of the parameter path length}
\label{sec:parak}
\appendixsection{sec:parak}

In this section,
we focus on the parameter~$k$,
the maximum number of vertices in any~$s$-$t$ path.
It is not hard to see that all variants allow for an \XP{}-algorithm when parameterized by the number~$k$ of maximal vertices in each path.

\ifarxiv{}
\begin{proposition}
		\else{}
\begin{proposition}[\appref{thm:allxp}]
\fi{}
 \label{thm:allxp}
 \vmsp{} and \vimsp{}, and hence \emsp{} and \eimsp{}, 
 are solvable in~$\Delta_{\max}^{O(k)}\cdot |\TG|^{O(1)}$ time,
 where~$\Delta_{\max} = \max_{i\in\set{\tau}}\Delta((V,E_i))$.
\end{proposition}

\appendixproof{thm:allxp}
{
\begin{proof}[Proof Sketch]
 The proof is in line with the proof of~\cite[Proposition~4.2]{FluschnikNRZ19}.
 We sketch the proof in the general setup~$\Pi$-\textsc{MstP}.
 
 Given an instance~$I=(\TGfull,s,t,k,\ell)$,
 construct a directed graph~$D=(V',A)$ with vertex set~$V'=V_1'\uplus\dots\uplus V_\tau'\cup\{s',t'\}$ and arc set~$A$ together with a mapping~$\gamma:V'\to (2^V,2^{\binom{V}{2}})$ as follows.
 For each~$i\in\set{\tau}$ and each $s$-$t$ path $P$ of length at most $k-1$ in $(V,E_i)$ 
 add a vertex~$v$ to~$V_i'$ and set~$\gamma(v)=P$.
 It is easy to verify that a straight-forward search tree algorithm (starting in $s$ and exploring edges until the path has length $k-1$) 
 can enumerate all $s$-$t$ path of length $k-1$ in $(V,E_i)$ in $O(\Delta_{\max}^k\cdot |E_i|)$ time, for any $i \in \set{\tau}$.
 Next,
 for each~$i\in\set{\tau-1}$,
 if for two vertices~$v\in V_i'$ and~$w\in V_{i+1}'$ it holds true that~$\dist[\Pi]{\gamma(v),\gamma(w)}\leq \ell$,
 then add the arc~$\{v,w\}$.
 Finally make~$s'$ adjacent with all vertices in~$V_1'$,
 and~$t'$ adjacent with all vertices in~$V_\tau'$.
 This finishes the construction.
 It is not difficult to see that~$I$ is a \yes-instance if and only if there is an~$s'$-$t'$ path in~$D$ 
 (which can be checked in time linear in the size of~$D$).
\end{proof}
}

\noindent
We will prove that the parameterization with~$k$ distinguishes similarity from dissimilarity:
While~\emsp{} and~\vmsp{} are \W{1}-hard regarding~$k$ (even regarding~$k+\tau$),
each of~\eimsp{} and~\vimsp{} turn out to be \fpt{}.

\subsection{W[1]-hardness for the similarity variant \texorpdfstring{regarding~\boldmath$k+\tau$ and~$\ug{\nu}$}{}}

We prove that \emsp{} is \W{1}-hard regarding~$k+\tau$ even if the upper bound~$\ell$ on the sizes of consecutive symmetric differences is constant.
Due to~\cref{prop:edgetovertex},
we then obtain the same result for~\vmsp{}.
The proof is by a parameterized reduction from the~\W{1}-complete \textsc{Multicolored Clique} problem.

\ifarxiv{}
\begin{theorem}
		\else{}
\begin{theorem}[\appref{thm:w1hard}]
\fi{}
 \label{thm:w1hard}
 Even if~$\ell=4$ and each snapshot is bipartite,
 \emsp{} is \NP-hard and 
 \W{1}-hard when parameterized by~$k+\tau$.
\end{theorem}

\appendixproof{thm:w1hard}
{
  \noindent
  To prove~\cref{thm:w1hard},
  we reduce from the \W{1}-complete \prob{Multicolored Clique} problem.
  \problemdef{Multicolored Clique}
  {An undirected, $r$-partite graph~$G=(V_1,\dots,V_r,E)$.}
  {Is there a clique of size~$r$ in~$G$?}
  Intuitively,
  in each snapshot we order the $r$~parts differently such that any two colors appear at least once consecutively.
  Hence,
  if there is a sequence of $s$-$t$~paths through all $r$~parts in each snapshot over the same vertex set,
  then this witnesses the existence of each edge of any two vertices from distinct parts.
  For the ordering of the $r$~parts in the snapshots,
  we define the following.

  \begin{definition}
  For all~$1\leq i\leq 1+\binom{r}{2}$,
  let $\pi^r_i$~be a permutation of~$(1,\ldots,r)$ as follows.
  Let~$\pi^r_1=(1,\ldots,r)$.
  For $i>1$,
  let~$\pi^r_i$ be obtained from~$\pi^r_{i-1}$ as follows.
  Let~$j$ be the index such that~$\pi^r_{i-1}(j)<\pi^r_{i-1}(j+1)$ and there is no~$j'\neq j$ such that~$\pi^r_{i-1}(j')<\pi^r_{i-1}(j)$ and~$\pi^r_{i-1}(j')<\pi^r_{i-1}(j'+1)$.
  Then set~$\pi^r_i(j)=\pi^r_{i-1}(j+1)$,
  $\pi^r_i(j+1)=\pi^r_{i-1}(j)$,
  and~$\pi^r_i(j')=\pi^r_{i-1}(j')$ for all~$j'\in\set{r}\setminus\{j,j+1\}$.
  \end{definition}

  Note that each pair is swapped exactly once,
  hence we have that~$\pi_{1+\binom{r}{2}}^r=(r,r-1,\dots,1)$.
  Moreover,
  we have the following.
  
  \begin{observation}
   \label{obs:w1hardcouple}
   For every distinct~$r_1,r_2\in\set{r}$,
   there is an~$i\in\set{1+\binom{r}{2}}$ 
   such that~$|j_1-j_2|=1$,
   where~$\pi_{i}^r(j_1)=r_1$ and~$\pi_{i}^r(j_2)=r_2$.
  \end{observation}

  \noindent
  Next we describe the construction for the reduction.

  \begin{construction*}
  \label{constr:w1hard}
  Let~$(G=(V_1,\dots,V_r,E))$ be an instance of \prob{Multicolored Clique}.
  Let~$E_{i,j}\subseteq E$ denote the set of all edges between~$V_i$ and~$V_j$.
  We construct an instance~$(\TG=(V,E_1,\dots,E_\tau),s,t,k,\ell)$ with~$\tau=r(r-1)+1$ of \emsp{} as follows.
  Let~$V=\{s,t\}\cup V_1\cup\dots\cup V_k$.
  Add the edge sets~$\bigcup_{v\in V_{\pi_i^r(1)}} \{\{s,v\}\}$ and
  $\bigcup_{v\in V_{\pi_i^r(r)}} \{\{t,v\}\}$ to~$E_i$.
  Moreover,
  add~$E_{\pi^r_i(j),\pi^r_i(j+1)}$ for all~$1\leq j<r$.
  Set~$k=r+2$ and~$\ell=4$.
  \end{construction*}

  \begin{proof}[Proof of~\cref{thm:w1hard}]
  Let~$I=(G=(V_1,\dots,V_r,E)$ be an instance of \prob{Multicolored Clique}.
  Let~$E_{i,j}\subseteq E$ denote the set of all edges between~$V_i$ and~$V_j$.
  Let~$I'=(\TG=(V,E_1,\dots,E_\tau),s,t,k,\ell)$ be the instance obtained from~$I$ using~\cref{constr:w1hard}.
  We claim that $I$ is a \yes-instance if and only if~$I'$ is a \yes-instance.
  
  \RD{}
  Let~$I$ be a \yes-instance, and let~$C\subseteq V_1\cup\dots\cup V_r$ form a multicolored clique in~$G$.
  We claim that~$(P_1,\ldots,P_\tau)$ with~$V(P_i)=C\cup\{s,t\}$ and~$E(P_i)=E(G_i[V(P_i)])$ is a solution to~$I'$.
  Note that each~$P_i$ is an $s$-$t$ path with~$k=r+2$ vertices,
  since in~$G_i$ the edge set~$E_{\pi^r_i(j),\pi^r_i(j+1)}$ exists for~$j\in\set{r-1}$.
  Moreover,
  $\symdif{E(P_i)}{E(P_{i+1})}$ contains at most four edges,
  since~$\pi^r_i=(\dots,a,b,c,d,\dots)$ and~$\pi^r_{i+1}=(\dots,a,c,b,d,\dots)$,
  where~$b,c$ denote the two unique indices that are swapped from~$\pi^r_i$ to~$\pi^r_{i+1}$.
  
  \LD{}
  Let~$\calP=(P_1,\ldots,P_\tau)$ be a solution to~$I'$.
  Note that~$|V(P_i)\cap V_x|=1$ for all~$x\in\set{r}$,
  since each~$V_x$ forms an~$s$-$t$ separator and~$|V(P)|\leq k=r+2$.
  We claim that~$V(P_i)=V(P_j)$ for all~$i,j\in\set{\tau}$.
  Suppose not,
  then there exists an~$i$ such that~$V(P_i)\neq V(P_{i+1})$.
  Then there are at least five edges in~$\symdif{E(P_i)}{E(P_{i+1})}$:
  Let~$\pi^r_i=(\dots,a,b,c,d,\dots)$ and~$\pi^r_{i+1}=(\dots,a,c,b,d,\dots)$,
  then~$\symdif{E(P_i)}{E(P_{i+1})}$ is a superset of the edge set~$E'$ containing one edge in~$E_{a,b}$, 
  one edge in~$E_{a,c}$,
  one edge in~$E_{c,d}$,
  and one edge in~$E_{b,d}$.
  Moreover,
  let~$x$ be the (smallest) index such that~$V(P_i)\cap V_x\ni v\neq v'\in V(P_{i+1})\cap V_x$.
  Then~$\symdif{E(P_i)}{E(P_{i+1})}$ contains two edges incident with~$v$ and two edges with~$v'$,
  where at most two edges intersect with~$E'$ (in the case of~$x\in\{b,c\}$).
  This contradicts the fact that $\calP$ is a solution. 
  Let~$C=V(P_1)\setminus \{s,t\}$.
  We claim that~$C$ forms a multicolored clique in~$G$.
  First, 
  recall that~$|C\cap V_i|=1$ for all~$i\in\set{r}$.
  Suppose there are~$v,w\in C$, 
  $v\neq w$, 
  such that~$\{v,w\}\not\in E$.
  Let~$v\in V_i$ and~$w\in V_j$.
  Due to~\cref{obs:w1hardcouple},
  there is a snapshot~$G_x$ that contains~$E_{i,j}$.
  Then~$P_x$ is not an~$s$-$t$ path in~$G_x$,
  contradicting $\calP$ being a solution.
  Hence,
  $\{v,w\}\in E$ for all~$v,w\in C$, 
  $v\neq w$.
  That is,
  $C$ forms a multicolored clique in~$G$.
  \end{proof}
}

\noindent
Due to~\cref{prop:edgetovertex},
we get the following.

\begin{corollary}
 \label{cor:vmspw1hardktau}
 \vmsp{} is \W{1}-hard when parameterized by~$k+\tau$, 
 even if~$\ell$ is constant.
\end{corollary}

By~\cref{thm:allxp} and since $k\leq n$,
we know that~\emsp{} and~\vmsp{} are \fpt{} regarding the number~$n$ of 
graph vertices.
Regarding the parameter number~$k$ of path vertices (and even for~$k+\tau$),
by~\cref{thm:w1hard,cor:vmspw1hardktau} we know that both problems are in~\XP{} yet~\W{1}-hard.
Since we can assume~$k\leq 2\ug{\nu}+1$ (recall that $\ug{\nu}$~is the vertex cover number of the underlying graph) 
in every instance and thus naturally~$\ug{\nu}\leq n$,
we can settle the parameterized complexity regarding~$\ug{\nu}$:

\begin{theorem}
 \label{thm:symdifvcw1hardnes}
 When parameterized by~$\ug{\nu}$,
 \vmsp{} with~$\ell=1$ and \emsp{} are \W{1}-hard.
\end{theorem}

\noindent
We prove each statement of~\cref{thm:symdifvcw1hardnes} separately, 
both proofs rely on parameterized reductions from \textsc{Multicolored Clique}.

\ifarxiv{}
\begin{proposition}
		\else{}
\begin{proposition}[\appref{prop:emspvcwhard}]
\fi{}
 \label{prop:emspvcwhard}
 \emsp{} when parameterized by~$\ug{\nu}$ is \W{1}-hard.
\end{proposition}

\appendixproof{prop:emspvcwhard}
{
\noindent
For the construction to follow,
we employ the following.

\begin{definition}
 \label{def:bijectionkminus2}
 For~$k\in\N$,
 we define for all~$i,j\in\set{k}$, $i\neq j$, the bijection
 $\pi^k_{i,j}:\set{k}\setminus\{i,j\}\to \set{k-2}$ such that
 for~$x,y\in \set{k}\setminus\{i,j\}$ if $x<y$, then~$\pi^k_{i,j}(x)<\pi^k_{i,j}(y)$.
\end{definition}

\noindent
We next describe the construction in the parameterized reduction
behind~\cref{prop:emspvcwhard}.

\begin{construction*}
 \label{constr:emspvcwhard}
 Let~$(G=(V_1,\dots,V_k,E))$ be an instance of \prob{Multicolored Clique} with $n = \sum_{i=1}^k |V_i|$.
 We construct a temporal graph~$\TG=(V',E_1,\ldots,E_\tau)$ with $\tau=n\cdot (k-1)$ as follows (see~\cref{fig:emspvcwhard} for an illustration).
 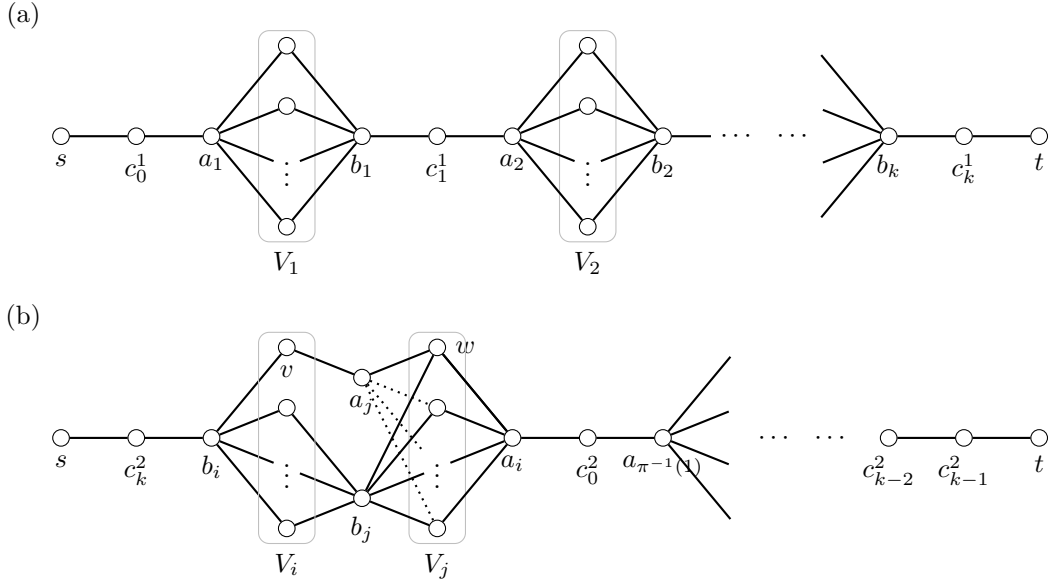
\begin{figure}[t]
 \centering
  \begin{tikzpicture}
    \tikzstyle{xnode}=[circle,scale=0.66,draw]
    \tikzstyle{xedge}=[thick]
    \def\xr{0.99}
    \def\yr{0.8}

    \begin{scope}
    \node at (-0.5*\xr,2*\yr)[]{(a)};
    \node (s) at (0,0)[xnode,label=-90:{$s$}]{};
    \node (c0) at (1*\xr,0)[xnode,label=-90:{$c_0^1$}]{};
    \node (a1) at (2*\xr,0)[xnode,label=-90:{$a_1$}]{};
    \draw[xedge] (s) to (c0) to (a1);
    \node (v11) at (3*\xr,1.5*\yr)[xnode]{};
    \node (v12) at (3*\xr,0.5*\yr)[xnode]{};
    \node (v1x) at (3*\xr,-0.5*\yr)[]{$\vdots$};
    \node (v1n) at (3*\xr,-1.5*\yr)[xnode]{};
    \node at (3*\xr,0*\yr)[minimum height=\yr*3.5 cm, minimum width=\xr*0.75 cm, rounded corners, draw, lightgray, label=-90:{$V_1$}]{};
    \node (b1) at (4*\xr,0*\yr)[xnode,label=-90:{$b_1$}]{};
    \foreach \x in {1,2,x,n}{\draw[xedge] (a1) to (v1\x) to (b1);}
    \node (c1) at (5*\xr,0*\yr)[xnode,label=-90:{$c_1^1$}]{};
    \node (a2) at (6*\xr,0*\yr)[xnode,label=-90:{$a_2$}]{};
    \draw[xedge] (b1) to (c1) to (a2);
    \node (v21) at (7*\xr,1.5*\yr)[xnode]{};
    \node (v22) at (7*\xr,0.5*\yr)[xnode]{};
    \node (v2x) at (7*\xr,-0.5*\yr)[]{$\vdots$};
    \node (v2n) at (7*\xr,-1.5*\yr)[xnode]{};
    \node at (7*\xr,0*\yr)[minimum height=\yr*3.5 cm, minimum width=\xr*0.75 cm, rounded corners, draw, lightgray, label=-90:{$V_2$}]{};
    \node (b2) at (8*\xr,0*\yr)[xnode,label=-90:{$b_2$}]{};
    \foreach \x in {1,2,x,n}{\draw[xedge] (a2) to (v2\x) to (b2);}
    \node (dots) at (9*\xr,0*\yr)[]{$\cdots$};
    \node (dotsx) at (9.75*\xr,0*\yr)[]{$\cdots$};
    \node (up1) at (10*\xr,1.5*\yr)[]{};
    \node (up2) at (10*\xr,0.5*\yr)[]{};
    \node (upx) at (10*\xr,-0.5*\yr)[]{};
    \node (upn) at (10*\xr,-1.5*\yr)[]{};
    \node (bk) at (11*\xr,0*\yr)[xnode,label=-90:{$b_k$}]{};
    \foreach \x in {1,2,x,n}{\draw[xedge] (bk) to (up\x);}
    \node (ck) at (12*\xr,0*\yr)[xnode,label=-90:{$c_k^1$}]{};
    \node (t) at (13*\xr,0*\yr)[xnode,label=-90:{$t$}]{};
    \draw[xedge] (b2) to (dots); 
    \draw[xedge] (bk) to (ck) to (t);
    \end{scope}

    \begin{scope}[yshift=\yr*-5cm]
    \node at (-0.5*\xr,2*\yr)[]{(b)};
    \node (s) at (0,0)[xnode,label=-90:{$s$}]{};
    \node (ck) at (1*\xr,0)[xnode,label=-90:{$c_k^2$}]{};
    \node (bi) at (2*\xr,0)[xnode,label=-90:{$b_i$}]{};
    \draw[xedge] (s) to (ck) to (bi);
    \node (v11) at (3*\xr,1.5*\yr)[xnode,label=-90:{$v$}]{};
    \node (v12) at (3*\xr,0.5*\yr)[xnode]{};
    \node (v1x) at (3*\xr,-0.5*\yr)[]{$\vdots$};
    \node (v1n) at (3*\xr,-1.5*\yr)[xnode]{};
    \node (aj) at (4*\xr,1*\yr)[xnode,label=-90:{$a_j$}]{};
    \node (bj) at (4*\xr,-1*\yr)[xnode,label=-90:{$b_j$}]{};
    \foreach \x in {2,x,n}{\draw[xedge] (bi) to (v1\x) to (bj);}
    \foreach \x in {1}{\draw[xedge] (bi) to (v1\x) to (aj);}
    \node at (3*\xr,0*\yr)[minimum height=\yr*3.5 cm, minimum width=\xr*0.75 cm, rounded corners, draw, lightgray, label=-90:{$V_i$}]{};
    \node (c0) at (7*\xr,0*\yr)[xnode,label=-90:{$c_0^2$}]{};
    \node (ai) at (6*\xr,0*\yr)[xnode,label=-90:{$a_i$}]{};
    \node (v21) at (5*\xr,1.5*\yr)[xnode,label=0:{$w$}]{};
    \node (v22) at (5*\xr,0.5*\yr)[xnode]{};
    \node (v2x) at (5*\xr,-0.5*\yr)[]{$\vdots$};
    \node (v2n) at (5*\xr,-1.5*\yr)[xnode]{};
    \node at (5*\xr,0*\yr)[minimum height=\yr*3.5 cm, minimum width=\xr*0.75 cm, rounded corners, draw, lightgray, label=-90:{$V_j$}]{};
    \node (api1) at (8*\xr,0*\yr)[xnode,label=-90:{$a_{\pi^{-1}(1)}$}]{};
    \foreach \x in {1,2,x,n}{\draw[xedge] (bj) to (v2\x) to (ai);}
    \draw[xedge] (aj) to (v21) to (ai);
     \foreach \x in {2,x,n}{\draw[dotted, thick] (aj) to (v2\x);}
    \draw[xedge] (ai) to (c0) to (api1);
    \node (dots) at (9.5*\xr,0*\yr)[]{$\cdots$};
    \node (dotsx) at (10.25*\xr,0*\yr)[]{$\cdots$};
    \node (up1) at (9*\xr,1.5*\yr)[]{};
    \node (up2) at (9*\xr,0.5*\yr)[]{};
    \node (upx) at (9*\xr,-0.5*\yr)[]{};
    \node (upn) at (9*\xr,-1.5*\yr)[]{};
    \foreach \x in {1,2,x,n}{\draw[xedge] (api1) to (up\x);}
    \node (ck2) at (11*\xr,0*\yr)[xnode,label=-90:{$c_{k-2}^2$}]{};
    \node (ck1) at (12*\xr,0*\yr)[xnode,label=-90:{$c_{k-1}^2$}]{};
    \node (t) at (13*\xr,0*\yr)[xnode,label=-90:{$t$}]{};
    \draw[xedge] (ck2) to (ck1) to (t);
    \end{scope}
    \end{tikzpicture}
  \caption{Illustration of \cref{constr:emspvcwhard} with (a) showing an odd snapshot and (b) showing the even snapshot~$G_{2\phi(i,v,j)}$  with edge~$\{a_j,w\}$ being present assuming~$\{v,w\}\in E$, and dotted edges may or may not be present (depending on~$E$).}
  \label{fig:emspvcwhard}
 \end{figure}
 Let~$V'$ initially contain~$V_1,\ldots,V_k$ and~$s,t$.
 Moreover,
 $V'$ contains the sets~$A=\{a_1,\dots,a_k\}$ and~$B=\{b_1,\dots,b_k\}$.
 Finally,
 $V'$ contains the sets~$C^1=\{c_i^1\mid 0\leq i\leq k\}$ and $C^2=\{c_i^2\mid 0\leq i\leq k\}$.
 We construct the edge set~$E_{odd}$ as follows.
 It contains the edges~$\{s,c_0^1\}$, $\{c_0^1,a_1\}$, $\{c_k^1,b_k\}$, and~$\{c_k^1,t\}$.
 Moreover,
 it contains the edges~$\{b_i,c_i^1\}$, $\{c_i^1,a_{i+1}\}$ for every~$1\leq i<k$.
 Finally, 
 it contains the edge set~$\bigcup_{v\in V_i} \{\{a_i,v\},\{b_i,v\}\}$ for every~$i\in\set{k}$.
 We set~$E_i\ceq E_{odd}$ for each odd~$i\in\set{\tau}$.
 Next, 
 let~$\phi$ be a bijection that maps each~$(i,v,j)$ to a distinct integer in~$\set{n\frac{(k-1)}{2}}$, where~$i<j$, $i,j\in\set{k}$, $v\in V_i$.
 We construct the edge set~$E_{2\phi(i,v,j)}$ as follows.
 We add the edges~$\{s,c_k^2\}$, $\{c_k^2,b_i\}$.
 Then, 
 $b_i$ is connected with all~$w\in V_i$.
 Next, 
 $v$~is adjacent with~$a_j$,
 and all~$w\in V_i\setminus\{v\}$ are adjacent with~$b_j$.
 Next, 
 $a_j$ is adjacent to a vertex in~$w\in V_j$ if and only if~$\{w,v\}\in E$.
 Vertices~$b_j$ and~$a_i$ are adjacent with all vertices in~$V_j$,
 and vertex~$a_i$ is also adjacent with~$c_0^2$.
 Let~$\pi=\pi^k_{i,j}:\set{k}\setminus\{i,j\}\to \set{k-2}$ (see~\cref{def:bijectionkminus2}).
 Then~$c_0^2$ is adjacent with~$a_{\pi^{-1}(1)}$ and~$c_{k-2}^2$ is adjacent with~$b_{\pi^{-1}(k-2)}$ and with~$c_{k-1}^2$ which in turn is adjacent with~$t$.
 Moreover, 
 for all $p \in \set{k-3}$ the vertex
 $c_p^2$ is adjacent with~$a_{\pi^{-1}(p+1)}$ and $b_{\pi^{-1}(p)}$.
 Finally,
 $a_{\pi^{-1}(p)}$ and $b_{\pi^{-1}(p)}$ are adjacent to all vertices in~$V_{\pi^{-1}(p)}$.
 This finishes the construction of~$E_{2\phi(i,v,j)}$.
 Set~$k'=4k+3$
 and~$\ell=4k+7$.
\end{construction*}

\begin{observation}
 \label{obs:emspvcwhardExact-sep}
	Let $p \in \set{n\cdot(k-1)}$.
 In~$(V,E_{2p-1})$,
 each vertex in~$A\cup B\cup C^1$, and each set~$V_i$ is an~$s$-$t$ separator,
 and in~$(V,E_{2p})$ with~$p=\phi(i,v,j)$
 each vertex in~$(A\setminus\{a_j\})\cup (B\setminus \{b_j\}) \cup C^2$, each set~$V_i$, and the set~$\{a_j,b_j\}$  is an~$s$-$t$ separator.
\end{observation}

\begin{observation}
 \label{obs:emspvcwhardExact}
	Let $p \in \set{n\cdot(k-1)}$.
 Every~$s$-$t$ path in~$(V,E_p)$ with at most~$k'$ vertices contains exactly one vertex from each~$V_i$.
\end{observation}

\begin{proof}
 For every odd snapshot,
 the statement is clear by construction.
 Consider~$p=\phi(i,v,j)$ and~$(V,E_{2p})$,
 and let~$P$ be an arbitrary~$s$-$t$ path with at most~$k'$ vertices.
 We know from \cref{obs:emspvcwhardExact-sep} that every~$s$-$t$ path in $(V,E_{2p})$
 contains every vertex in $(A\setminus\{a_j\})\cup (B\setminus \{b_j\}) \cup C^2$, one vertex from each set~$V_i$, and one vertex from~$\{a_j,b_j\}$.
 It follows that
 $|V(P)|\geq 2+|(A\setminus\{a_j\})\cup (B\setminus \{b_j\}) \cup C^2|+k+1=2+(2k-2+(k+1))+k+1=4k+2$.
 Moreover,
 with the same argument as for the odd snapshots,
 it contains exactly one vertex from each set~$V_q$ with~$q\in\set{k}\setminus\{i,j\}$.
 So, suppose~$P$ contains one more vertex from~$V_i$ or~$V_j$.
 Then~$P$ must contain both~$\{a_j,b_j\}$, 
 and hence~$|V(P)|=(4k+2)+2=4k+4>k'$,
 yielding a contradiction.
\end{proof}

Since in every snapshot
each vertex from~$C^1\cup C^2$ is of degree two or zero,
we have the following.

\begin{observation}
	Let $p \in \set{n\cdot(k-1)}$.
 Every~$s$-$t$ path in~$(V,E_{2p-1})$ contains the edge set~$E_{2p-1}'$ consisting of all edges incident with a vertex in~$\{c_0^1,\dots,c_k^1\}$.
 Every~$s$-$t$ path in~$(V,E_{2p})$ contains the edge set~$E_{2p}'$ consisting of all edges incident with $\{c_0^2,\dots,c_k^2\}$.
 Hence,
 we have that~$\symdif{E_{2p-1}}{E_{2p}}\supseteq E_{2p-1}'\cup E_{2p}'$ and~$|E_{2p-1}'\cup E_{2p}'|=4k+3=\ell-4$,
 and~$\symdif{E_{2p}}{E_{2p+1}}\supseteq E_{2p}'\cup E_{2p+1}'$ and~$|E_{2p}'\cup E_{2p+1}'|=4k+3=\ell-4$.
\end{observation}

\begin{lemma}
 \label{lem:empsvcwhardV}
 Let~$\calP=(P_1,\dots,P_\tau)$ be a solution to the instance obtained using~\cref{constr:emspvcwhard}.
 Then~$V(P_p)\cap V=V(P_q)\cap V$ for all~$p,q\in\set{\tau}$.
\end{lemma}

\begin{proof}
 Assume towards a contradiction that there is~$r=\phi (i,v,j)$ such that~$V(P_{2r-1})\cap V\neq V(P_{2r})\cap V$ or $V(P_{2r})\cap V\neq V(P_{2r+1})\cap V$.
 We consider the first case (the second case is analogous).
 We know that each~$V_x$ is an~$s$-$t$ separator in~$(V,E_p)$ for every $x \in \set{k}$ and~$p\in\set{\tau}$.
 Moreover, 
 we know from \cref{obs:emspvcwhardExact} that each of~$P_{2r-1}$ and~$P_{2r}$ contains exactly one vertex from each~$V_x$, $x \in \set{k}$.
 So, 
 there is a~$z\in \set{k}$ such that there are distinct~$v'$ and~$v''$ in~$V_z$ such that~$v'\in V(P_{2r-1})$ and~$v''\in V(P_{2r})$.
 If~$z\not\in\{i,j\}$, 
 then~$\{v',a_z\},\{v',b_z\},\{v'',a_z\},\{v'',b_z\} \in \symdif{E_{2p-1}}{E_{2p}}$.
 If~$z=i$,
 then~$\{v',b_z\},\{v'',b_z\} \in  \symdif{E_{2p-1}}{E_{2p}}$.
 If~$z=j$,
 then~$\{v',a_j\},\{v',b_j\}\in\symdif{E_{2p-1}}{E_{2p}}$.
 Let~$u\in V(P_{2r})\cap V_i$ and let~$w\in V(P_{2r})\cap V_j$.
 By construction,
 we know that~$\{u,a_j\},\{w,a_i\},\{u,a_i\},\{w,b_j\}\in \symdif{E_{2p-1}}{E_{2p}}$.
 Hence,
 $\symdif{E_{2p-1}}{E_{2p}}$ contains~$\ell-4$ edges each being incident with a vertex in~$C^1\cup C^2$,
 and at least six further edges,
 amounting to~$\ell+2$ edges,
 contradicting the fact that~$\calP$ is a solution.
\end{proof}

\begin{proof}[Proof of~\cref{prop:emspvcwhard}]
  Let~$I=(G=(V_1,\dots,V_k,E)$ be an instance of \prob{Multicolored Clique},
  and let~$I'=(\TG=(V',E_1,\dots,E_\tau),s,t,k,\ell)$ be the instance obtained from~$I$ using~\cref{constr:emspvcwhard} in polynomial time.
  Note that every edge in~$\bigcup_{p=1}^\tau E_p$ is incident with~$M\ceq A\cup B\cup C^1\cup C^2\cup \{s,t\}$,
  and hence~$M$ is a vertex cover of the underlying graph of size~$|M|=2k+2k+2+3=4k+5$.
  Denote by~$G_p=(V,E_p)$ the~\ith{$p$} snapshot of~$\TG$ for every~$p\in\set{\tau}$.
  We claim that~$I$ is a \yes-instance if and only if~$I'$ is a \yes-instance.
  
  \RD{}
  Let~$W\subseteq V$ form a multicolored clique.
  Let~$P_{odd}$ be the path in~$G_{odd}\ceq (V,E_{odd})$ with vertex set~$V(P_{odd})=A\cup B\cup C^1\cup \{s,t\}\cup W$,
  and the edge set~$E(P_{odd})=E(G_{odd}[V(P_{odd})])$.
  Note that~$|V(P_{odd})|=3k+1+2+k=4k+3= k'$.
  Set~$P_{2p-1}\ceq P_{odd}$ for every~$p\in\set{\tau/2}$.
  Next we construct~$P_{2p}$ for every~$p\in\set{\tau/2}$. 
  Let~$p=\phi(i,v,j)$.
  We distinguish two cases whether~$v\in W$ or not.
  
  Case 1: $v\in W$.
  Let~$V(P_{2p})=A\cup B\setminus\{b_j\}\cup C^2\cup \{s,t\}\cup W$,
  and~$E(P_{2p})=E(G_{2p}[V(P_{2p})])$.
  Note that~$|V(P_{2p})|=4k+2\leq k'$.
  Moreover,
  $P_{2p}$ is an~$s$-$t$ path
  since the edges~$\{v,a_j\},\{a_j,w\}$ are contained in~$G_{2p}$,
  where~$w\in W\cap V_j$,
  since~$\{v,w\}\in E$.
  
  Case 2: $v\not\in W$.
  Let~$V(P_{2p})=A\cup B\setminus\{a_j\}\cup C^2\cup \{s,t\}\cup W$,
  and~$E(P_{2p})=E(G_{2p}[V(P_{2p})])$.
  Note that~$|V(P_{2p})|=4k+2\leq k'$.
  Moreover,
  $P_{2p}$ is an~$s$-$t$ path
  since the edges~$\{u,b_j\},\{b_j,w\}$ are contained in~$G_{2p}$,
  where~$u\in W\cap V_i$ and~$w\in W\cap V_j$
  since $b_j$ is adjacent to every vertex in~$V_i\setminus\{v\}$ and~$V_j$.
  
  It remains to show that~$|\symdif{E(P_{2p-1})}{E(P_{2p})}|\leq \ell$ for all~$p\in\set{\tau/2}$,
  and that $|\symdif{E(P_{2p})}{E(P_{2p+1})}|\leq \ell$ for all~$p\in\set{\tau/2-1}$.
  We prove the former, 
  as the latter follows analogously.
  Let~$p=\phi(i,v,j)$.
  By construction,
  $\symdif{E(P_{2p-1})}{E(P_{2p})}$ contains all edges incident with~$C^1$ and~$C^2$.
  Let~$u\in V_i\cap W$, and~$w\in V_j\cap W$.
  We consider two cases:
  
  Case 1: $u=v$.
  Note that~$P_{2p}$ has the subpath~$b_iua_jwa_i$,
  and hence
  $\symdif{E(P_{2p-1})}{E(P_{2p})}$ contains the edges~$\{u,a_j\},\{w,a_i\}\in E(P_{2p})$ and the edges $\{u,a_i\},\{w,b_j\} \in E(P_{2p-1})$.
  Note that all other edges in $E(P_{2p-1}) \cup E(P_{2p})$ not incident to a vertex in $C^1 \cup C^2$ are also in $E(P_{2p-1}) \cap E(P_{2p})$.
  Hence,~$|\symdif{E(P_{2p-1})}{E(P_{2p})}|=2(k+1)+2(k+1)-1+4=4k+7=\ell$.
  
  Case 2: $u\neq v$.
  Note that~$P_{2p}$ has the subpath~$b_iub_jwa_i$,
  and hence
  $\symdif{E(P_{2p-1})}{E(P_{2p})}$ contains the edges~$\{u,b_j\},\{w,a_i\}\in E(P_{2p})$ and the edges $\{u,a_i\},\{w,a_j\} \in E(P_{2p-1})$.
  Note that all other edges in $E(P_{2p-1}) \cup E(P_{2p})$ not incident to a vertex in $C^1 \cup C^2$ are also in $E(P_{2p-1}) \cap E(P_{2p})$.
  Hence, 
  $|\symdif{E(P_{2p-1})}{E(P_{2p})}|=2(k+1)+2(k+1)-1+4=4k+7=\ell$.
  
  It follows that~$(P_1,\dots,P_\tau)$ is a solution to~$I'$.
  
  \LD{}
  Let~$(P_1,\ldots,P_\tau)$ be a solution to~$I'$.
  Due to~\cref{lem:empsvcwhardV},
  we know that~$V(P_p)\cap V=V(P_q)\cap V=:W$ for all~$p,q\in\set{\tau}$.
  We claim that~$W$ forms a multicolored clique in~$G$.
  By~\cref{obs:emspvcwhardExact},
  we know that $|W\cap V_i|=1$, for all~$i\in\set{k}$.
  Let~$w_i\in W\cap V_i$ denote the corresponding vertex, for all~$i\in\set{k}$.
  It remains to show that for each distinct pair~$w_i,w_j$,
  we have that~$\{w_i,w_j\}\in E$.
  Assume without loss of generality that~$i<j$,
  and let~$p=\phi(i,w_i,j)$.
  Since~$P_{2p}$ is an~$s$-$t$ in~$G_{2p}$,
  it contains the subpath~$w_ia_jw_j$,
  since~$w_i$ is only adjacent to~$b_i$ and~$a_j$.
  By construction of snapshot~$G_{2p}$,
  we know that~$\{a_j,w_j\}\in E(G_{2p})$ if and only if~$\{w_i,w_j\}\in E$.
  Hence,
  the claim follows.
\end{proof}
}

\noindent
For~\vmsp{}, 
we have an even stronger result:
the problem is~\W{1}-hard regarding~$\ug{\nu}$
even if the size of any symmetric difference of the vertex sets of consecutive paths is at most one.
The proof is,
however,
similar to the proof of~\cref{prop:emspvcwhard}.

\ifarxiv{}
\begin{proposition}
		\else{}
\begin{proposition}[\appref{prop:vmspvcwhard}]
\fi{}
 \label{prop:vmspvcwhard}
 \vmsp{} when parameterized by~$\ug{\nu}$ is \W{1}-hard,
 even if~$\ell=1$.
\end{proposition}

\appendixproof{prop:vmspvcwhard}
{
\begin{construction*}
 \label{constr:vmspvcwhard}
 Let~$(G=(V_1,\dots,V_k,E))$ be an instance of \prob{Multicolored Clique}.
 We construct a temporal graph~$\TG=(V',E_1,\ldots,E_\tau)$ with $\tau=n\cdot (k-1)$ as follows.
 Let~$V'$ initially contain~$V_1,\ldots,V_k$ and~$s,t$.
 Finally,
 $V'$ contains the sets~$A=\{a_0,\dots,a_k\}$ and two special vertices~$x$ and~$y$.
 We construct the edge set~$E_{odd}$ as follows.
 It contains the edges~$\{s,a_0\}$ and~$\{a_k,t\}$.
 Finally, 
 it contains the edge set~$\bigcup_{v\in V_i} \{\{a_{i-1},v\},\{a_i,v\}\}$ for every~$i\in\set{k}$.
 We set~$E_i\ceq E_{odd}$ for each odd~$i\in\set{\tau}$.
 Next, 
 let~$\phi$ be a bijection that maps~$(i,v,j)$ to~$\set{n\frac{(k-1)}{2}}$, where~$i<j$, $i,j\in\set{k}$, $v\in V_i$.
 We construct the edge set~$E_{2\phi(i,v,j)}$ as follows.
 We add the edge~$\{s,a_i\}$.
 Then, 
 $a_i$ is connected with all~$w\in V_i$.
 Next, 
 $v$ is adjacent with~$x$,
 and all~$w\in V_i\setminus\{v\}$ are adjacent with~$y$.
 Next, 
 $x$ is adjacent to a vertex in~$w\in V_j$ if and only if~$\{w,v\}\in E$.
 Vertices~$y$ and~$a_j$ are adjacent with all vertices in~$V_j$,
 and vertex~$a_j$ is also adjacent with~$a_{\pi^{-1}(1)}$,
 where~$\pi=\pi^k_{i,j}:\set{k}\setminus\{i,j\}\to \set{k-2}$ (see~\cref{def:bijectionkminus2}).
 Then~$t$ is adjacent with~$a_0$ which in turn is also adjacent with~$a_{\pi^{-1}(k-2)}$,
 and
 for each~$p\in\set{k-3}$,
 $a_{\pi^{-1}(p)}$ and $a_{\pi^{-1}(p+1)}$ are adjacent to all vertices in~$V_{\pi^{-1}(p)}$.
 This finishes the construction of~$E_{2\phi(i,v,j)}$.
 Set~$k'=2k+4$
 and~$\ell=1$.
\end{construction*}

\begin{observation}
 In~$(V,E_{2p-1})$,
 each vertex in~$A$, and each set~$V_i$ is an~$s$-$t$ separator,
 and in~$(V,E_{2p})$ with~$p=\phi(i,v,j)$
 each vertex in~$A$, each set~$V_i$, and the set~$\{x,y\}$ is an~$s$-$t$ separator.
\end{observation}

We know that each~$s$-$t$ path in an even snapshot contains~$s$ and~$t$, 
and~$k+1$ vertices from~$A$,
and one of~$x$ and~$y$,
leaving $k$~vertices.
Since each~$V_i$ forms an~$s$-$t$ separator,
we have the following.

\begin{observation}
 \label{obs:vmspvcwhardExact}
 Every~$s$-$t$ path in~$(V,E_p)$ with at most~$k'$ vertices contains exactly one vertex from each~$V_i$.
\end{observation}

\begin{proof}[Proof of~\cref{prop:vmspvcwhard}]
 Let~$I=(G=(V_1,\dots,V_k,E)))$ be an instance of \prob{Multicolored Clique},
  and let~$I'=(\TG=(V',E_1,\dots,E_\tau),s,t,k,\ell)$ be the instance obtained from~$I$ using~\cref{constr:vmspvcwhard} in polynomial time.
  Note that every edge in~$\bigcup_{p=1}^\tau E_p$ is incident with~$M\ceq A\cup \{x,y\}\cup \{s,t\}$,
  and hence~$M$ is a vertex cover of the underlying graph of size~$|M|=k+5$.
  Denote by~$G_p=(V,E_p)$ the~\ith{$p$} snapshot of~$\TG$ for every~$p\in\set{\tau}$.
  We claim that~$I$ is a \yes-instance if and only if~$I'$ is a \yes-instance.
  
  \RD{}
  Let~$W\subseteq V$ be a multicolored clique.
  Define~$P_{odd}$ as the path in~$G_{odd}=(V,E_{odd})$ with vertex set~$V(P_{odd})=\{s,t\} \cup A \cup W$ and edge set~$E(G_{odd}[V(P_{odd})])$.
  Note that~$P_{odd}$ is an~$s$-$t$ path with~$2k+3$ vertices.
  Set~$P_{2p-1}\ceq P_{odd}$.
  For~$P_{2p}$ with~$p=\phi(i,v,j)$,
  we set~
  $V(P_{2p})=V(P_{odd}) \cup \begin{cases} \{x\},& \text{if $v\in W$} \\\{y\},& \text{otherwise,}\end{cases}$
  and~$E(P_{2p})=E(G_{2p}[V(P_{2p})])$.
  Note that~$P_{2p}$ is an~$s$-$t$ path,
  since if~$v\in W$, then
  the edge~$\{x,w\}$ with~$w\in W\cap V_j$ exists.
  Moreover,
  $|V(P_{2p})|=2k+4$,
  and by construction we have that~$|\symdif{V(P_p)}{V(P_{p+1}})|=1$ for all~$p\in\set{\tau-1}$.
  
  \LD{}
  Let~$(P_1,\ldots,P_\tau)$ be a solution to~$I'$.
  Due to~\cref{obs:vmspvcwhardExact},
  we know that each~$P_i$ contains exactly one vertex from~$V_i$.
  In fact, 
  it holds true that~$V(P_i)\cap V=V(P_j)\cap V$ for all~$i,j\in \set{\tau}$:
  Suppose not,
  that is,
  there is an~$i\in\set{\tau-1}$ such that~$w\in V\cap (\symdif{V(P_i)}{V(P_{i+1})})$.
  In both cases ($w\in V(P_i)\setminus V(P_{i+1})$ or $w\in V(P_{i+1})\setminus V(P_i)$)
  we get a contradiction to~\cref{obs:vmspvcwhardExact}.
  Let~$W\ceq V\cap V(P_1)$.
  We claim that~$W$ is a multicolored clique in~$G$.
  Let~$v\in V_i\cap W$ and~$w\in V_j\cap W$ with~$i,j\in\set{k}$, 
  $i<j$, be arbitrary but fixed.
  Then,
  path~$P_{2\phi(i,v,j)}$ contains the subpath~$(v,x,w)$,
  proving that~$\{v,w\}\in E$.
  It follows that~$W$ is a multicolored clique in~$G$.
\end{proof}
}

\noindent
We will see in the next section that a similar result for~\eimsp{} or \vimsp{} is unlikely.

\subsection{Fixed-parameter tractability for dissimilarity variant \texorpdfstring{regarding~\boldmath$k$}{}}

In stark contrast to \cref{thm:w1hard,cor:vmspw1hardktau}, we show in this section that \vimsp{} and \eimsp{} can be solved in linear time for constant 
path lengths; put differently, they 
are fixed-parameter tractable 
when parameterized by path length~$k-1$. %

\begin{theorem}
 \label{thm:eimspfptk}
 \vimsp{} and \eimsp{} can be solved in $2^{O(k)}\cdot |\TG|$ time. 
\end{theorem}
We defer the proof of \cref{thm:eimspfptk} towards the end of this section and, moreover, only describe the algorithm for \vimsp{}.
In a nutshell, the algorithm behind \cref{thm:eimspfptk} computes for each snapshot \emph{sufficiently many}
$s$-$t$ paths such that no matter which vertices are used in 
the snapshots beforehand and afterwards, 
one of these $s$-$t$ paths has a small intersection with these vertices.
To this end, we introduce $q$-robust sets\footnote{In a nutshell, 
$q$-robust sets are $q$-representative families \cite{monien1985find},
just explicitly coined to $s$-$t$ paths of length at most~$k$.
This notion shall avoid confusion with the later defined $q$-representatives of independent sets.}
of $s$-$t$ paths.
\begin{definition}
	Let $G=(V,E)$ be a graph, $s,t \in V$ two distinct vertices, 
	$\mathcal F$ be a set of $s$-$t$ paths of length at most~$k-1$,
	and $q \in N_0$.
	We call $\mathcal F$ \emph{$q$-robust} if for each 
	set $X \subseteq (V(G) \setminus \{s,t\})$ 
	of size at most $q$ the following holds:
	if there is an $s$-$t$ path in $G - X$ of length at most $k-1$,
	then there is an $s$-$t$ path $P \in \mathcal F$ which is an $s$-$t$ path in $G - X$.
\end{definition}
To find a solution,
it is sufficient to have a $2(k-\ell)$-robust set of $s$-$t$ paths of length at most~$k-1$
for each snapshot of the temporal graph:

\ifarxiv{}
\begin{lemma}
		\else{}
\begin{lemma}[\appref{lem:q-robust-correct}]
\fi{}
	\label{lem:q-robust-correct}
	Let $I=(\TGcompact,s,t,k,\ell)$ be an instance of \vimsp{} and
	$\mathcal F_i$ be a $2(k-\ell)$-robust set of $s$-$t$ paths of length at most~$k-1$ in~$G_i=(V,E_i)$, for all $i \in \set{\tau}$.
	Then, $I$ is a \yes-instance 
	if and only if 
	there is a solution $(P_1,\dots,P_\tau)$ such that $P_i \in \mathcal F_i$, 
	for all $i \in \set{\tau}$.
\end{lemma}

\appendixproof{lem:q-robust-correct}
{
\begin{proof}
	Since the converse is trivially true, we only show that
	if $I$ is a \yes-instance, then
	there is a solution~$(P_1,\dots,P_\tau)$ for~$I$ such that for all~$i \in \set{\tau}$ 
	we have $P_i \in \mathcal F_i$.

	For all $p \in \set{\tau+1}$, let $\mathcal S_p$ be the set of solutions for $I$ 
	such that for all $j < p$ we have~$P_j \in \mathcal F_j$.
	Let $i \ceq  \max \{p \in \set{\tau+1} \mid \mathcal S_p \not= \emptyset \}$.
	If $i=\tau+1$, then we are done. 
	Hence, assume towards a contradiction that $i \leq \tau$.

	(Case 1): Suppose $1<i<\tau$.
	Let $X_1=V(P_{i-1}) \setminus V(P_i)$ and $X_2=V(P_{i+1}) \setminus V(P_i)$.
	If $X \in \{X_1,X_2\}$ is larger than $k-\ell$, then remove arbitrary vertices from $X$ such that $|X| = k-\ell$.
	Note that $|V(P_{i-1}) \setminus X_1| \leq \ell$ and $|V(P_{i+1}) \setminus X_2| \leq \ell$.
	Observe that $P_i$ is an $s$-$t$ path of length at most $k-1$ in $G_i-(X_1 \cup X_2)$.
	Since $\mathcal F_i$ is $2(k-\ell)$-robust, there is an $s$-$t$ path $P \in \mathcal F_i$ of length at most $k-1$
	in $G_i - (X_1 \cup X_2)$, see \cref{fig:q-robust} for an illustration.
	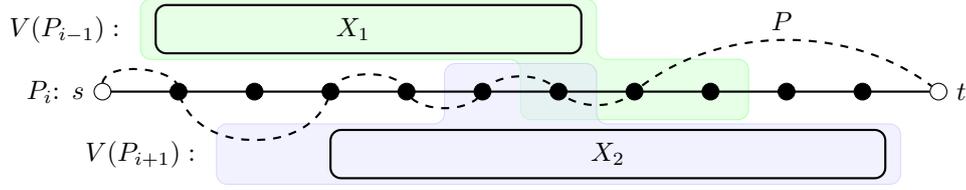
\begin{figure}
		\centering
		\begin{tikzpicture}[yscale=0.8]
    \tikzstyle{xnode}=[circle, draw,scale=2/3]
    \tikzstyle{xedge}=[thick]
	\begin{scope}[yshift=1]
	\filldraw[fill=green!10!white,draw=green!30,rounded corners] 
			(1,0.5) -- (1,1.5) -- (7,1.5) -- (7,0.5) -- (9,0.5) -- (9,-0.5) -- (6,-0.5) -- (6,0.5) -- cycle;
	\draw[thick,rounded corners] 
			(1.2,0.6) --  (6.8,0.6) -- 
			(6.8,1.4)--  (1.2,1.4) --  cycle;
	\node at (1+2.8,1) {$X_1$};
	\node at (0,1) {$V(P_{i-1}):$};
	\end{scope}
	\begin{scope}[yshift=-1]
	\filldraw[fill=blue!10!white,draw=blue!30,rounded corners,opacity=0.5] 
			(2,-0.5) -- (5,-0.5) -- (5,0.5) -- (7,0.5) -- (7,-0.5) -- (11,-0.5) -- 
			(11,-1.5)--  (2,-1.5) --  cycle;
	\draw[thick,rounded corners] 
			(3.5,-0.6) --  (10.8,-0.6) -- 
			(10.8,-1.4)--  (3.5,-1.4) --  cycle;
	\node at (3.5+3.65,-1) {$X_2$};
	\node at (1,-1) {$V(P_{i+1}):$};
	\end{scope}
    \node (s) at (0.5,0)[xnode,label=left:{$P_i$:  $s$}]{};
    \node (z) at (11.5,0)[xnode,label=right:{$t$}]{};
	\node at (1.5,0)[xnode,fill]{};
	\node at (2.5,0)[xnode,fill]{};
	\node at (3.5,0)[xnode,fill]{};
	\node at (4.5,0)[xnode,fill]{};
	\node at (5.5,0)[xnode,fill]{};
	\node at (6.5,0)[xnode,fill]{};
	\node at (7.5,0)[xnode,fill]{};
	\node at (8.5,0)[xnode,fill]{};
	\node at (9.5,0)[xnode,fill]{};
	\node at (10.5,0)[xnode,fill]{};
	\draw[xedge] (s) to (z);
	\draw[xedge,dashed] (s) 
			to[bend left=90] (1.5,0) 
			to[bend right=40] (2.5,-0.8) 
			to[bend right=40] (3.5,0) 
			to[bend left=80] (4.5,0) 
			to[bend right=70] (5.5,0) 
			to[bend left=60] (6.5,0) 
			to[bend right=50] (7.5,0) 
			to[bend left=45] node[above] {$P$} (z) ;
	\end{tikzpicture}	
	\caption{Illustration of case 1 in the proof of \cref{lem:q-robust-correct}, where $|V(P_{i+1}) \setminus V(P_i)| > k-\ell$.}
	\label{fig:q-robust}
    \end{figure}   
	Hence, $|V(P) \cap V(P_{i-1})| \leq |V(P) \cap (V(P_{i-1})\setminus X_1)| \leq \ell$ and
	$|V(P) \cap V(P_{i+1})| \leq |V(P) \cap (V(P_{i+1})\setminus X_2)| \leq \ell$.
	Thus, $S = (P_1,\dots,P_{i-1},P,P_{i+1},\dots,P_\tau)$ is a solution for $I$.
	This contradicts $i$ being maximal.

	(Case 2): If $i=1$ ($i=\tau$), then we set $X_1=\emptyset$ ($X_2=\emptyset)$ and conclude analogously to Case 1 
	that $i$ is not maximized.
\end{proof}
}
The main tool of our algorithm is a fast (``linear-time FPT'') computation of small sets of $s$-$t$ paths of length at most~$k-1$ which are $q$-robust.
We believe that such a use of representative families may become a general algorithmic tool being potentially helpful for other multistage problems.
Formally, we show the following.
\ifarxiv{}
\begin{lemma}
		\else{}
\begin{lemma}[\appref{lem:q-robust}]
\fi{}
	\label{lem:q-robust}
	Let $G=(V,E)$ be a graph with two distinct vertices $s,t \in V$, 
	and $k,q \in \mathbb N_0$.
	We can compute, in $2^{O(k+q)} \cdot|E|$ time, 
	a $q$-robust set $\mathcal F$ of $s$-$t$ paths of length at most~$k-1$ 
	such that~$|\mathcal F| \le 2^{q+k}$.
\end{lemma}
In order to prove \cref{lem:q-robust}, 
we extend the ``representative-family-based'' algorithm 
for \textsc{$k$-Path} 
of Fomin~\etal~\cite{FLPS16} 
such that 
we can find $s$-$t$ paths 
avoiding a size-at-most-$q$ set of vertices\ifarxiv{}.\else{}, 
see Appendix~\ref{proof:lem:q-robust}.\fi{}%
\appendixproof{lem:q-robust}
{
The proof of \cref{lem:q-robust} is deferred to the end of this section.

We use standard terminology from matroid theory~\cite{Oxl92}.
A pair $(U,\mathcal I)$, where $U$~is the \emph{ground set}
and $\mathcal I\subseteq 2^U$ is a family of \emph{independent sets},
is a \emph{matroid} if the following holds:
\begin{itemize}
	\item  $\emptyset \in \mathcal I$;
 	\item  if $A' \subseteq A$ and $A \in \mathcal I$, then $A' \in \mathcal I$;
	\item  if $A,B \in \mathcal I$ and $|A| < |B|$, then there is an $x \in B \setminus A$ such that $A \cup \{x\} \in \mathcal I$.
\end{itemize}
An inclusion-wise maximal independent set~$A\in \mathcal I$
of a matroid~$M=(U,\mathcal I)$ is a \emph{basis}.
The cardinality of the bases of~$M$
is called the \emph{rank} of~$M$.
The \emph{uniform matroid of rank~$r$} on $U$
is the matroid~$(U,\mathcal I)$
with $\mathcal I=\{S\subseteq U\mid |S|\leq r\}$.
A matroid~$(U,I)$ is \emph{linear} or \emph{representable over a field $\mathbb F$}
if there is a matrix~$A$ with entries in $\mathbb F$ and the columns labeled by the elements of~$U$
such that $S \in \mathcal I$ if and only if the columns of~$A$ 
with labels in~$S$ are linearly independent over~$\mathbb F$. The matrix $A$ is called a \textit{representation}
of $(U,I)$.

\begin{definition}[$q$-representative family of independent sets]\label[definition]{def:qrep}
  Given a matroid $(U,\mathcal I)$,
  a~family~$\mathcal S \subseteq \mathcal I$ of independent sets,
  we say that a subfamily~$\widehat{\mathcal S} \subseteq \mathcal S$
  is a \emph{$q$-representative
  of~$\mathcal S$} if
  for each set~$Y \subseteq U$ of size at most~$q$
  it holds that
  if there is a set~$X \in \mathcal S$  with
  $X \uplus Y \in \mathcal I$,  then there is a
  set~$\widehat X \in \widehat{\mathcal S}$ such that  $\widehat X \uplus Y \in \mathcal I$.
\end{definition}
We are only interested in uniform matroids, hence, to simplify matters we reformulate the definition of representative families.
\begin{definition}[$q$-representative family]
	\label{def:repsets}
	Let $\mathcal S = \{S_1,\dots,S_t\}$ be a family of sets of size~$p$ over a universe~$U$.
	A subfamily $\widehat{\mathcal S} \subseteq \mathcal S$
	is a $q$-representative of $\mathcal S$ 
	if for every set $Y \subseteq U$ of size at most~$q$ it holds that 
	if there is a set~$X \in \mathcal S$ disjoint from~$Y$,
	then there is a set $\widehat{X} \in \widehat{\mathcal S}$ disjoint from $Y$.
\end{definition}

For linear matroids, there are fixed-parameter algorithms parameterized by rank that compute
small representatives for large families of independent sets.
\begin{lemma}[Fomin~\etal~{\cite[Theorem 1.1]{FLPS16}}]
		\label{lem:matroid-eff-rep-set}
	Let $M = (U,\mathcal I)$ be a linear matroid of rank $p+q$ 
	given together with its representation matrix $A_M$ over a field $\mathbb F$.
	Let $\mathcal S = \{S_1,\dots,S_t\}$ be a family of independents sets of $M$ of size $p$. 
	For a given $q$, a $q$-representative family $\widehat{\mathcal S} \subseteq \mathcal S$
	of size ${p+q} \choose p$ can be computed in $O\left( {{p+q} \choose p} tp^\omega + t{{p+q} \choose p}^{\omega-1}\right)$ time.
	Here, $\omega < 2.373$ is the matrix multiplication exponent.
\end{lemma}

\begin{lemma}
		\label{lem:fast-rep}
		Given a set $U$ and an integer~$r$,
		we can compute in $O(r\cdot|U|)$ time
		a representation~$A$ %
		of the uniform matroid of rank $r$ on $U$, where $p \in O(|U|)$
		and $A$ is over a prime field~$\mathbb F_p$.
\end{lemma}
\begin{proof}
		A Vandermonde matrix of size $(p+q) \times |U|$ in a field with at least $|U|$ distinct elements 
		suffices as representation of the uniform matroid of rank $r$ on $U$ \cite[Section 3.4]{Mar09}.

		Let $p \in \set[|U|]{2|U|}$ be a prime number.
		Such a prime exists by the folklore Bertrand-Chebyshev theorem and 
		can be computed in $O(|U|^{1/2 + o(1)}) \leq O(|U|)$ time
		using the Lagarias-Odlyzko method \cite{TCH12}.
		Observe that we can perform a primitive operation in the prime field~$\mathbb F_p$ 
		by first performing the operation in $\mathbb Z$ and then taking the result modulo~$p$.
		Since we only need~$O(\log |U|)$ many bits to store
		one element of $\mathbb F_p$, 
		each element of $\mathbb F_p$ fits into one memory cell 
		of the Word RAM computation model.
		Hence, we can perform a primitive operation over $\mathbb F_p$ in constant time.

		Finally, we can compute the 
		Vandermonde matrix of size $r \times |U|$ in $O(r\cdot |U|)$ time, 
		because each entry is either $1$ 
		or an elementary element of $\mathbb F_p$
		or can be compute by one multiplication
		from another entry calculated earlier.
\end{proof}

In a nutshell, we extend the representative family based algorithm of 
Fomin~\etal~\cite{FLPS16} for \textsc{$k$-Path} such 
that we find $s$-$t$ paths 
which can avoid a set of vertices of size at most~$q$.
\begin{algorithm*}
	\label{repdp}
	Let $G=(V,E)$ be a graph with two distinct vertices $s,t \in V$, 
	and $k,q \in \mathbb N_0$.
	Define $\mathcal N^i_v$ to be a $(q+k-i)$-representative of 
	the family of all sets $A\subseteq V$ such that 
	there is an $s$-$v$ path~$P$ in~$G$ of length~$i-1$ with $V(P) = A$.

	Our goal is to compute $\mathcal N^{k}_t$, as we will construct the desired $q$-robust set of $s$-$t$ paths from it later on.
	We start by setting $\mathcal N_s^1 \ceq  \{ s \}$ and $\mathcal N_v^1 \ceq  \emptyset$ for all $v \in V' \setminus \{s\}$.
	Then, we compute for all $i \in \set[2]{k}$ (in ascending order)  
	\begin{align}
		\mathcal T^i_v \ceq  \bigcup_{\{v,w\} \in E'} \bigcup_{X \in \mathcal N_w^{i-1}: v\not \in X} (X \cup \{v\}).
	\end{align}
	Then (using \cref{lem:matroid-eff-rep-set}) we compute a $(q+k-i)$-representative $\mathcal N^i_v$ of $\mathcal T^i_v$.
\end{algorithm*}

\begin{lemma}
	\label{lem:repdp-correct}
	For all $i \in \set{k}$, the family $\mathcal N^i_v$ (from \cref{repdp}) is of size at most ${q+k-i} \choose i$ and a $(q+k-i)$-representative of 
the family of all sets $A\subseteq V$ such that 
there is an $s$-$v$ path $P$ in $G$ of length $i-1$ with $V(P) = A$.
\end{lemma}
\begin{proof}
	We will prove this claim by induction.  
	Observe that $\mathcal N_v^1$ is correctly computed for all~$v \in V$. %
	Now assume that for all $j < i \leq k$ 
	the family $\mathcal N^j_v$ is of size at most ${q+k-j} \choose j$ and $\mathcal N^j_v$ is a $(q+k-j)$-representative of 
	the family of all sets $A\subseteq V$ such that 
	there is an $s$-$v$ path $P$ in $G$ of length $j-1$ with $V(P) = A$.

	Let $Y \subseteq V'$ be a set of size at most $(q+k-i)$ and $v \in V$.
	Assume there is an $s$-$v$ path~$P$ of length~$i-1$ such that $Y \cap V(P) = \emptyset$.
	Let $w\in V(P)$ be the vertex which is visited by~$P$ directly before $v$ (starting from $s$).
	Let $P'$ be the $s$-$w$ path of length $i-2$ induced by $P$ without $v$.
	Since $(Y \cup \{ v \}) \cap V(P') = \emptyset$ and $Y \cup \{ v \}$ is a set of size $q+k-(i-1)$,
	we know, by induction hypothesis,  that there is an $A \in \mathcal N^{i-1}_w$ 
	and an $s$-$w$ path $P''$ of length $i-2$ with $V(P'') = A$ and $(Y \cup \{ v \}) \cap V(P'') = \emptyset$.
	Hence, by \cref{repdp}, $V(P'') \cup \{v\} \in \mathcal T^i_v$.
	Since $Y \cap V(P'')=\emptyset$ and $\mathcal N^i_v$ is an $(q+k-i)$-representative of $\mathcal T^i_v$,
	we know that $\mathcal N^i_v$ contains a set~$B$ such that there is an $s$-$v$ path~$P'''$ of length~$i-1$ with 
	$V(P''') = B$ and $B \cap Y = \emptyset$. 
	Hence,~$\mathcal N^j_v$ is indeed a $(q+k-i)$-representative of 
the family of all sets $A\subseteq V$ such that 
there is an $s$-$v$ path $P$ in $G$ of length $i-1$ with $V(P) = A$.

	The upper bound on the size of $\mathcal N^i_v$ follows from \cref{lem:matroid-eff-rep-set}.
	This completes the proof.
\end{proof}

\begin{lemma}
	\label{lem:dp-runtime}
	The family $\mathcal N^{k}_t$ from \cref{repdp} can be computed in $2^{O(q+k)}\cdot|E|$ time.
\end{lemma}
\begin{proof}
		As a preprocessing step, we remove in $O(|E|)$ time via breadth-first search all vertices which are not on an $s$-$t$ path. 
		Hence, $|V| \leq |E|$.
		Furthermore, we use \cref{lem:fast-rep} to compute a representation of the uniform matroid $M$ of rank $q+k$ on $V$ in $O(|E|\cdot(q+k))$ time.
		Then, for each $i \in \set{k}$ and each $v \in V$
		we compute $\mathcal T^i_v$ in
		$O(\deg(v) \cdot 2^{q+k})$ time, since for all $w \in V$ the family $\mathcal N^{i-1}_w$ is of size at most $2^{q+k}$, see \cref{lem:matroid-eff-rep-set}.
		Hence, $\mathcal T^i_v$ is of size at most $2^{q+k}\deg(v)$.
		Computing (with \cref{lem:matroid-eff-rep-set}) the $(q+k-i)$-representative~$\mathcal N^i_v$ of~$\mathcal T^i_v$ takes $2^{O(k+q)}\cdot\deg(v)$ time.
		Hence, by the Handshaking Lemma, 
		this yields an overall running time of $2^{O(k+q)}\cdot |E|$ time.
\end{proof}
In the proof of \cref{lem:dp-runtime}, 
one could use Theorem 1.2 instead of Theorem 1.1 from Fomin~\etal~\cite{FLPS16} to improve the constant hidden in the Big-$O$ notation. 
However, 
we would lose the linear dependency in $|E|$ by doing so. 

We are now ready to prove \cref{lem:q-robust}.
\begin{proof}[Proof of \cref{lem:q-robust}]
	First, we construct the graph $G'=(V',E')$ where we add $k$ new dummy vertices $d_1,\dots,d_k$ to $G$.
	Hence,
	$V' \ceq  V \cup \{ d_1,\dots,d_k \}$ and
	\begin{align*}
		E' \ceq  E 	&\cup \{ \{v,w\} \in E \mid t \not \in \{v,w\} \} \\
			&\cup \{ \{v,d_i\} \mid \{v,t\} \in E, i \in \set{k}\} \\
			&\cup \{ \{d_i,d_{i+1}\} \mid i \in \set{k-1} \} \\
			&\cup \{ \{d_i,t\} \mid i \in \set{k}\}.
	\end{align*}
	Note that for each $s$-$t$ path $P$ in $G$ of length at most $k-1$ 
	there is an $s$-$t$ path~$P'$ in~$G'$ of length exactly~$k-1$ 
	such that $V(P) = V(P') \setminus \{d_1,\dots,d_k\}$.
	Furthermore, for each $s$-$t$ path~$P'$ in~$G'$ of length exactly~$k-1$
	there is an $s$-$t$ path~$P$ in~$G$ of length at most~$k-1$ 
	such that $V(P) = V(P') \setminus \{d_1,\dots,d_k\}$.

	Using \cref{repdp}, we compute in $2^{O(q+k)}\cdot|E|$ time (\cref{lem:dp-runtime})
	$\mathcal N^{k+1}_t$ for $G'$, $s$, $t$, $k$, and~$q$.
	By \cref{lem:repdp-correct}, we know 
	that $\mathcal N^{k}_t$  is of size at most ${q+k} \choose {k}$ and a $q$-representative of 
	the family of all sets $A\subseteq V'$ such that 
	there is an $s$-$v$ path~$P$ in~$G$ of length~$k-1$ with $V(P) = A$.

	Now we compute the desired set~$\mathcal F$, which we initialize by
	$\mathcal F \ceq  \emptyset$.
	Observe, that during the execution of \cref{repdp}, 
	we can store for each set $A\in \mathcal T^i_v$ a corresponding $s$-$v$ path $P$ in $G$ with $V(P)=A$, where $i \in \set{k},v \in V'$.
	We now go over all $A \in \mathcal N^{k}_t$ and their corresponding $s$-$t$ paths~$P_A$ of length~$k-1$ in~$G'$.
	Next, we store in $\mathcal F$ an $s$-$t$ path~$P'$ in~$G$ of length at most~$k-1$ such that $V(P') = V(P_A)\setminus \{d_1,\dots,d_k\}$.
	The whole procedure ends after $2^{O(q+k)}\cdot |E|$ time and $\mathcal F$ is of size at most $|\mathcal F| \leq 2^{q+k}$.

	It remains to show that $\mathcal F$ is $q$-robust.
	Let $X \subseteq V$ of size at most~$q$ such that there is an $s$-$t$ path~$P$ of length at most~$k-1$ in $G - X$.
	Hence, there is an $s$-$t$ path~$P'$ in~$G'$ of length exactly~$k$ such that $V(P) = V(P') \setminus \{d_1,\dots,d_k\}$.
	Since $X \cap V(P') = \emptyset$, we know that there is an $A \in \mathcal N^{k}_t$ such that 
	there is an $s$-$v$ path~$P''$ in~$G'$ of length~$k$ with $V(P'') = A$ and $A \cap X = \emptyset$.
	Thus, we added an $s$-$t$ path $P^*$ to $\mathcal F$ with $V(P^*) = A \setminus \{d_1,\dots,d_k\}$.
	Hence, $V(P^*) \cap X$ and it thus is an $s$-$t$ path in~$G - X$.
\end{proof}
}

Having \cref{lem:q-robust-correct,lem:q-robust}, 
we are set to prove \cref{thm:eimspfptk}.
\begin{proof}[Proof of \cref{thm:eimspfptk}]
	We only show the proof for \vimsp{}. 
	The fixed-parameter tractability of \eimsp{} follows from \cref{prop:edgetovertexintersect}.

	Given an instance $I = (\TGcompact,s,t,k,\ell)$ of \vimsp{},
	we first check whether there is an empty $E_i$. 
	If this is the case, then $I$~is a \no-instance.
	Afterwards, we can assume that $\tau \leq |\TG|$.
	For each $i \in \set{\ltime}$, 
	we compute in $2^{O(k+2(k-\ell))}|E_i| =2^{O(k)}|E_i|$ time
	a $2(k-\ell)$-robust set $\mathcal F_i$ of $s$-$t$ paths of length at most~$k-1$ in $G_i = (V,E_i)$
	such that~$|\mathcal F_i| \leq 2^{O(k)}$, see \cref{lem:q-robust}. 
	
	Next, we construct a directed graph $G'=(V',E')$, where beside $s,t$ each path in $\mathcal F_i$ has a corresponding vertex, for all $i \in \set{\tau}$.
	Formally, that is, $ V' \ceq  \{s,t\} \cup \bigcup_{i=1}^\ltime \mathcal F_i$, and
  $E' \ceq\{ (P,P') \mid P \in \mathcal F_i,P' \in \mathcal F_{i+1}, |V(P) \cap V(P')| \leq \ell, \text{ for some } i \in\set{\ltime-1}\}
				\cup \{ (s,P) \mid P \in \mathcal F_1 \}
				\cup \{ (P,t) \mid P \in \mathcal F_{\ltime} \}$.
	Observe that $|V'|+|E'| \leq 2^{O(k)}\cdot\ltime$.
	We note that $I$ is a \yes-instance if and only if there is an $s$-$t$ path in $G'$.
	Since $\sum_{i=1}^\tau |E_i| \leq |\TG|$, this yields an overall running time of $2^{O(k)}\cdot\max\{\ltime,|\TG|\} = 2^{O(k)}\cdot|\TG|$.

	It remains to show that $I$ is a \yes-instance if and only if there is an $s$-$t$ path in~$G'$.
	We only show that if $I$ is a \yes-instance, then there is an $s$-$t$ path in $G'$ since the converse is easy to verify from the definition of $G'$.
	Let $I$ be a \yes-instance. 
	Then, by \cref{lem:q-robust-correct}, there is a solution $(P_1,\dots,P_\ltime)$ such that $P_i \in \mathcal F_i$, for all $i\in \set{\ltime}$.
	For each $i \in \set{\ltime-1}$, we have that $|V(P_i) \cap V(P_{i+1})| \leq \ell$.
	It follows that $G'$~has an edge from the vertex corresponding to $P_i$ 
	to the vertex corresponding to $P_{i+1}$.
	Hence, there is an $s$-$t$ path in $G'$ 
	because $s$ is adjacent to all vertices corresponding to a path in~$\mathcal F_1$ and
	each vertex corresponding to a path in~$\mathcal F_\tau$ is adjacent to $t$.
\end{proof}

\section{Looking through the lens of efficient data reduction}
\label{sec:kernel}
\appendixsection{sec:kernel}

In this section,
we study whether (polynomial) problem kernels for our four multistage $s$-$t$ path problems exist.
We start from the simple observation that every problem trivially admits a problem kernel of size polynomial in~$n+\tau$.
When strengthening~$n$ to~$\ug{\nu}$,
that is,
when parameterizing by~$\ug{\nu}+\tau$,
where~$\ug{\nu}$ denotes the vertex cover number of the underlying graph,
for~\eimsp{} and~\vimsp{} we prove a polynomial-size problem kernel (\cref{ssec:pkdis}) and
for~\emsp{} and~\vmsp{} we prove a single-exponential-size problem kernel (\cref{ssec:seksim}).
We prove that,
unless~\NPincoNPslashpoly,
the latter cannot be improved to polynomial size for~\vmsp{}
and that when parameterized by~$n$ (i.e., dropping~$\tau$ from~$n+\tau$)
none of the four problems admits a polynomial kernel (\cref{sec:incompr}).

\subsection{Polynomial kernel for the dissimilarity variant \texorpdfstring{regarding \boldmath$\ug{\nu}+\tau$}{}}
\label{ssec:pkdis}

In this section,
we prove \vimsp{} and \eimsp{} to admit problem kernels of polynomial size in~$\ug{\nu}+\tau$.

\begin{theorem}
 \label{thm:interseckernel-vc}
 Each of \vimsp{} and \eimsp{} admits a problem kernel with at most~$\tau\cdot(2\ug{\nu}+2 + \binom{2\ug{\nu}}{2}(3k-3)) \in O(\tau\ug{\nu}^3)$ vertices and~$\tau$ snapshots.
\end{theorem}

\noindent
The kernelization behind \cref{thm:interseckernel-vc} basically relies on the 
following data reduction rule.

\begin{rrule}
		\label{rr:vc-rule}
		Let $I=(\TGfull,s,t,k,\ell)$ be an instance of \vimsp{} or \eimsp{} with underlying graph~$\ug{\TG}$.
		\begin{compactenum}
				\item Compute a vertex cover $V'$ of~$\ug{\TG}$ of size at most $2\ug{\nu}$.
				\item For each pair of distinct vertices $v,w \in V'$ and each $i \in \set{\tau}$, 
						in $N_{vw}^i\ceq (N_{(V,E_i)}(v) \cap N_{(V,E_i)}(w))\setminus V'$ mark $\min\{3k-3,|N_{vw}^i|\}$ 
						vertices.
				\item Construct a set $V''$ containing $\{s,z\} \cup V'$ and all marked vertices, and then construct the temporal graph $\TG'=(V'',E'_1,\dots,E'_\tau)$,
						where $E'_i = \{ \{v,w\} \in E_i \mid v,w \in V''\}$, for all~$i \in \set{\tau}$.
				\item Output the instance $O=(\TG',s,t,k,\ell)$.
		\end{compactenum}
\end{rrule}

\noindent
First,
we prove that we can efficiently execute \cref{rr:vc-rule}.

\ifarxiv{}
\begin{lemma}
		\else{}
\begin{lemma}[\appref{lem:rtpk}]
\fi{}
    \label{lem:rtpk}
	\cref{rr:vc-rule} is correct and can be executed in $O(n\cdot \ug{\nu}^2)$ time.
\end{lemma}

\appendixproof{lem:rtpk}
{
\begin{proof}
 We can compute a 2-approximate vertex cover in linear time via a maximal matching (Step 1).
 Next,
 we compute for each of the at most~$\binom{2\ug{\nu}}{2}$ pairs of vertices in~$V'$, 
 in each of the~$\tau$ snapshots,
 their neighborhood and mark a subset therein in linear time.
 Finally,
 we can compute the set~$V''$,
 then~$\TG'$,
 and then~$O$ to output,
 each in linear time.
 Hence, this procedure ends after $O(n \cdot \ug{\nu}^2)$ time.

		Let $I=(\TGcompact,s,t,k,\ell)$ be an instance of \vimsp{} or \eimsp{},
		and let~$O=(\TG',s,t,k,\ell)$ be the output instance of~\cref{rr:vc-rule} on~$I$.
		Furthermore, 
		for all $i \in \set{\tau}$,
		let $G_i$ and $G'_i$ respectively denote the \ith{$i$} snapshot of $\TG$ and of $\TG'$.

		\LD{}
		Since each path in a snapshot of $\TG'$ is also a path in $\TG$, we have that if $O$ is a \yes-instance,
		then $I$ is a \yes-instance as well.

		\RD{}
		Now let $(P_1,\dots,P_\tau)$ be a solution for $I$.
		Clearly, if for each $i \in \set{\tau}$ we have that $P_i$ is a path in $G'_i$, then $(P_1,\dots,P_\tau)$ is also a solution for $O$.
		Let $\mathcal S_p$ be the set of solutions for $I$ such that $P_j$ is a path in $G'_j$, for all $j \in \set{p-1}$ and all $p \in \set{\tau}$.
		Note that if $\mathcal S_{\tau+1}$ is not empty, then $O$ is clearly a \yes-instance.
		Let $i = \max \{ p \in \set{\tau} \mid \mathcal S_p \not = \emptyset \}$ and
		let $S=(P_1,\dots,P_\tau) \in \mathcal S_i$, $P_i = (v_0,v_1,\dots, v_{k'})$, 
		$s=v_0$, 
		and $t=v_{k'}$
		such that $j$ is maximum under the condition that $v_0,\dots,v_{j-1}$ is a path in $G'_i$.
		We can conclude that $v_j$ is not a vertex in~$\TG'$.
		Let $V'' = V' \cup \{s,t\}$ where $V'$ is the vertex cover we computed during the execution of \cref{rr:vc-rule}.
		Hence, $v_j \not\in V''$ but $v_{j-1},v_{j+1} \in V''$, otherwise $V''$ is not a vertex cover.
		Let $N = (N_{(V,E'_i)}(v_{j-1}) \cap N_{(V,E'_i)}(v_{j+1}))\setminus V'$.
		From \cref{rr:vc-rule}, we know that $N$ is of size at most~$3k-3$.
		Now we distinguish  into four cases:
		\begin{compactenum}
				\item If $1=i=\tau$, then set $X = V(P_i)\setminus \{ s,t ,v_j\}$.
				\item If $1=i<\tau$, then set $X = (V(P_i)\cup V(P_{i+1}))\setminus \{ s,t ,v_j\}$.
				\item If $1<i=\tau$, then set $X = (V(P_{i-1})\cup V(P_i))\setminus \{ s,t ,v_j\}$.
				\item If $1<i<\tau$, then set $X = (V(P_{i-1})\cup(V(P_i)\cup V(P_{i+1})))\setminus \{ s,t ,v_j\}$. 
		\end{compactenum}
		Since all paths in~$S$ are of length at most $k$, we know that $X$ is of size at most $3k-4$.
		Hence, there is a vertex $w \in N \setminus X$ such that $P' = ( s=v_0,v_1,\dots, v_{j-1},w,v_{j+1}, \dots ,v_{k'}=t)$
		is an~$s$-$t$ path in~$G_i'$ of length~$k'\leq k$.
		Moreover, we note that
		\begin{itemize}
		 	\item if $i>1$, then $|V(P_{i-1}) \cap V(P')|\leq |V(P_{i-1}) \cap V(P_{i})|$ and
						$|E(P_{i-1}) \cap E(P')| \leq |E(P_{i-1}) \cap E(P_{i})|$;
			\item if $i<\tau$, then $|V(P') \cap V(P_{i+1})| \leq |V(P_{i}) \cap V(P_{i+1})|$ and
							$|E(P') \cap E(P_{i+1})| \leq |E(P_{i}) \cap E(P_{i+1})|$.
		\end{itemize}
		Hence, 
		in either case of $I$ and $O$ both being instances of \vimsp{} or \eimsp{},
		$(P_1,\dots,P_{i-1},P,P_{i+1},\dots,P_\tau)$ is a solution for~$O$.
\end{proof}
}

\begin{proof}[Proof of~\cref{thm:interseckernel-vc}]
 Given an instance~$I=(\TGfull,s,t,k,\ell)$,
 we apply \cref{rr:vc-rule} in polynomial time 
 to obtain the instance~$O=(\TG',s,t,k,\ell)$ being equivalent to~$I$ 
 (\cref{lem:rtpk}),
 containing~$\tau$ snapshots and at most~$\tau\cdot (2\ug{\nu}+2 + \binom{2\ug{\nu}}{2}(3k-3))$ vertices.
\end{proof}

\subsection{Single-exponential kernel for the similarity variant \texorpdfstring{regarding \boldmath$\ug{\nu}+\tau$}{}}
\label{ssec:seksim}

We prove that~\emsp{} and~\vmsp{} admit problem kernels of single-exponential size in~$\ug{\nu}+\tau$,
proving containment in~\FPT{}.
As we will see later,
unless~\NPincoNPslashpoly{} this result for \vmsp{} cannot be improved to size polynomial in~$\ug{\nu}+\tau$.

\begin{theorem}
 \label{thm:symdifexpkernelnutau}
 Each of \emsp{} and \vmsp{} admits a problem kernel with at most~$2\ug{\nu}+4^{\ug{\nu}\tau}(2\ug{\nu}+1)$ vertices and~$\tau$ snapshots.
\end{theorem}

\noindent
To prove~\cref{thm:symdifexpkernelnutau},
we lift the well-known graph-theoretic notion of (false) twins to temporal graphs as follows.

\begin{definition}
 \label{def:temptwin}
 Two vertices~$v,w$ in a temporal graph~$\TGfull$ are called \emph{(false) temporal twins} if~$N_{(V,E_i)}(v)=N_{(V,E_i)}(w)$ for every~$i\in\set{\tau}$.
\end{definition}

\noindent
Note that \cref{def:temptwin} implies an equivalence relation~$\sim$ on the vertex set~$V$, where~$v\sim w$ if and only if they are temporal twins,
and, hence,
a partition of the vertex set into classes of temporal twins.
Moreover,
every pair of vertices in the same temporal twin class is non-adjacent.
We show that such a partition is efficiently computable.

\ifarxiv{}
\begin{lemma}
		\else{}
\begin{lemma}[\appref{lem:computetemptwins}]
\fi{}
 \label{lem:computetemptwins}
 For a temporal graph~$\TGfull$,
 a partition~$V=(V_1,\dots,V_p)$ of~$V$ into temporal twin classes is computable in~$O(\tau\cdot |V|^2)$ time.
\end{lemma}

\appendixproof{lem:computetemptwins}
{
\begin{proof}
 Firstly,
 we compute all (false) twin classes in the first snapshot~$(V,E_1)$ in time linear in~$|V|+|E_1|$.
 Next,
 for each vertex~$v\in V$,
 check for each~$w$ with~$v\sim w$ whether~$w$ is a false twin in each snapshot~$(V,E_2),\dots,(V,E_\tau)$,
 and adjust~$\sim$ accordingly.
\end{proof}
}

In a nutshell,
given a vertex cover~$X$ of our underlying graph,
we aim for having few (i.e., upper-bounded by some single-exponential function in~$\ug{\nu}+\tau$) temporal twin classes in the independent set~$Y=V\setminus X$,
where each temporal twin class in turn contains only few vertices.
By definition we have only few temporal twin classes.

\ifarxiv{}
\begin{observation}
		\else{}
\begin{observation}[\appref{obs:numberofttwinclasses}]
\fi{}
 \label{obs:numberofttwinclasses}
 Let~$\TGfull$ be a temporal graph with partition~$V=(X,Y)$ of~$V$ such that~$Y$ is an independent set in each snapshot.
 Then the size of every partition of~$Y$ into temporal twin classes is at most~$2^{|X|\tau}$.
\end{observation}

\appendixproof{obs:numberofttwinclasses}
{
\begin{proof}
 There are at most~$2^{|X|}$ different neighborhoods for any vertex in~$Y$ per snapshot.
 As there are~$\tau$ snapshots,
 there are at most~$(2^{|X|})^\tau$ many temporal twin classes.
\end{proof}
}

\noindent
We next aim for shrinking temporal twin classes.
Note that for every temporal twin class,
any $s$-$t$ path contains at most the number of vertices neighboring the class minus one vertex from the temporal twin class:
recall that each temporal twin class forms an independent set,
and hence every $s$-$t$ path must ``alternate'' between the class and its neighboring vertices. 
In fact, 
temporal twin classes that are large compared to their neighborhood size can be shrunk. 

\begin{rrule}
 \label{rr:deletetemptwins}
 Let~$S$ be a temporal twin class with~$|S\setminus\{s,t\}|\geq \max_{1\leq i\leq \tau} |N_{(V,E_i)}(S)|$.
 Then delete a vertex~$v\in S\setminus\{s,t\}$.
\end{rrule}

\ifarxiv{}
\begin{lemma}
		\else{}
\begin{lemma}[\appref{lem:ttwinsrrcorrect}]
\fi{}
 \label{lem:ttwinsrrcorrect}
 \cref{rr:deletetemptwins} is correct and exhaustively applicable in $O(\tau \cdot |V|^3)$~time.
\end{lemma}

\appendixproof{lem:ttwinsrrcorrect}
{
\begin{proof}
		The reduction is clearly applicable in $O(\tau \cdot |V|^3)$ time.
 We prove its correctness.
 To this end, 
 let~$\TG$ and~$\TG'$ respectively denote the temporal graphs before and after application of~\cref{rr:deletetemptwins},
 and let~$S'\ceq S\setminus\{v,s,t\}$.
 Note that~$|S'|\geq\max_{1\leq i\leq \tau} |N_{(V,E_i)}(S')|-1$.
 Moreover,
 observe that due to~\cref{lem:computetemptwins}
 we can exhaustively apply \cref{rr:deletetemptwins} in polynomial time.
 We claim that~$I=(\TG,s,t,k,\ell)$ is a \yes-instance if and only if $I'=(\TG',s,t,k,\ell)$ is a \yes-instance.
 
 \LD{}
 As~$\TG'=\TG-v$,
 every sequence of~$s$-$t$ paths forming a solution for~$I'$ is also a solution to~$I$.
 
 \RD{}
 Let~$I$ be a \yes-instance, 
 and assume that every solution to~$I$ contains the vertex~$v$ (otherwise we are done).
 Let~$\calP=(P_1,\ldots,P_\tau)$ be a solution to~$I$ such that~$v$ appears latest in the sequence among all solutions.
 Let~$P_{r_1}$ be the first~$s$-$t$ path that contains~$v$,
 and let~$r_1,\ldots,r_p$ be a maximal sequence such that~$v\in V(P_{r_q})$ for each~$1\leq q\leq p$.
 Since~$|S\setminus\{s,t\}|\geq \max_{1\leq i\leq \tau} |N_{(V,E_i)}(S)|$ and~$S$ forms an independent set,
 there is a vertex~$w\in S'$ such that~$w\not\in V(P_r)$.
 We claim that ``replacing''~$v$ by~$w$ in~$P_{r_1},\dots,P_{r_p}$ forms a solution to~$I$ where~$v$ appears later than in~$\calP$,
 yielding a contradiction.
 Let~$r_s>r_1$ denote the smallest index such that~$w\in V(P_{r_{s}+1})$, or~$r_s=r_p$ if no such index exists.
 For all~$1\leq q\leq s$,
 let~$P_{r_q}'$ be the $s$-$t$ path with~$V(P_{r_q}')=(V(P_{r_q})\setminus\{v\})\cup\{w\}$ and~$E(P_{r_q}')=(E(P_{r_q})\setminus\{\{v,u\}\mid u\in N_{P_{r_q}}(v)\})\cup\{\{w,u\}\mid u\in N_{P_{r_q}}(v)\}$.
 For each~$i\in\set{\tau}\setminus\{r_1,\dots,r_s\}$,
 we set~$P_i'=P_i$.
 Observe that~$|V(P_{r_q}')|=|V(P_{r_q})|$ and~$|E(P_{r_q}')|=|E(P_{r_q})|$.
 Moreover,
 for all~$1\leq q<r_s$
 we have that~$|\symdif{V(P_{r_q}')}{V(P_{r_{q+1}}')}|=|\symdif{V(P_{r_q})}{V(P_{r_{q+1}})}|$ and~$|\symdif{E(P_{r_q}')}{E(P_{r_{q+1}}')}|=|\symdif{E(P_{r_q})}{E(P_{r_{q+1}})}|$.
 If~$r_1>1$,
 then it also holds true that~$|\symdif{V(P_{r_1-1}')}{V(P_{r_{1}}')}|=|\symdif{V(P_{r_1-1})}{V(P_{r_{1}})}|$ and~$|\symdif{E(P_{r_1-1}')}{E(P_{r_{1}}')}|=|\symdif{E(P_{r_1-1})}{E(P_{r_{1}})}|$.
 Finally,
 we consider the case of~$r_s<\tau$, 
 the cases herein whether or not~$w\in V(P_{r_{s}+1})$.
 
 \emph{Case 1: $w\not\in V(P_{r_{s}+1})$, $r_s\leq r_q$.}
 Then for the vertices we have that~$\symdif{V(P_{r_s}')}{V(P_{r_{s}+1}')}=((\symdif{V(P_{r_s})}{V(P_{r_{s}+1})})\setminus\{v\})\cup \{w\}$.
 For the edges,
 we have that 
 \begin{align*}
    \symdif{E(P_{r_s}')}{E(P_{r_{s}+1}')} &= ((\symdif{E(P_{r_s})}{E(P_{r_{s}+1})})\setminus\{\{v,u\}\mid u\in N_{P_{r_s+1}}(v)\}) \\ &\qquad \cup \{\{w,u\}\mid u\in N_{P_{r_s}'}(w)\}.
 \end{align*}

 \emph{Case 2: $w\in V(P_{r_{s}+1})$, $r_s< r_q$.}
 Then for the vertices we have that~$\symdif{V(P_{r_s}')}{V(P_{r_{s}+1}')}=((\symdif{V(P_{r_s})}{V(P_{r_{s}+1})})\setminus\{w\})\cup \{v\}$.
 For the edges,
 we have that 
 \begin{align*}
    \symdif{E(P_{r_s}')}{E(P_{r_{s}+1}')} &= (\symdif{E(P_{r_s})}{E(P_{r_{s}+1})}\setminus\{\{w,u\}\mid u\in N_{P_{r_s+1}}(w)\}) \\ &\qquad\cup \{\{v,u\}\mid u\in N_{P_{r_s+1}}(v)\}.
 \end{align*}
 
 \emph{Case 3: $w\in V(P_{r_{s}+1})$, $r_s= r_q$.}
 Then for the vertices we have that~$\symdif{V(P_{r_s}')}{V(P_{r_{s}+1}')}=(\symdif{V(P_{r_s})}{V(P_{r_{s}+1})})\setminus(\{w\}\cup \{v\})$.
 For the edges,
 we have that~
 \begin{align*}\symdif{E(P_{r_s}')}{E(P_{r_{s}+1}')}&=\big((\symdif{E(P_{r_s})}{E(P_{r_{s}+1})})  \setminus \\&\qquad(\{\{v,u\}\mid u\in N_{P_{r_s}}(v)\}\cup \{\{w,u\}\mid u\in N_{P_{r_s+1}}(w)\})\big) \\ &\qquad \cup (\{\{w,u\}\mid u\in N_{P_{r_s}'}(w)\}\cup \{\{w,u\}\mid u\in N_{P_{r_s+1}}(w)\}).
 \end{align*}
 
 Hence,
 in either case we have that the sizes of the symmetric differences both for vertex and edge sets are not increased.
 It follows that~$\calP'=(P_1',\dots,P_\tau')$  is a solution in which~$v$ appears later than in~$\calP$,
 contradicting the choice of~$\calP$.
\end{proof}
}

\begin{proof}[Proof of~\cref{thm:symdifexpkernelnutau}]
 First,
 in~$\ug{\TG}$ compute (via a maximal matching) a vertex cover~$X$ of size at most~$2\ug{\nu}$ in linear time.
 Let~$V=(X,Y)$,
 where~$Y=V\setminus X$ is an independent set.
 Next,
 compute all temporal twin classes of~$Y$ in polynomial time (\cref{lem:computetemptwins}).
 Apply~\cref{rr:deletetemptwins} exhaustively on every temporal twin class.
 Due to~\cref{lem:ttwinsrrcorrect},
 this returns an equivalent instance in polynomial time where every temporal twin class contains at most~$|X|+1$ vertices.
 Due to~\cref{obs:numberofttwinclasses},
 there are at most~$2^{|X|\tau}$ many temporal twin classes.
 In total,
 the obtained temporal graph contains at most~$|X|+2^{|X|\tau}(|X|+1)$ vertices and~$\tau$ snapshots.
\end{proof}

\subsection{Lower bounds on kernelization \texorpdfstring{regarding~\boldmath$n$ and~$\ug{\nu}+\tau$}{}}
\label{sec:incompr}
\appendixsection{sec:incompr}

We know that~relaxing~$n$ to~$\ug{\nu}$ in~$n+\tau$ allows for polynomial and single-exponential kernelization for dissimilarity and similarity, respectively.
We know that dropping~$n$ is not possible (\cref{prop:vmspvcwhard}).
In this section,
we prove that,
unless \NPincoNPslashpoly,
dropping~$\tau$ is not possible.

\ifarxiv{}
\begin{theorem}
		\else{}
\begin{theorem}[\appref{thm:empvincompr}]
\fi{}
 \label{thm:empvincompr}
 Unless~$\NPincoNPslashpoly$, 
 none of \emsp{}, \vmsp{}, \eimsp{}, and \vimsp{} admits a problem kernel of size polynomial in~$n$.
\end{theorem}

\toappendix
{
  For proving that kernels of polynomial size are unlikely to exists,
  we use the cross-composition framework of~Bodlaender~\etal~\cite{BodlaenderJK14}.
  The framework,
  like the original framework~\cite{BodlaenderDFH09,FortnowS11},
  bases upon the complexity-theoretic assumption that the polynomial time hierarchy does not collapse to its third level,
  which implies that~\NPnotincoNPslashpoly{}~\cite{Yap83}.
  The central notions of the framework are OR- and AND-cross-compositions,
  which require the notion of polynomial equivalence relations~\cite{BodlaenderJK14}:
  we call~$\calR$ a polynomial equivalence relation on~$\Sigma^*$ if we can decide in polynomial time whether any two~$x,y\in\Sigma^*$ are~$\calR$-equivalent,
  and the number of equivalence classes in any finite set~$S\subseteq \Sigma^*$ is in~$(\max_{x\in S}|x|)^{O(1)}$.

  \begin{definition}[\cite{BodlaenderJK14}]
  Given an \NP-hard problem $L \subseteq \Sigma^*$, 
  a parameterized problem $P\subseteq  \Sigma^* \times \N$, 
  and a polynomial equivalence relation~$\calR$ on the instances of L, 
  an OR-cross-composition of~$L$ into $P$ (with respect to~$\calR$) is an algorithm
  that takes $p$ $\calR$-equivalent
  instances $x_1 \ldots,x_p$ of $L$ and constructs in time~$(\sum_{i=1}^p |x_i|)^{O(1)}$ 
  an instance $(x, k)$ of~$P$ such that
  \begin{inparaenum}[(i)]
      \item $k\in (\max_{1\leq i\leq p} |x_i | + \tlog{p})^{O(1)}$ and
      \item $(x, k) \in P \iff x_i \in L$ for at least one $i\in\set{p}$.
  \end{inparaenum}
  An AND-cross-composition is an OR-cross-composition where (ii) is replaced by
  $(x, k) \in P \iff x_i \in L$ for all~$i\in\set{p}$.
  \end{definition}

  \noindent
  The connection is now the following:
  If a parameterized problem admits an OR-cross-composition (or AND-cross-composition)
  and a polynomial problem kernelization,
  then~$\NPincoNPslashpoly$ and the polynomial hierarchy collapses to its third level~\cite{BodlaenderJK14,Drucker15}.
}

\appendixproof{thm:empvincompr}
{
We call two instances~$I=(\TG,s,t,k,\ell), I'=(\TG',s',t',k',\ell')$ $\calR$-equivalent if~$|V(\TG)|=|V(\TG')|$,
$\tau(\TG)=\tau(\TG')$,
$k=k'$, 
and~$\ell=\ell'$.

\begin{proposition}
 \label{prop:emspcrocon}
 There is an algorithm that given~$p$ $\calR$-equivalent instances~$I_1,\dots,I_p$ of \emsp{},
 computes in polynomial time an instance~$I$ of \emsp{} such that~$n\in (\max_{1\leq q\leq p}|V_q|)^{O(1)}$ and 
 $I$ is a \yes-instance if and only if each of~$I_1,\dots,I_p$ is a \yes-instance.
\end{proposition}

\begin{construction*}
 \label{constr:emspcrocon}
 Let $I_1=(\TG_1=(V,E_1^1,\dots E_\tau^1),s_1,t_1,k,\ell),\dots,I_p=(\TG_1=(V,E_1^p,\dots E_\tau^p),\allowbreak s_p,t_p,k,\ell)$ be~$p$ $\calR$-equivalent instances of~\emsp{}.
 We construct an instance~$I=(\TG',s,t,k',\ell)$
 with~$\TG'=(V',E_1,\dots,E_{\tau'})$ and~$k'=k+2$ as follows.
 Let~$V'=\{s,t\}\cup V$ with two new distinct vertices $s$ and~$t$.
 Let~$E_{trans}=\{\{v,w\}\mid v,w\in V'\}$, 
 that is,
 $E_{trans}$ describes the edge set of a clique on~$V'$.
 Next,
 let~$\vmod{E}_r^q=E_r^q\cup \{\{s,s_q\},\{t,t_q\}\}$ for every~$r\in\set{\tau}$ and~$q\in\set{p}$.
 For~$1\leq q\leq p$ and~$1\leq j\leq \tau+k'$,
 we set~$E_{(q-1)(\tau+k')+j}=\vmod{E}_j^{q}$ if~$j\leq \tau$,
 and~$E_{(q-1)(\tau+k')+j}=E_{trans}$ if~$j> \tau$.
 This finishes the construction.
 Note that the construction is computable in polynomial time.
\end{construction*}

\begin{observation}
 \label{obs:tranferthroughclique}
 Let~$G$ be a clique with two distinct vertices~$s,t$,
 and let~$P,P'$ be two $s$-$t$ paths each with at most~$k\in \N$ vertices.
 Then there is a sequence~$(P=P_1,\ldots,P_k=P')$ of~$k$ $s$-$t$~paths each with at most~$k-1$ vertices, such that~$|\symdif{E(P_i)}{E(P_{i+1})}|\leq 4$ for all~$i\in\set{k-1}$ computable in polynomial time.
\end{observation}

\begin{proof}
 Let~$P=(s,v_1,\ldots,v_x,t)$ and~$P'=(s,v_1',\ldots,v_{x'}',t)$.
 We consider two cases:
 
 \emph{Case 1: $x\leq x'$.}
 Set~$P_i=(s,v_1',\dots,v_{i-1}',v_i,\dots,v_x,t)$ for every~$2\leq i\leq x$.
 Note that~$|\symdif{E(P_i)}{E(P_{i+1})}|\leq 4$ as we switch two vertices yielding four edges.
 If~$x=x'$, then~$P_x=P'$.
 Otherwise,
 for~$1\leq i\leq x'-x$,
 let~$P_{x+i}=(s,v_1',\dots,v_{x}',v_{x+1}',\dots,v_{x+i}',t)$.
 Note that~$|\symdif{E(P_{x+i})}{E(P_{x+i+1})}|\leq 4$ as we replace the edge~$\{v_{x+i}',t\}$ by the edges~$\{v_{x+i}',v_{x+i+1}'\}$ and $\{v_{x+i+1}',t\}$.
 
 \emph{Case 2: $x>x'$.}
 Set~$P_i=(s,v_1',\dots,v_{i-1}',v_i,\dots,v_x,t)$ for every~$2\leq i\leq x'$.
 Note that~$|\symdif{E(P_i)}{E(P_{i+1})}|\leq 4$ as we switch two vertices yielding four edges.
 For~$1\leq i\leq x-x'$,
 let~$P_{x+i}=(s,v_1',\dots,v_{x'}',v_{x+1},\dots,v_{x-i},t)$.
 Note that~$|\symdif{E(P_{x+i})}{E(P_{x+i+1})}|\leq 4$ as we replace the edges~$\{v_{x-i},v_{x-i-1}\}$ and $\{v_{x-i},t\}$ by the edge~$\{v_{x-i-1},t\}$.
 
 Finally, 
 if~$r=\max\{x,x'\}<k$, then
 pad the path~$P_r$~$k-r$ times (note that since the paths are identical, their symmetric difference is zero).
 The sequence is computable in polynomial time.
\end{proof}

\begin{proof}[Proof of~\cref{prop:emspcrocon}]
 Let $I_1=(\TG_1,s_1,t_1,k,\ell),\dots,I_p=(\TG_p,s_p,t_p,k,\ell)$ be~$p$ $\calR$-equiv\-a\-lent instances of~\emsp{} with $\TG_q=(V,E_1^q,\dots E_\tau^q)$ for every~$q\in\set{p}$ and~$\ell=4$,
 and let~$I=(\TG',s,t,k',\ell)$
 with~$\TG'=(V',E_1,\dots,E_{\tau'})$ and~$k'=k+2$ be the instance obtained from~$I_1,\dots,I_p$ using~\cref{constr:emspcrocon}.
 Note that~$|V(\TG')|=|V|+2$
 We claim that $I$~is a \yes-instance if and only if each of~$I_1,\dots,I_p$ is a \yes-instance.
 
 \RD{}
 Let~$(P_1,\ldots,P_{\tau'})$ be a solution to~$I$.
 For~$1\leq q\leq p$ and~$1\leq j\leq \tau$,
 we define $P_j^{q}=P_{(q-1)(\tau+k')+j}-\{s,t\}$ as the path obtained from~$P_{(q-1)(\tau+k')+j}$ when deleting $s$ and~$t$.
 with vertex set~$V(P_{(q-1)(\tau+k')+j})\setminus\{s,t\}$ and edge set~$E(P_{(q-1)(\tau+k')+j})\setminus\{\{s,s_q\},\{t,t_q\}\}$.
 We claim that for each~$1\leq q\leq p$,
 $(P_1^{q},\dots,P_\tau^{q})$ is a solution for~$I_q$.
 First note that for every~$j\in\set{\tau}$,
 $P_j^{q}$ is an~$s_q$-$t_q$ path in~$(V,E_j^q)$ and
 $|V(P_j^{q})|=|V(P_{(q-1)(\tau+k')+1})\setminus\{s,t\}|\leq k'-2=k$.
 Moreover,
 for every~$j\in\set{\tau-1}$,
 $|\symdif{E(P_j^{q})}{E(P_{j+1}^{q})}|= |\symdif{E(P_{(q-1)(\tau+k')+j})}{E(P_{(q-1)(\tau+k')+j+1})}|\leq \ell$ (recall that~$s$ is only adjacent with~$s_q$ and~$t$ is only adjacent with~$t_q$).
 Hence,
 the claim follows.
 
 \LD{}
 Let~$(P_1^q,\dots,P_\tau^q)$ be a solution for~$I_q$ for every~$q\in\set{p}$.
 For each~$q\in\set{p}$ and each~$i\in\set{\tau}$,
 let~$\vmod{P}_i^q$ be the path obtained from~$P_i^q$ with~$V(\vmod{P}_i^q)=V(P_i^q)\cup\{s,t\}$ and~$E(\vmod{P}_i^q)=E(P_i^q)\cup\{\{s,s_q\},\{t_q,t\}\}$.
 Note that~$\vmod{P}_i^q$ is an~$s$-$t$ path and~$|V(\vmod{P}_i^q)|=|V(P_i^q)|+2\leq k'$,
 and~$|\symdif{E(\vmod{P}_i^q)}{E(\vmod{P}_{i+1}^q)}|=|\symdif{E(P_i^q)}{E(P_{i+1}^q)}|\leq \ell$.
 Due to~\cref{obs:tranferthroughclique},
 for each~$q\in\set{p-1}$,
 we can compute for~$\vmod{P}_\tau^q$ and~$\vmod{P}_1^{q+1}$ a sequence~$(\vmod{P}_\tau^q=P^{q,q+1}_1,\dots,P^{q,q+1}_{k'}=\vmod{P}_1^{q+1})$ of~$k'$ $s$-$t$ paths such that
 each path has at most~$k'$ vertices and~$|\symdif{E(P^{q,q+1}_i)}{E(P^{q,q+1}_{i+1})}|\leq 4=\ell$ for all~$i\in\set{k'-1}$.
 Next we construct the path sequence~$\calP=(P_1,\dots,P_{\tau'})$.
 For each~$q\in\set{p}$,
 we set~$P_{(q-1)(\tau+k')+j}=\vmod{P}_j^q$ for~$1\leq j\leq \tau$,
 and we set~$P_{(q-1)(\tau+k')+\tau+j}=P^{q,q+1}_j$ for~$1\leq j\leq k'$.
 Clearly,
 $|\symdif{E(P_{(q-1)(\tau+k')+\tau})}{E(P_{(q-1)(\tau+k')+\tau+1})}|=|\symdif{E(P_{(q-1)(\tau+k')+\tau+k'})}{E(P_{q(\tau+k')+1})}|=0$ by construction for every~$q\in\set{p}$.
 It follows that
 for every~$i\in\set{\tau'}$, 
 $P_i$ is an~$s$-$t$ path with at most~$k'$ vertices,
 and for every~$i\in\set{\tau'-1}$,
 it holds true that
 $|\symdif{E(P_i)}{E(P_{i+1})}|\leq \ell$.
 Hence,
 $\calP$ is a solution to~$I$,
 and the claim follows.
\end{proof}

\begin{proposition}
 \label{prop:eimspcrocon}
 There is an algorithm that given~$p$ $\calR$-equivalent instances~$I_1,\dots,I_p$ of \eimsp{},
 computes in polynomial time an instance~$I$ of \eimsp{} such that~$n\in (|V_1|)^{O(1)}$ and 
 $I$ is a \yes-instance if and only if each of~$I_1,\dots,I_p$ is a \yes-instance.
\end{proposition}

\begin{construction*}
 \label{constr:eimspcrocon}
 Let $I_1=(\TG_1,s_1,t_1,k,\ell),\dots,I_p=(\TG_p,s_p,t_p,k,\ell)$ be~$p$ $\calR$-equivalent instances of~\eimsp{} with~$\TG_q=(V,E_1^q,\dots E_\tau^q)$ for all~$q\in\set{p}$ and~$\ell=0$.
 We construct an instance~$I=(\TG',s,t,k',\ell)$
 with~$\TG'=(V',E_1,\dots,E_{\tau'})$ and~$k'=k+2$.
 Let~$V'=\{s,t\}\cup V$ with two new distinct vertices~$s,t$.
 Let~$E_{trans}=\{\{s,t\}\}$, 
 that is,
 $E_{trans}$ only contains the edge~$s,t$.
 Next,
 let~$\vmod{E}_r^q=E_r^q\cup \{\{s,s_q\},\{t,t_q\}\}$ for every~$r\in\set{\tau}$ and~$q\in\set{p}$.
 For~$1\leq q\leq p$ and~$1\leq j\leq \tau+1$,
 we set~$E_{(q-1)(\tau+1)+j}=\vmod{E}_j^{q}$ if~$j\leq \tau$,
 and~$E_{(q-1)(\tau+1)+j}=E_{trans}$ if~$j=\tau+1$.
 This finishes the construction.
 Note that the construction runs in polynomial time.
\end{construction*}

\begin{proof}[Proof of \cref{prop:eimspcrocon}]
 Let $I_1=(\TG_1,s_1,t_1,k,\ell),\dots,I_p=(\TG_p,s_p,t_p,k,\ell)$ be~$p$ $\calR$-equiv\-a\-lent instances of~\eimsp{} with $\TG_q=(V,E_1^q,\dots E_\tau^q)$ for every~$q\in\set{p}$ and~$\ell=0$,
 and let~$I=(\TG',s,t,k',\ell)$
 with~$\TG'=(V',E_1,\dots,E_{\tau'})$ and~$k'=k+2$ be the instance obtained from~$I_1,\dots,I_p$ using~\cref{constr:eimspcrocon}.
 Note that~$|V(\TG')|=|V|+2$
 We claim that $I$~is \yes-instance if and only if each of~$I_1,\dots,I_p$ is a \yes-instance.
 
 \RD{}
 Let~$(P_1,\ldots,P_{\tau'})$ be a solution to~$I$.
 For~$1\leq q\leq p$ and~$1\leq j\leq \tau$,
 we define $P_j^{q}=P_{(q-1)(\tau+1)+j}-\{s,t\}$ as the path obtained from~$P_{(q-1)(\tau+1)+j}$ when deleting~$s$ and~$t$,
 which has vertex set~$V(P_{(q-1)(\tau+1)+j})\setminus\{s,t\}$ and edge set~$E(P_{(q-1)(\tau+1)+j})\setminus\{\{s,s_q\},\{t,t_q\}\}$.
 We claim that for each~$1\leq q\leq p$,
 $(P_1^{q},\dots,P_\tau^{q})$ is a solution for~$I_q$.
 First note that for every~$j\in\set{\tau}$,
 $P_j^{q}$ is an~$s_q$-$t_q$ path in~$(V,E_j^q)$ and
 $|V(P_j^{q})|=|V(P_{(q-1)(\tau+k')+1})\setminus\{s,t\}|\leq k'-2=k$.
 Moreover,
 for every~$j\in\set{\tau-1}$,
 $|E(P_j^{q})\cap E(P_{j+1}^{q})|= |E(P_{(q-1)(\tau+1)+j})\cap E(P_{(q-1)(\tau+1)+j+1})|\leq \ell$ (recall that~$s$ is only adjacent with~$s_q$ and~$t$ is only adjacent with~$t_q$).
 Hence,
 the claim follows.
 
 \LD{}
  Let~$(P_1^q,\dots,P_\tau^q)$ be a solution for~$I_q$ for every~$q\in\set{p}$.
 For each~$q\in\set{p}$ and each~$i\in\set{\tau}$,
 let~$\vmod{P}_i^q$ be the path obtained from~$P_i^q$ with~$V(\vmod{P}_i^q)=V(P_i^q)\cup\{s,t\}$ and~$E(\vmod{P}_i^q)=E(P_i^q)\cup\{\{s,s_q\},\{t_q,t\}\}$.
 Note that~$\vmod{P}_i^q$ is an~$s$-$t$ path and~$|V(\vmod{P}_i^q)|=|P_i^q|+2\leq k'$,
 and~$|E(\vmod{P}_i^q)\cap E(\vmod{P}_{i+1}^q)|=|E(P_i^q)\cap E(P_{i+1}^q)|\leq \ell$.
 Let~$P=(s,t)$ be the~$s$-$t$ path with vertex set~$V(P)=\{s,t\}$ and edge set~$E(P)=\{\{s,t\}\}$.
 Next we construct the path sequence~$\calP=(P_1,\dots,P_{\tau'})$.
 For each~$q\in\set{p}$,
 we set~$P_{(q-1)(\tau+1)+j}=\vmod{P}_j^q$ for~$1\leq j\leq \tau$,
 and we set~$P_{(q-1)(\tau+1)+\tau+1}=P$.
 Clearly,
 $|E(P_{(q-1)(\tau+1)+\tau})\cap E(P_{(q-1)(\tau+1)+\tau+1})|=|E(P_{(q-1)(\tau+1)+\tau+1})\cap E(P_{q(\tau+k')+1})|=0$ by construction for every~$q\in\set{p}$,
 since~$P$ is the only path using only the edge~$\{s,t\}$.
 It follows that
 for every~$i\in\set{\tau'}$, 
 $P_i$ is an~$s$-$t$ path with at most~$k'$ vertices,
 and for every~$i\in\set{\tau'-1}$,
 it holds true that
 $|E(P_i)\cap E(P_{i+1})|\leq \ell$.
 Hence,
 $\calP$~is a solution to~$I$,
 and the claim follows.
\end{proof}

\noindent
While~\cref{thm:empvincompr} is proven via an AND-cross-composition~\cite{BodlaenderJK14},
we prove that \vmsp{} admits no problem kernel of size polynomial in~$\tau+\ug{\nu}$ (unless~\NPincoNPslashpoly)
via an OR-cross-composition.
Recall that~$\ug{\nu}$ denotes the vertex cover number of the underlying graph,
and the result can be understood as that relaxing~$n$ in~$n+\tau$ does not allow for efficient preprocessing.
}

\noindent
We prove that,
unless~\NPincoNPslashpoly{},
improving the single-exponential kernel for \vmsp{} regarding~$\ug{\nu}+\tau$
to polynomial size is not possible.

\ifarxiv{}
\begin{theorem}
		\else{}
\begin{theorem}[\appref{thm:vmspnopkvctau}]
\fi{}
 \label{thm:vmspnopkvctau}
 Unless~$\NPincoNPslashpoly$,
 \vmsp{} admits no problem kernel of size polynomial in~$\ug{\nu}+\tau$.
\end{theorem}

\noindent
To prove~\cref{thm:vmspnopkvctau},
we 
OR-cross-compose~\cite{BodlaenderJK14} from the following \NP-complete
\cite{schaefer1978complexity} problem.

\problemdef{Positive 1-in-3 SAT}
{A set~$X$ of variables and a set~$\calC$ of clauses each containing three positive literals over~$X$.}
{Is there~$X'\subseteq X$ such that setting exactly the variables in~$X'$ to true results in each clause having exactly one variable set to true?}
We call two instances~$(X,\calC),(X',\calC')$ of \prob{Positive 1-in-3 SAT} $\calR$-equivalent if~$|X|=|X'|$ and~$|\calC|=|\calC'|$.
Note that~$\calR$ defines a \per~\cite{BodlaenderJK14}.
In particular, we show the following.
\ifarxiv{}
\begin{proposition}
		\else{}
\begin{proposition}[\appref{prop:orcrocoemps}]
\fi{}
 \label{prop:orcrocoemps}
 There is an algorithm that given a power~$p$ of two $\calR$-equivalent instances~$I_1=(X_1,\calC_1),\dots,I_p=(X_p,\calC_p)$ of \prob{Positive 1-in-3 SAT},
 computes in polynomial time an instance~$I$ of \vmsp{} such that~$k+\tau+\ug{\nu}\in (\max_{i\in \set{p}}{|X_i|+|\calC_i|}+\tlog{p})^{O(1)}$ and 
 $I$ is a \yes-instance if and only if at least one of~$I_1,\dots,I_p$ is a \yes-instance.
\end{proposition}
We use the following \cref{constr:orcrocoemps} to show \cref{prop:orcrocoemps}, see \cref{fig:orcrocoemps} for an illustration.
The basic idea of the construction is that 
the temporal graph has, among other vertices, a vertex set $D = \bigcup_{q=1}^p D^q$, 
where $D^q$ has one vertex for each variable in the \ith{$q$} input instance.
If we use a vertex from $D^q$ in the $s$-$t$ path, then we set the corresponding variable to true.
In the first $\log(p)$ snapshots,
we ensure that each $s$-$t$ path can only use vertices 
from $D$ which come from the same input instance.
The remainder of the snapshots ensures that the clauses are satisfied.
Here, the \ith{($\log(p) + r$)} snapshot ensures that the \ith{$r$} clause of some input instance is satisfied with exactly one variable (vertex).
Since we only use variables from one instance, \cref{prop:orcrocoemps} follows.
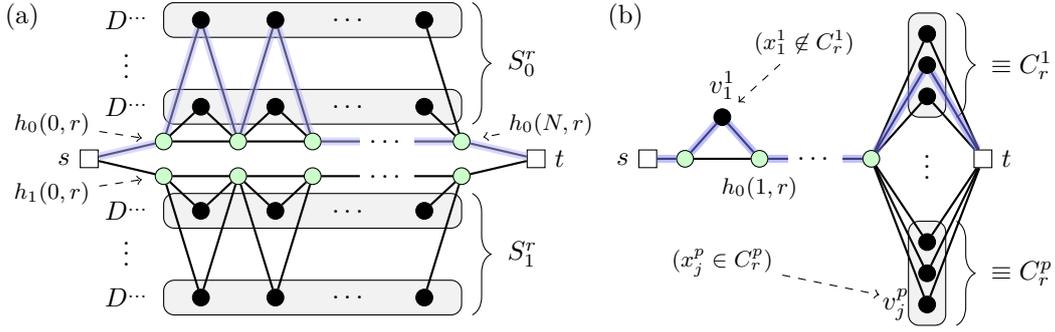
\begin{figure}
		\centering
		\begin{tikzpicture}[xscale=0.98,yscale=0.46]
    \tikzstyle{xnode}=[circle, draw,scale=2/3]
    \tikzstyle{xedge}=[thick]
	\begin{scope}[yshift=4cm,xshift=0.5cm]
			\draw[rounded corners, fill=gray!10] (0.5,-.5) -- (4.5,-.5) -- (4.5,.5) -- (0.5,.5) -- cycle;
			\node at (0,0) {$D^{\dots}$};
			\node at (1,0) (v11) [xnode,fill]{};
			\node at (2,0) (v21) [xnode,fill]{};
			\node at (3,0) {$\dots$};
			\node at (4,0) (v31) [xnode,fill]{};
			\node at (0,-1.1) {$\vdots$};
	\end{scope}
	\begin{scope}[yshift=1.5cm,xshift=0.5cm]
			\draw[rounded corners, fill=gray!10] (0.5,-.5) -- (4.5,-.5) -- (4.5,.5) -- (0.5,.5) -- cycle;
			\node at (0,0) {$D^{\dots}$};
			\node at (1,0) (v13) [xnode,fill]{};
			\node at (2,0) (v23)[xnode,fill]{};
			\node at (3,0) {$\dots$};
			\node at (4,0) (v33) [xnode,fill]{};
	\end{scope}
	\draw [decorate,decoration={brace,amplitude=7pt,mirror,raise=4pt},yshift=0pt]
			(5,1) -- (5,4.5) node [black,midway,xshift=0.8cm] {$S^r_0$}; 
	\draw [decorate,decoration={brace,amplitude=7pt,raise=4pt},yshift=0pt]
			(5,-1) -- (5,-4.5) node [black,midway,xshift=0.8cm] {$S^r_1$}; 
	\begin{scope}[yshift=0]
			\node at (-0.9,4.1) []{(a)};
			\node (s) at (0,0) (s) [draw,label=left:{$s$}]{};
			\node at (1,0.5) (a0) [xnode,fill=green!20]{};
			\node at (1,-0.5) (b0) [xnode,fill=green!20]{};
			\node at (2,0.5) (a1) [xnode,fill=green!20]{};
	\node at (2,-0.5) (b1) [xnode,fill=green!20]{};
	\node at (4,0.5) (d1) {$\dots$};
	\node at (4,-.5) (d2) {$\dots$};
	\node at (3,0.5) (a2)[xnode,fill=green!20]{};
	\node at (3,-0.5) (b2)[xnode,fill=green!20]{};
	\node at (5,0.5) (a3) [xnode,fill=green!20]{};
	\node at (5,-0.5) (b3) [xnode,fill=green!20]{};
	\node at (6,0) (z) [draw,label=right:{$t$}]{};
	\node at (-.5,1) (l1) {\footnotesize $h_0(0,r)$}; 
	\node at (-.5,-1) (l2) {\footnotesize $h_1(0,r)$}; 
	\node at (6.2,1) (l3) {\footnotesize $h_0(N,r)$}; 
	\draw[->,dashed,shorten >= 5pt] (l1) to (a0);
	\draw[->,dashed,shorten >= 5pt] (l2) to (b0);
	\draw[->,dashed,shorten >= 5pt] (l3) to (a3);
	\draw[xedge] (s) -- (a0) -- (a1) -- (a2) -- (d1) -- (a3) -- (z);
	\draw[xedge] (s) -- (b0) -- (b1) -- (b2) -- (d2) -- (b3) -- (z);
	\end{scope}

	\begin{scope}[yshift=-4cm,xshift=0.5cm]
			\draw[rounded corners, fill=gray!10] (0.5,-.5) -- (4.5,-.5) -- (4.5,.5) -- (0.5,.5) -- cycle;
			\node at (0,0) {$D^{\dots}$};
			\node at (1,0) (v12) [xnode,fill]{};
			\node at (2,0) (v22) [xnode,fill]{};
			\node at (3,0) {$\dots$};
			\node at (4,0) (v32) [xnode,fill]{};
			\node at (0,1.5) {$\vdots$};
	\end{scope}
	\begin{scope}[yshift=-1.5cm,xshift=0.5cm]
			\draw[rounded corners, fill=gray!10] (0.5,-.5) -- (4.5,-.5) -- (4.5,.5) -- (0.5,.5) -- cycle;
			\node at (0,0) {$D^{\dots}$};
			\node at (1,0) (v14) [xnode,fill]{};
			\node at (2,0) (v24) [xnode,fill]{};
			\node at (3,0) {$\dots$};
			\node at (4,0) (v34) [xnode,fill]{};
	\end{scope}
			\draw[xedge] (a0) -- (v13) -- (a1) -- (v23) -- (a2);
			\draw[xedge] (v33) -- (a3); 
			\draw[xedge] (a0) -- (v11) -- (a1) -- (v21) -- (a2);
			\draw[xedge] (v31) -- (a3); 

			\draw[xedge] (b0) -- (v12) -- (b1) -- (v22) -- (b2);
			\draw[xedge] (v32) -- (b3); 
			\draw[xedge] (b0) -- (v14) -- (b1) -- (v24) -- (b2);
			\draw[xedge] (v34) -- (b3); 

			\draw[-,line width=3pt,opacity=0.5,blue!30] (s) -- (a0) -- (v11) -- (a1) -- (v21) -- (a2) -- (d1) -- (a3) -- (z);
			\node at (7+0.2,4.1) []{(b)};
	\begin{scope}[xshift=7cm,yscale=1.2]
		\node (s) at (0.5,0) (s) [draw,label=left:{$s$}]{};
		\node at (1.5,1) (v) [xnode,fill,label={90:$v_1^1$}]{};
		\node at (1,0) (a0) [xnode,fill=green!20]{};
		\node at (2,0) (a1) [xnode,fill=green!20,label=below:{\footnotesize $h_0(1,r)$}]{};
		\node at (2.75,0) (d1) {$\dots$};
		\node at (3.5,0) (a2)[xnode,fill=green!20]{};
		\node at (4.25,0) {$\vdots$}{};
		\node at (5,0) (z) [draw,label=right:{$t$}]{};
		\draw[xedge] (a0) -- (v) -- (a1);

		\node at (2.6,2.8) (l3) {\footnotesize ($x^1_1 \not \in C^1_r$)};
		\node at (1.5,-2.5) (l4) {\footnotesize ($x^p_j \in C^p_r$)};
		\draw[->,dashed,shorten >= 5pt] (l3) to (v);
		\draw[xedge] (s) -- (a0) -- (a1) -- (d1) -- (a2); 
		\begin{scope}[xshift=4.25cm,yshift=1.5cm]
				\draw[rounded corners, fill=gray!10] (-.25,2) -- (.25,2) -- (.25,-.5) -- (-.25,-.5) -- cycle;
				\node at (0,0) (v1) [xnode,fill]{};
				\node at (0,0.75) (v2) [xnode,fill]{};
				\node at (0,1.5) (v3) [xnode,fill]{};
				\draw[xedge] (a2) -- (v1) -- (z); 
				\draw[xedge] (a2) -- (v2) -- (z); 
				\draw[xedge] (a2) -- (v3) -- (z); 
				\draw [decorate,decoration={brace,amplitude=7pt,mirror,raise=4pt},yshift=0pt]
				(0.25,-.5) -- (0.25,2) node [black,midway,xshift=1cm] {$\equiv C^1_r$}; 
		\end{scope}
		\draw[-,line width=3pt,opacity=0.5,blue!50] (s) -- (a0) -- (v) -- (a1) -- (d1) -- (a2) -- (v2) -- (z);
		\begin{scope}[xshift=4.25cm,yshift=-3.5cm]
				\draw[rounded corners, fill=gray!10] (-.25,2) -- (.25,2) -- (.25,-.5) -- (-.25,-.5) -- cycle;
				\node at (0,0) (v1) [xnode,fill,label={left:{$v_j^p$}}]{};
				\draw[->,dashed,shorten >= 15pt] (l4) to (v1);
				\node at (0,0.75) (v2) [xnode,fill]{};
				\node at (0,1.5) (v3) [xnode,fill]{};
				\draw[xedge] (a2) -- (v1) -- (z); 
				\draw[xedge] (a2) -- (v2) -- (z); 
				\draw[xedge] (a2) -- (v3) -- (z); 
				\draw [decorate,decoration={brace,amplitude=7pt,mirror,raise=4pt},yshift=0pt]
				(0.25,-.5) -- (0.25,2) node [black,midway,xshift=1cm] {$\equiv C^p_r$}; 
		\end{scope} 
	\end{scope}
	\end{tikzpicture}	
	\caption{Illustration of \cref{constr:orcrocoemps} with $p$ input instances. 
			(a) shows a snapshot $(V,E_r)$ with $r \leq \log(p)$.
	(b) shows a snapshot $(V,E_{\log(p)+r})$ for the \ith{$r$} clause of each input instance.
	Observe that the green (bright) vertices (including $s,t$) form a vertex cover of the underlying graph.}

	\label{fig:orcrocoemps}
\end{figure}

\begin{construction*}
 \label{constr:orcrocoemps}
 Let $I_1=(X_1,\calC_1),\dots,I_p=(X_p,\calC_p)$ be $p$, 
 where~$p$ is a power of two,
 $\calR$-equivalent instances of \prob{Positive 1-in-3 SAT}  
 where~$N=|X_i|$ and~$M=|\calC_i|$ for all $i \in \set{p}$.
 Let~$D^q=\{v_i^q\mid i\in\set{N}\}$ for all~$q\in\set{p}$,
 and~$D=\bigcup_{q\in\set{p}} D^q$.
 Let~$A=\{a_0^i,a_1^i\mid i\in\set[0]{N}\}$ and~$B=\{b_0^i,b_1^i\mid i\in\set[0]{N}\}$.
 Set~$V=\{s,t\}\cup D\cup A\cup B$.
 Define for each~$d\in\{0,1\}$ the auxiliary function
 \[h_d(i,r)\ceq \begin{cases} a_d^i,& \text{$r$ odd}\\ b_d^i,& \text{$r$ even}.\end{cases}\]
 We next describe the edge sets~$E_1,\ldots,E_{\tlog{p}}$ and~$E_{\tlog{p}+1},\ldots,E_{\tlog{p}+M}$.
 For edge set~$E_r$ with~$r\leq \tlog{p}$,
 let~$E_r$ contain the edges~$\{s,h_d(0,r)\},\{t,h_d(N,r)\}$ and the edge set~$\bigcup_{1\leq i\leq N}\{\{h_d(i-1,r),h_d(i,r)\}\}$ for each~$d\in\{0,1\}$.
 These sets form two~$s$-$t$ paths in~$(V,E_r)$.
 Finally,
 let~$S^r_0$ be the union of~$D^q$ with the~\ith{$r$} bit of the binary encoding of~$q-1$ being $0$,
 and~$S^r_1$ be the union of~$D^q$ with the~\ith{$r$} bit of the binary encoding of~$q-1$ being $1$.
 For~$v^q_i\in S^r_0$, add the edges~$\{h_0(i-1,r),v^q_i\}$ and~$\{h_0(i,r),v^q_i\}$.
 Similarly,
 for~$v^q_i\in S^r_1$, add the edges~$\{h_1(i-1,r),v^q_i\}$ and~$\{h_1(i,r),v^q_i\}$.
 For edge set~$E_{\tlog{p}+r}$ with $r\leq M$,
 let~$E_{\tlog{p}+r}$ contain the edge~$\{s,h_0(0,r)\}$ and the edge set~$\bigcup_{1\leq i\leq N}\{\{h_0(i-1,r),h_0(i,r)\}\}$.
 Consider the clauses~$C_r^1,\dots,C_r^p$.
 For each~$C_r^q$,
 if~$x_i^q\in C_r^q$,
 then add the edges~$\{h_0(N,r),v_i^q\},\{v_i^q,t\}$,
 and if~$x_i^q\not\in C_r^q$,
 then add the edges~$\{h_0(i-1,r),v_i^q\},\{h_0(i,r),v_i^q\}$.
 Set~$k=2N+3$ and~$\ell=2(N+1)$.
 This finishes the construction.
\end{construction*}

\appendixproof{prop:orcrocoemps}
{
\begin{observation}
 \label{lem:orcrocoempspropone}
 If~$(P_1,\ldots,P_{\tau})$ is a solution to~$I$ of \cref{constr:orcrocoemps},
 then for every~$r\in\set{\tau-1}$
 \begin{compactenum}[(i)]
  \item $|\symdif{V(P_r)}{V(P_{r+1})}|=\ell$, 
  \item $\symdif{V(P_r)}{V(P_{r+1})}\subseteq A\cup B$, and
  \item $V(P_r)\cap D=V(P_{r'})\cap D$ for all~$r'\in\set{\tau}$.
 \end{compactenum}
\end{observation}

\begin{proof}
 Let~$r\in\set{\tau-1}$.
 Note that in~$(V,E_r)$,
 $\{h_0(i,r),h_1(i,r)\}$ is an~$s$-$t$ separator for each~$0\leq i\leq N$.
 Hence,
 $P_r$ must contain for each~$0\leq i\leq N$ a vertex from~$\{h_0(i,r),h_1(i,r)\}$.
 The same holds for~$P_{r+1}$:
 $\{h_0(i,r+1),h_1(i,r+1)\}$ is an~$s$-$t$ separator for each~$0\leq i\leq N$,
 and hence
 $P_{r+1}$ must contain for each~$0\leq i\leq N$ a vertex from~$\{h_0(i,r+1),h_1(i,r+1)\}$.
 Since~$h_d(i,r)\neq h_{d'}(i',r+1)$ for all~$i,i'\in\set[0]{N}$ and~$d,d'\in\{0,1\}$,
 it follows that~$|\symdif{V(P_r)}{V(P_{r+1})}|\geq 2(N+1) = \ell$.
 Since~$(P_1,\ldots,P_{\tau})$ is a solution, 
 it also holds true that~$|\symdif{V(P_r)}{V(P_{r+1})}|\leq \ell$,
 and hence~$\symdif{V(P_r)}{V(P_{r+1})}\subseteq A\cup B$.
 This in turn implies that~$D\cap \symdif{V(P_r)}{V(P_{r+1})}=\emptyset$,
 and hence~$V(P_r)\cap D=V(P_{r'})\cap D$ for all~$r'\in\set{\tau}$.
\end{proof}

\begin{lemma}
 \label{lem:orcrocoempsproptwo}
 If~$(P_1,\ldots,P_{\tau})$ is a solution to~$I$ of \cref{constr:orcrocoemps},
 then for all~$r\in\set{\tau}$ 
 it holds true that~$\emptyset\neq V(P_r)\cap D\subseteq D^q$ for some~$q\in \set{p}$.
\end{lemma}

\begin{proof} 
		Observe that for each $r \in \set{M}$, we have that $D$ is an $s$-$t$ separator in the snapshot~$(V,E_{\tlog{p}+r})$,
 and hence every~$s$-$t$ path must contain a vertex from~$D$.
 Due to~\cref{lem:orcrocoempspropone},
 we know that~$V(P_r)\cap D=V(P_{r'})\cap D$ for all~$r,r'\in\set{\tau}$.
 Suppose that each path from~$P_1,\ldots,P_{\tau}$ contains a vertex~$v\in D^q$ and a vertex~$v'\in D^{q'}$ for~$q\neq q'$ in~$V(P_r)$.
 Let~$r\leq \tlog{p}$ be such that the~\ith{$r$} bit of~$q$ is~$d$ and of~$q'$ is~$1-d$ with~$d\in\{0,1\}$ (that is, where their~\ith{$r$} bits differ).
 Since for~$G_r=(V,E_r)$ it holds by construction that~$G_r-\{s,t\}$ contains two connected components,
 one containing the vertex set~$\bigcup_{i=0}^N h_{d}(i,r)$, 
 and the other containing the vertex set~$\bigcup_{i=0}^N h_{1-d}(i,r)$.
 Note that in~$G_r$,
 $v\in D^q$ is only connected to two vertices from~$\bigcup_{i=0}^N h_{d}(i,r)$,
 and~$v'\in D^{q'}$ is only  connected to two vertices from~$\bigcup_{i=0}^N h_{1-d}(i,r)$.
 Hence,
 $P_r-\{s,t\}$ contains vertices from two connected components,
 contradicting the fact that~$P_r$ is an~$s$-$t$~path in~$G_r$.
\end{proof}

\begin{proof}[Proof of~\cref{prop:orcrocoemps}]
 Let $I_1=(X_1,\calC_1),\dots,I_p=(X_p,\calC_p)$ be $p$, $p$ being a power of two, $\calR$-equivalent instances of \prob{Positive 1-in-3 SAT}  where~$N=|X|$ and~$M=|\calC|$.
 Let~$I=(\TGfull,s,t,k,\ell)$ be the instance obtained by \cref{constr:orcrocoemps} from~$I_1,\ldots,I_p$.
 Observe that $A \cup B \cup \{ s,t \}$ is a vertex cover of the underlying graph of $\TG$.
 Hence, we have that $k+\tau+\ug{\nu} \leq 2N+3 + \log(p)+M + N+4$.

 We claim that~$I$ is a \yes-instance if and only if at least one of~$I_1,\ldots,I_p$ is a \yes-instance.
 
 \LD{}
 Let~$X\subseteq X_q$ be a solution to~$I_q$, for some $q \in \set{p}$.

 We construct a solution~$(P_1,\dots,P_\tau)$ to~$I$ as follows.
 Set for each $r \in \set{\log(p)}$,
 \begin{align*}
  V(P_r) &= \bigcup_{x_i^q\in X} \{v_i^q\} \cup \{s,t\} \cup \bigcup_{0\leq i\leq N} h_d(i,r) \text{ } \\
  E(P_r) &= \{\{s,h_d(0,r)\}\}\cup \{\{t,h_d(N,r)\}\} \cup \bigcup_{x_i^q\in X} \{\{h_d(i-1,r),v_i^q\},\{h_d(i,r),v_i^q\}\}  \\ &\qquad \cup \bigcup_{x_i^q\in X_q\setminus X} \{\{h_d(i-1,r),h_d(i,r)\}\},\\ 
 \end{align*}
 where~$d=0$ if the~\ith{$r$} bit of~$q-1$ is $0$, and $1$ otherwise.
 Moreover, for each $r \in \set{M}$ set
 \begin{align*}
  V(P_{\tlog{p}+r}) & =\bigcup_{x_i^q\in X} \{v_i^q\} \cup \{s,t\} \cup \bigcup_{0\leq i\leq N} h_0(i,r), \\
  E(P_{\tlog{p}+r}) &=\{\{s,h_0(0,r)\},\{h_0(j-1,r),h_0(j,r)\},\{h_0(N,r),v_j^q\},\{t,v_j^q\}\}\\ 
					&\qquad \cup \bigcup_{x_i^q\in X \setminus \{x_j^q\}} \{\{h_0(i-1,r),v_i^q\},\{h_0(i,r),v_i^q\}\}  \\ 
					&\qquad \cup \bigcup_{x_i^q\in X_q\setminus X} \{\{h_0(i-1,r),h_0(i,r)\}\} ,\\
 \end{align*}
 where $x_j^q \in X^q \cap C^q_r$.

 First observe that~$|V(P_r)|\leq N+2+N+1$, for all $r \in \set{\tau}$.
 Second,
 observe that~$|\symdif{V(P_r)}{V(P_{r+1})}|= \ell$, for all $r \in \set{\tau-1}$.
 Finally, 
 we claim that~$P_r$ is an~$s$-$t$ path in $(V,E_r)$ for each $r \in \set{\tau}$.
 For~$P_r$ with~$r\leq \tlog{p}$,
 this follows by construction.
 Consider~$P_{\tlog{p}+r}$ with~$1\leq r\leq M$.
 Note that~$X$ contains exactly one~$x_j^q$ with~$x_j^q\in C_r^q$ ($j \in \set{N}$)
 and hence the subpath~$(h_0(N,r),v_i^q,t)$ of $P_{\tlog{p}+r}$ exists in $(V,E_{\log(p)+r})$.
 By construction the subpath of $P_{\tlog{p}+r}$ from $s$ to $h_0(N,r)$ also exists. 
 
 \RD{}
 Let~$(P_1,\dots,P_\tau)$ be a solution to~$I$.
 Due to~\cref{lem:orcrocoempsproptwo},
 we know that for all~$r\in\set{\tau}$ 
 it holds true that~$\emptyset\neq V(P_r)\cap D\subseteq D^q$ for some~$q\in \set{p}$.
 Let~$X=\{x_i^q \mid v_i^q \in V(P_1)\}$.
 We claim that~$X$ is a solution to~$I_q$,
 that is,
 for every clause~$C_r^q$ there is an~$x\in X$ with~$x\in C_r^q$.
 Consider the snapshot~$G_{\tlog{p}+r}=(V,E_{\tlog{p}+r})$.
 Since~$P_{\tlog{p}+r}$ is an~$s$-$t$~path in~$G_{\tlog{p}+r}$
 and $D$ is an $s$-$t$ separator in $G_{\tlog{p}+r}$,
 there is exactly one $v\in D$ such that subpath~$(h_0(N,r),v,t)$ is a subpath of $P_{\tlog{p}+r}$.
 We know that~$v\in D^q$,
 and hence there is an~$x\in X$ such that~$x\in C_r^q$.
\end{proof}
}

\noindent
\cref{prop:orcrocoemps} describes an OR-cross-composition from an \NP-hard problem to \vmsp{} parameterized by~$\ug{\nu}+\tau$,
and hence \cref{thm:vmspnopkvctau} follows~\cite{BodlaenderJK14}.
We leave open whether~\emsp{} allows for a problem kernel of size polynomial in~$\ug{\nu}+\tau$.

\section{Conclusion}

On the one extreme, 
our hardness results exploit that the temporal graph can change dramatically from one time step to another. 
On the other extreme, 
the NP-hard (and typically parameterized hard) \textsc{Length-Bounded Disjoint Path} problem~\cite{GT11} easily reduces to all four \textsc{MstP} variants with each snapshot having the same edge set. 
This leads to the natural question for further islands of computational 
tractability between these two extremes.
Moreover, for the similarity case, we leave open whether working with edge 
distances decisively differs from working with vertex distances.

The models we introduced (and future, more refined models based upon these) 
may find several applications
as they naturally capture time-dependent route-querying tasks.
Besides resolving questions we explicitly stated as open throughout the text, 
future work could address generalizing the ``consecutiveness'' property by requiring that also short sequences 
(as in the time-window model of temporal graphs~\cite{LFZ19,LVM18}) 
of consecutive paths are (pairwise) similar or dissimilar.
Furthermore, with introducing the ``dissimilarity view'' we entered new territory 
in the context of multistage problems; it seems natural to 
also study it for other problems beyond \textsc{$s$-$t$~Path}.
Finally, to analyze \textsc{$s$-$t$~Path} in the \emph{global multistage}\footnote{That is, 
the total sum over all differences between consecutive paths 
in the solution is upper-bounded.} %
setting is well-motivated as well~\cite{HeegerHKNRS19}.

\bibliography{msp-bib}

\ifarxiv{}\else{}
  \newpage
  \appendix
  \section*{Appendix}
  \appendixProofText
\fi{}

\end{document}